\documentclass[11pt,reqno]{amsart}
\usepackage{a4wide}
\usepackage{amssymb,amsmath,amsthm,enumerate,amsfonts}
\usepackage{newlfont}
\usepackage{dsfont}
\usepackage{amsbsy}
\usepackage[dvips]{graphicx}
\usepackage{verbatim} 
\usepackage{float}
\usepackage{color}

\def\CC{\mathrm \hbox{C\hstpace{-1.em}\raisebox{.47ex}{
         \textrm{\mbox{$\scriptscriptstyle |$}}}\hspace{+0.8 em} }}
\def\CC{\hskip.2em\raisebox{.22em}{\makebox[0pt][c]{$\scriptscriptstyle |$}}
\hskip-0.25em\mathrm{C}}

\def\hcboxcm#1#2{\hbox to #1{\hfill #2 \hfill}}

\newcommand{\n}{\noindent}
\newcommand{\erre}{\mathbb{R}}

\def\null{\hbox{}}
\let\eps\varepsilon

\def\tn1{\widetilde n_1}
\def\tn2{\widetilde n_2}
\def\tn{\widetilde n }

\let\ds\displaystyle

\let\la\langle  \let\ra\rangle

\def\be{\begin{equation}}
\def\ee{\end{equation}}
\def\bea{\begin{eqnarray}}
\def\eea{\end{eqnarray}}
\let\vect\overrightarrow
\def\bean{\begin{eqnarray*}}
\def\eean{\end{eqnarray*}}

\def\Tr{\mathrm{ Tr}}

\def\indic {{\mathds 1}} 

\def\mt{\mathcal}

\def\tpsi{\tilde \psi}

\let\ssc\scriptscriptstyle

\def\={\, = \, }

\def\sch{{Schr\"odinger}\,\,}

 \def\OO{\rm \hbox{O\kern-.34em\raise.47ex
         \hbox{$\scriptscriptstyle |$}\kern-.46em\raise.47ex
         \hbox{$\scriptscriptstyle |$}\kern+0.5 em }}
\def\RR{{\mathbb R} }

\def\S{{\cal S}}
\def\T{{\cal T}}

\def\L{{\mathcal L}}
\def\I{{\mathcal I}}
\def\T{{\mathcal T}}

\def\CC{\mathbb{C}}

\def\OO{\mathbb{O}}

\def\RR{\mathbb{R}}

\def\mt{\mathcal}
\def\tpsi{\widetilde \psi}

\newcommand{\f}{\frac}
\newcommand{\ep}{\varepsilon}

\def\Box{\leavevmode\vbox{\hrule
     \hbox{\vrule\kern4pt\vbox{\kern4pt}%
           \vrule}\hrule}}
\def\blackbox{\leavevmode\vrule height 5pt width 4pt depth 0pt\relax}
\catcode`@=11

\def\eqalign#1{\null\,\vcenter{\openup1\jot \m@th
   \ialign{\strut \hfil$\displaystyle{##}$ & $\displaystyle{{}##}$\hfil
      \crcr#1\crcr}}\,}
\def\eqalignrll#1{\null\,\vcenter{\openup1\jot \m@th
   \ialign{\strut \hfil$\displaystyle{##}$ & $\displaystyle{{}##}$\hfil
    & $\displaystyle{{}##}$\hfil
      \crcr#1\crcr}}\,}
\def\eqalignrcl#1{\null\,\vcenter{\openup1\jot \m@th
   \ialign{\strut \hfil$\displaystyle{##}$ &\hfil $\displaystyle{{}##}$\hfil
    & $\displaystyle{{}##}$\hfil
      \crcr#1\crcr}}\,}
\def\eqalignno#1{\displ@y \tabskip\@centering
  \halign to\displaywidth{\hfil$\@lign\displaystyle{##}$\tabskip\z@skip
    &$\@lign\displaystyle{{}##}$\hfil\tabskip\@centering
    &\llap{$\@lign##$}\tabskip\z@skip\crcr
    #1\crcr}}
\newcounter{appendix}
\newcounter{sectionz}
\def\appendix{\advance\c@appendix by 1
\def\thesectionz {\Alph{appendix}}
\def\thesection{\Alph{appendix}} 
   \ifnum\c@appendix=1 \setcounter{section}{-1} \fi
   \@startsection {section}{1}{\z@}{-3.5ex plus -1ex minus 
  -.2ex}{2.3ex plus .2ex}{\large\bf}}

\catcode`@=12
\newtheorem{lemme}{Lemma}[section]  

\newtheorem{theorem}[lemme]{Theorem}

\newtheorem{corollary}[lemme]{Corollary}

\newtheorem{definition}[lemme]{Definition}

\newtheorem{proposition}[lemme]{Proposition}

\newtheorem{remark}[lemme]{Remark} 

\def\deblem{\begin{lemme}\it }
\def\finlem{\end{lemme}}
\def\debthm{\begin{theorem}\it }
\def\finthm{\end{theorem}}
\def\debprop{\begin{proposition} \it}
\def\finprop{\end{proposition}}
\def\debcor{\begin{corollary}\it }
\def\fincor{\end{corollary}}
\def\debdef{\begin{definition}\it}
\def\findef{\end{definition}}
\def\debrem{\begin{remark}\em}
\def\finrem{\null\hfill\blackbox\end{remark}}
\newtheorem{thm}{Theorem}  
\newtheorem{lem}[thm]{Lemma}
\newtheorem{defi}[thm]{Definition}
\newtheorem{prop}[thm]{Proposition}

\title[Decoherence for a heavy particle interacting with a light particle]
{Decoherence for a heavy particle interacting with a light particle: new
  analysis and numerics}

\author[R. Adami, M. Hauray, C. Negulescu]{Riccardo Adami, Maxime Hauray,
Claudia Negulescu}

\address{Riccardo Adami, Dipartimento di Scienze Matematiche
  G.L. Lagrange, Politecnico di Torino, Torino,
Italia.}

\email{riccardo.adami@polito.it}

\address{Maxime Hauray, Universit\'e d'Aix-Marseille \& CNRS \& \'Ecole 
Centrale Marseille, LATP UMR 7353, F-13453
Marseille, France.}

\email{maxime.hauray@univ-amu.fr}

\address{Claudia Negulescu, Universit\'e de Toulouse \& CNRS, UPS, Institut de Math\'ematiques de Toulouse UMR 5219, F-31062 Toulouse, France.}

\email{claudia.negulescu@math.univ-toulouse.fr}

\date{\today}



\begin{document}

\maketitle

{\sc Abstract.} {\footnotesize
We study the dynamics of a
quantum heavy particle undergoing a repulsive interaction with a light
particle. The main motivation is the detailed description of the loss of
coherence induced on a quantum system  (in our model, the heavy particle) by 
the interaction with the environment (the light particle). 

The content of the paper is analytical and
numerical. 

Concerning the analytical contribution, we show that an approximate description 
of the dynamics of the heavy particle can be carried out in two steps: first 
comes the interaction, then the free evolution. 
In particular, all effects of the interaction can be embodied in the
action of a collision operator
that acts on the initial state of the heavy particle.
With respect to previous analytical results on the same topics, we
turn our focus 
from the M{\o}ller wave operator to the full
scattering operator, whose analysis proves to be simpler.

Concerning the numerical contribution, 
we exploit the previous analysis to construct an efficient numerical scheme
that turns the original, multi-scale, two-body problem into two
one-body problems
which can be solved separately. This leads to a considerable gain in
simulation time. We present and interpret some simulations carried out on specific
one-dimensional systems by using the new scheme.

According to simulations, decoherence is produced by
an interference-free bump which arises from the 
initial state of the heavy particle immediately after the collision.
 We support such a picture  by numerical evidence as well as
by an approximation theorem.
}


\bigskip

\maketitle

\section{Introduction} \label{SEC1}
In the present paper we describe, both through a theoretical analysis
and numerical simulations,
the following idealized experiment: a
quantum particle lies in a state given by the superposition of two
localized wave functions (``bumps''), initially separated and moving
towards each other. At 
a certain time, the
particle interacts with another particle that is
considerably lighter. As a consequence, the quantum interference
arising when the two bumps corresponding to the heavy
particle  eventually meet is damped. The damping of the interference  
is called 
\emph{decoherence}, and provides a description of the transition from the
quantum to the classical world
(\cite{giulini,Joos-Zeh,schlosshauer,petruccione,caldeira,hornberger1,hornberger2,hornberger3,omnes}).     
Despite the conceptual relevance of decoherence to the foundations of quantum
mechanics as well as in applications (e.g.\  in quantum computation) 
and, more generally, in the understanding of the classical picture of
the macroscopic world, a rigorous and exhaustive description 
of this phenomenon is still at its beginnings; nevertheless, in the last decade
many important steps have been accomplished 
(see e.g.\  \cite{spohn,dft,adami_1,adami_2,adami_3,clark,carlone}).
\\

According to the
  principles of quantum mechanics, the time evolution of the wave
  function 
$\psi_\ep (t,X,x)$ 
representing the two-body quantum system is given by
the Schr\"odinger equation
\begin{equation} \label{SCH2D}
\left\{
\begin{array}{l}
\ds i \partial_t \psi_\ep  \ =  \ - \f 1 {2M} \Delta_X \psi_\ep   
- \f 1 {2   \varepsilon M} \Delta_x \psi_\ep + \f {1 +  \varepsilon}
\varepsilon V (x - X)\,  \psi_\ep ,\\[3mm]
\ds \psi_\ep (0,X,x)=\psi_\ep^0 (X,x)\,,
\end{array}
\right.
\end{equation}
 where we used units in which $\hbar = 1$,
$M$ is the mass and $X$ is the spatial coordinate of the heavy
particle, while $\eps M$ is the mass and $x$ is the spatial coordinate
of the light one. So $\eps$ is the ratio between the mass of the light
particle and the mass of the heavy one, and we study the
regime $\eps \ll 1$, which we call the \emph{small mass ratio}
regime.

The interaction is described by the potential $\f {1 +  \varepsilon}
\varepsilon V$;
the uncommon
coupling constant is chosen to be of order $\eps^{-1}$ so that even a
single collision leaves an observable mark on
the heavy particle; furthermore, the factor $1 + \eps$ hardly affects
the dynamics and simplifies some expressions.
We shall always choose a {\em factorized} initial state,
  i.e. $\psi_\eps^0$ will be the product of functions depending only on the variable
  $X$ and  the variable $x$, respectively (see
  \eqref{due}). Physically, this means that initially
the two particles are
uncorrelated. We shall always assume that $\psi_\eps^0$, and
consequently $\psi_\eps (t)$, is normalized in $L^2(\RR^{2d})$. \\

The aim of the present paper is threefold: first, we rigorously 
derive a {\em collisional dynamics}
  for the heavy particle as an approximation of the underlying quantum evolution
  \eqref{SCH2D} in the limit
  $\eps \to 0$ (Sections 3 and 4); second, we
  employ such a collisional dynamics in order to build up an efficient numerical
  scheme (Sections 5.1 and 5.2); third, we observe the appearance of
  decoherence through numerical simulation (Section 5.3). Eventually,
  simulations show an unpredicted mechanism for the
  occurrence of decoherence, which we are able to derive rigorously
  (Sections 5.4 and 6.3).

The emergence of a collisional dynamics, well-known since \cite{Joos-Zeh}
and rigorously deduced already in
\cite{spohn,dft,adami_1,adami_2,adami_3,clark,carlone}, 
can be explained by the fact that
the characteristic evolution time is of order
one for the heavy particle and
of order $\ep$ for the light one, so
the light particle diffuses almost instantaneously, while, during the
interaction, the 
heavy particle hardly moves.
Thus, the main effect of
the interaction on the heavy particle
is the reduction of the quantum interference 
among the two bumps.  This, roughly 
speaking, is the
content of the celebrated Joos-Zeh's heuristic formula
(see e.g.\  formula (3.43) in \cite{Joos-Zeh}), which establishes that
 the state of the
heavy particle has hardly changed, while the state of the light particle
is transformed by the action of a
suitable scattering operator.

In order to give a mathematical description to this scenery, in
Section 3 we introduce a {\em collision operator} $\mt I_\chi$, 
 whose action consists in multiplying 
the kernel $\rho^M (X,X')$ of the density operator $\rho^M$ of the
heavy particle by the {\em collision function}
\[
I_\chi(X,X') := \la S^{X'} \chi | S^{X} \chi \ra,
\] 
where, following the physicist's habit, the Hermitian product $\la
\cdot | \cdot \ra$ is anti-linear in the first factor and linear in
the second.
Furthermore,
$S^X$ stays for the one-particle scattering operator constructed assuming
that the interaction potential is $V (\cdot - X)$.
Notice that $0 \leq |I_\chi| \leq 1$; we will show
in Section 4 that decoherence arises precisely when $ I_\chi$
is not identically equal to one.

Several novelties are present 
with respect to previous achievements on analogous problems
(see \cite{adami_3,adami_2,carlone,dft}). 
  First, we use a different initial state for the light particle in 
order to replace the  M{\o}ller wave operator
with the  scattering operator (Theorem \ref{thm:geneS}); this
choice makes it possible to provide explicit formulas for the
function $I_\chi$ (Section
3). Second, in the present work the convergence
of the two-body dynamics to the two separated one-body
dynamics is given in Theorem  \ref{thm:geneS}, but the convergence
rate in $\eps$ is not 
explicitly specified (see formula \eqref{estim:4C}).
However, given an interaction potential $V$, one can compute the
related scattering operator $S$ and then use formula \eqref{def:C1} in
order to find the convergence rate. 
Third, we generalize Theorem \ref{thm:geneS} to 
the case of density operators in
Theorem \ref{thm:geneRho}. Fourth, we give explicit formulas for the
momentum and energy transfers between the two particles.  
To this regard, 
we remark that, even though the incoming particle has a negligible
mass, the transfer of energy and momentum is not trivial, 
 since in the limit $\eps \to 0$ the colliding light particle has finite
momentum and infinite energy.

\medskip

Concerning the numerical part of the paper, we recall that in
\cite{Cla_Riccardo} 
the authors exhibited some numerical simulations aimed at checking the
Joos-Zeh's 
approximation formula 
from a quantitative point of view: indeed, the error in such 
formulae, as computed in \cite{dft,adami_3,adami_2},
contains a multiplicative constant whose optimal size is
unknown (see e.g.\  estimate (2.2) in \cite{adami_2}). The numerical
simulations in \cite{Cla_Riccardo} show  
that,
in spite of this indeterminacy, the approximation in \cite{adami_2,
  adami_3} can be successfully employed  at least 
under some hypotheses on the interaction potential
 (for details see Section 3.2 in \cite{Cla_Riccardo}).
Those numerical results were achieved by using
 standard discretization arguments and a splitting 
(Peaceman-Rachford) procedure.
The main drawback of such a method was
that, for fixed grids, the 
precision was sensitive to the value of $\eps$, 
so that, in order to follow the fast evolution of the light particle,
one had to employ tiny meshes both in time and space, 
and the computations became expensive
in time and memory.

Conversely, in the present approach
the role of the 
light particle is limited to the computation of the collision function
$I_\chi$.
The focus of the analysis is the dynamics of the
heavy particle that, in our approximation, becomes free after the collision.
In this way, the computational problem drastically simplifies
and the numerical cost is 
considerably diminished, both in memory and in time; moreover, it
becomes possible to simulate an experiment with
many colliding light particles, which is crucial for the sake of studying 
the continuous damping of the interference.

As already stressed, our simulations lead to a description of
decoherence that, at least to our knowledge, was never put in evidence before. 
Indeed, simulations show that, 
  if the light particle initially has a non-vanishing mean momentum,
  then after the 
  collision a fraction of the wave function of the heavy
  particle organizes itself into a bump that moves independently of
  the rest (see the first picture in Figure \ref{fringes1}). Moreover,
  the newborn bump is uncorrelated with the other components of the
  state of the heavy particle, so it does not take part in the
  interference. Thus, the damping of the interference pattern can be
  explained by the fact that a fraction of the initial wave function
  decouples from the rest. We give a rigorous result which portrays
  this phenomenon,
  if some hypotheses on the involved spatial scales are satisfied
  (Theorem \ref{prop:last_approx}).

For simplicity, our numerical treatment is limited to the one-dimensional case,
even though the general idea and the theoretical results apply to
systems in arbitrarily 
high dimension.

\medskip

This paper is more concerned with a precise
  estimate on the dynamics of the heavy particle, than on
  interpretation of decoherence in terms of the foundations of quantum
  mechanics; nevertheless, some words on the conceptual background are
  in order.

In \cite{Cla_Riccardo} we introduced and discussed an
interpretation of decoherence based on the analysis of the
configuration space of the system. 
According
to such an interpretation, the two bumps representing the 
state of the heavy particle can be plotted as two bumps that, in the
absence of the light particle, move one towards the other: 
the simulation shows that when the centres of the two bumps coincide,
the overlap between the two bumps is complete and then  interference
is  maximal (see the last picture of
Figure 9 in \cite{Cla_Riccardo}).

On the other hand, if the light particle is
present, then its position appears as an additional dimension in 
the configuration space of the two-particle system. Now, since both
bumps of the heavy particle undergo a collision with the light one,
they will be perturbed in two different ways, so that 
the eventual overlap will not be complete and thus the interference is 
only partial
(to this regard, see the second
picture in the first row of Figure 9 in \cite{Cla_Riccardo}).
For the details of this explanation of decoherence
we refer to Section 4 of \cite{Cla_Riccardo} and to Remark~\ref{Rk:entang} of the present article: what we would like to
stress here is the presence of a portion of the two-body wave function
that, in the (full) configuration space, is prevented from
overlapping and hence from producing interference. 

Such a description has the advantage of being clear and simple, both
from the physical point of view and for the mathematics involved: the
only mathematical object that is needed is the wave function. On the
other hand, in the present paper we aim at getting rid of the
coordinates of the light particle, reducing then the number of
variables to consider, which is often important for numerical
computations. The price we have to pay is that we lose the
enlightening picture in the coordinate space and we have to deal with
more complicated mathematical objects, like density operators. 

\noindent
The
description of the decoherence phenomenon that emerges from our analysis can be summarized as follows, according to formula
(\ref{eq:ApproxDensMat}): each bump of the heavy particle interacts
with the light particle only through the portion of its wave function corresponding to the
reflection coefficient of the interaction. Then,
the density operator of the heavy particle
after the interaction turns out to be a convex combination of three
density operators: the one corresponding to the
suitably damped initial state, that did not interact when the light particle was transmitted, 
the one corresponding to the 
left bump when the light particle was reflected from the left, and the one corresponding to the right bump when the light particle was reflected from the right.
Only the first one has preserved the capability to produce
interference, while the two others did not. The overall interference
is then damped due to the fact that the portion of the heavy particle
that underwent interaction cannot interfere any more. 
The correspondence with the non-overlapping components of the two-body wave
function, as displayed in the analysis in the configuration space, is explained briefly in Remark~\ref{Rk:entang}. Actually, we plan to investigate that correspondence in more details in a future work.

\medskip

The outline of the present paper is the following. 
In Section \ref{preliminaries} we introduce the mathematical framework
and fix some notation.  
 Section \ref{theorem} provides
the general approximation theorems, enabling one to replace the two-body
Schr\"odinger picture by a suitable
collisional dynamics.
Section
\ref{sec:Icalc} is devoted to the study of the 
collision function $I_\chi$: we give general formula, study some
properties and provide approximations for some particular choices
of the interaction potential $V$.  
In Section
\ref{SEC4} we describe the numerical method and present results obtained 
with that scheme. In particular, we carefully analyze
the decoherence effect carried on the heavy particle by
the interaction with the light one, showing the appearance, after the
collision, of an uncorrelated bump and explaining it theoretically. 
Finally, the last sections are
devoted to the rigorous mathematical proofs of the main theorems of
Sections \ref{theorem} and \ref{thexpla}.

\section{Preliminaries} \label{preliminaries}

Let us recall some elementary notions of quantum mechanics and fix some
notation.
The state of a pair of
quantum particles in space dimension $d$ can be represented
by a function $\psi$ in $L^2 (\mathbb{R}^{2d})$ called the ``wave
function'', whose norm equals one, to be
interpreted according to the well-known Born's rule: 
given $\Omega_1,
\Omega_2 \subseteq \mathbb{R}^d$, the quantity $\int_{\Omega_1} \int_{\Omega_2}  |\psi
(X, x) |^2\,dx\, dX $ is the joint probability of finding the heavy particle
in the domain $\Omega_1$ and the light one in $\Omega_2$ after a
measurement of their positions. 
In order to detect and measure the decoherence, one has to study the
probability density $\int_{\RR^d} | \psi (X,x)|^2\, dx$ of finding the
heavy particle at a point $X$,
averaged on
the position of the light one.

As already stated, we assume that the ratio of the masses of the two
particles is small, and  fix the mass of the heavy particle to $M=1$ for the analytical investigations, while the mass of the light particle will be denoted by $m = \varepsilon \ll 1$. In the numerical section (number 5) however we will choose different values for $M$ for scaling reasons.

We assume that the interaction between the two particles can be
modeled by a regular, rapidly decaying and 
non-negative  potential $V$, that depends on the distance between the two
particles only.  Due to the regime of small mass ratio, in
order to obtain a non-trivial evolution for the heavy particle, we 
suppose that the strength of the interaction is of the
order of the inverse of the mass ratio ($\varepsilon^{-1}$). More
precisely, as this choice simplifies the analysis and does not affect the results,
we define the coupling constant as
$\varepsilon^{-1} (\varepsilon + 1)$, which is the inverse of the
reduced mass of the two-body system.

Under such assumptions,
the time evolution of the two-body wave function $\psi_\ep (t,X,x)$ is
given by (\ref{SCH2D}), associated with the initial condition 
\begin{equation} \label{due}
\psi_\ep (0,X,x) \ = : \ \psi^0_\varepsilon (X,x) \ = 
\ \varphi (X) [
  U_0 (- \varepsilon^{- \gamma}) \chi] (x), 
\end{equation}
where $\varphi$ and $\chi$ are regular functions (see next section for
the exact regularity required) and $\gamma \in (0,1)$.  The  
presence of the  free propagator $ U_0 (- \varepsilon^{- \gamma})$ in the
definition of the initial state of the light particle  is useful to describe a
situation in which the light particle comes 
from infinity and reaches $x=0$ in a time of order $\eps^{1
  - \gamma}$. Furthermore, it makes it possible
to treat $\chi$ as an incoming state in the sense of the
scattering theory (see e.g.\  \cite{rs3}).

\subsection{Notation}

\mbox{}

\smallskip \n
\textbullet\,  For $p \in [1, \infty]$, the
norm in $L^p (\RR^d)$ or in $L^p (\RR^{2d})$
is denoted by $\| \cdot \|_p$: the context will always clarify the
domain we refer to.
 For the
norm in $H^s (\RR^d)$ for $s \in \RR$, we use the symbol $ \| \cdot \|_{H^s}$.
 
\smallskip \n
\textbullet\, We denote the free Hamiltonian operator in $L^2 (\RR^d)$ by
$$H_0 : = - \f 1 2 \Delta, \quad H_0:H^2(\RR^d) \subset L^2(\RR^d) \rightarrow L^2(\RR^d)\,,$$
which is self-adjoint in
$L^2 (\RR^d)$ and
generates the free Schr\"odinger propagator, denoted in the following by $U_0 (t)$. The family of such operators is a strongly
continuous unitary group (for more details, see e.g.\  \cite{rs1}, Ch. 8). At 
fixed
$t$, $U_0 (t)$ acts as the convolution with the integral kernel 
\[ 
U_0 (t, x) \ = \ \f 1 {(2 \pi i t)^{d/2}} e^{i \f { | x |^2}{2t}}, \qquad x
\in \RR^d.
\]

\smallskip \n
\textbullet\, Whenever a tensor product appears, the
first factor refers to the heavy particle or to its state, while the
second refers to the light particle or to its state.
The convention applies to operators and wave functions. 
Given a wave function $\chi_X$ for the light particle, where the
  coordinate $X$ of the heavy particle enters as a parameter, $\varphi \otimes \chi_X$ will denote the wave
function defined by 
\[
[\varphi \otimes \chi_{_X} ] (X,x):= \varphi(X) \chi_{_X}(x).
\]
Of course, this is an abuse of notation since 
$\chi_{_X}$ depends on $X$, but it will be useful and unambiguous in the 
sequel.

\smallskip \n
\textbullet\,The interaction between the light and the heavy particle is
described by the potential $V : \RR^d \longrightarrow \RR$. In 
Theorems \ref{thm:geneS} and \ref{thm:geneRho} the potential $V$ is required to
fulfill some
general hypotheses (see assumptions (H1)-(H3) and related
  comments). 
For the numerical analysis (See Section
\ref{sec:diffpot}) three different kinds of  $V$ are considered,
which share the features 
of being non negative and rapidly decreasing, but are different in terms of local
regularity.

\smallskip \n
\textbullet\,
We denote by $H_V$ the Hamiltonian
\[  
H_V : =-\frac 12 \Delta + V,
\]
where $V$ is the multiplication by $V(x)$.
 In all cases we consider, $H_V$ is self-adjoint, and $U_V(t)$ denotes the unitary group generated by $H_V$, i.e.
\[
U_V(t) : = e^{-i  H_V t}.
\]

\smallskip \n
\textbullet\,  We denote by $S$ the scattering
operator between the self-adjoint operators $H_0$ and $H_V$, i.e.
\begin{equation} \label{defs}
S_V := s-\lim_{t,t' \to +\infty}  S_V(t,t')\,,
\quad  \text{where} \quad   S_V(t,t') : =
U_0(-t') U_V(t+t') U_0(-t),
\end{equation}
 and the limit holds in the strong operator topology. In all
  cases we consider, $S_V$ is well-defined and unitary.

\smallskip \n
\textbullet\, 
Consider the self-adjoint Hamiltonian operator $H_V$, its unitary group
$U_V$ and the related scattering operator $S_V$. Then,              
the shifts by any $X \in \RR^d$, denoted respectively 
by $H_V^X$, $U^X_V$ and $S_V^X$, are also well-defined
and share the properties of the unshifted ones. More
explicitly, 
\begin{align*}
H_V^X  & : = - \f 1 2 \Delta + V ( \cdot - X),    \qquad  U_V^X(t) : = e^{-i  H_V^X t},\\
S^X_V &:= s-\lim_{t,t' \to +\infty}  S^X_V (t,t'),
\quad \text{where} \quad S^X_V (t,t') \ : = \ U_0(-t') U_V^X(t+t') U_0(-t).
\end{align*}
When no confusion is possible, we will forget the subscript $V$ and use
the shorthand notation $H,S,U$ and $H^X,U^X,S^X$.  

\smallskip \n
\textbullet\, The two-particle free Hamiltonian operator and the
Hamiltonian operator containing the interaction among the two
particles shall be denoted respectively by
\[
H_\ep^f:=  - \f 1 2 \Delta_X - \f 1 {2 \eps} \Delta_x, \qquad
 H_\ep : =  - \f 1 2 \Delta_X - \f 1 {2 \eps} \Delta_x + \f {1 + \ep} \ep
V (|X-x |).
\]
Both are unbounded  self-adjoint operators on $L^2 (\RR^{2d})$.  The associated
unitary groups will be represented respectively by $U_\ep^f(t)$ and
$U_\ep(t)$. 

\noindent
The unitary group generated by $H_{\eps}^{f}$ factorizes
as
\[
U_\eps^f (t) \ = \ U_0 (t) \otimes U_0 (t/\eps).
\]

\smallskip \n
\textbullet\, The Fourier transform of a function $\phi \in L^2
(\RR^d)$ is denoted by $\widehat \phi$ and is defined by
\be \label{def:Fourier}
\widehat \phi (k) \ : = \ ( 2 \pi )^{- \f d 2} \int_{\RR^d} e^{-i k \cdot x} \phi
(x) \, dx,
\ee
where $k \cdot x$ is the Euclidean scalar product in $\RR^d$ between
the vectors $k$ and $x$.

\smallskip \n
\textbullet\, Given a functional space $H^s (\RR^d)$ (possibly with
$s=0$), we define the {\em translation operator $\theta_X$} by
\[
\theta_X \phi (x) \ = \ \phi (x - X),
\]
for any $\phi \in H^s (\RR^d)$. It turns out that $\theta_X$ is a
unitary operator.

\smallskip \n
\textbullet\, The space of self-adjoint trace-class operators (see
\cite{rs1}) on $L^2 (\RR^d)$ or in $L^2 (\RR^{2d})$ is denoted by
${\mathcal L}^1$ and the norm of a generic 
  element $\rho$ in that space is given by 
\[
\| \rho \|_{{\mathcal L}^1} \ : = \ {\mbox{Tr}} | \rho |\,, \quad
\forall \rho  \in {\mathcal L}^1, 
\]
where $\mathrm{Tr}$ denotes the trace functional (see \cite{rs1}, Ch. VI). 
The subspace of the positive elements of ${\mathcal L}^1$ is
denoted by 
${\mathcal L}^1_+$, without
specifying whether the operator of interest acts on $L^2
(\RR^{2d})$ or on $L^2
(\RR^{d})$. Anyway, the context will always be unambiguous: if $\rho$
is the density operator of a single particle, then  $\| \rho
\|_{{\mathcal L}^1}$ denotes its trace class norm as an operator on $L^2
(\RR^{d})$. Conversely, if  $\rho$
is the density operator of a two-particle system, then  $\| \rho
\|_{{\mathcal L}^1}$ denotes its trace class norm as an operator on $L^2
(\RR^{2d})$.

\smallskip \n
\textbullet\, We shall make occasional use of the so-called Dirac's bra-ket
notation: for example, the state of the heavy particle will be denoted
by $| \varphi \rangle$, while the state of the light particle by 
$| \chi \rangle$. A scalar product between two states of the light
particle shall be denoted by $\langle \chi' | \chi \rangle$, while the
orthogonal projector along the span of $| \chi \rangle$ will be
represented by $| \chi \rangle\langle \chi |$.

\smallskip \n
\textbullet\, We will always assume that wave functions $\varphi$,
$\chi$ and density operators $\rho$  are normalized, i.e.
$$
\| \varphi \|_2 = \| \chi \|_2 = 1, \qquad  \rho \in \L^1_+ \quad \text{and} \; \| \rho \|_{\L^1} = 1.
$$

\subsection{Assumptions} 

\noindent
We introduce three hypotheses that
 we shall use in Theorems~\ref{thm:geneS} and~\ref{thm:geneRho}.
\begin{itemize}
\item[(H1)] The Hamiltonian $H_V$ is self-adjoint on $L^2(\RR^d)$, its point
spectrum is empty and zero-energy resonances are absent.
\item[(H2)]   Asymptotic completeness holds for the couple of self-adjoint
  operators $H_0$ and $H_V$, and the scattering operator $S_V$ is well-defined and unitary in $L^2 (\RR^d)$.
\item[(H3)] There exist $s \in
\RR$ and a constant $C_s >0$ such that 
\[
 \forall \chi \in  L^2(\RR^d), \quad  \| |x| \, S_V \chi \|_2 \leq  \|
 |x| \, \chi \|_2 + C_s \|
\chi \|_{H^s}.
\]
\end{itemize}

\medskip
Let us comment on these hypotheses. The first one,
(H1),  requires self-adjointness of the Hamiltonian operator, which
  provides well-posedness of the associated 
Schr\"odinger equation and unitarity of the propagator; bound states, as
well as zero-energy
resonances are to be avoided for the wave
operators to be well-defined.
The second hypothesis (H2) prescribes the
unitarity of the scattering operator.
The third one (H3) is less common, and is 
a regularity assumption on the scattering operator $S_V$. 
For $d=1$, (H3) can be replaced by the stronger assumption
\begin{itemize}
\item[(H3')] There exists an $s \in
\RR$ and a constant $C_s >0$ such that the reflection and
transmission amplitudes $r_k$ and $t_k$ 
(see Section \ref{sec:scatmat})
satisfy 
\be
| \partial_k t_k| + | \partial_k r_k| \leq C_s (1 + |k|^2)^{\frac s2}.
\ee
\end{itemize}
The fact that (H3') implies (H3) is proven in Lemma~\ref{lem:ScatMom1D}.

Roughly speaking, hypotheses (H1)-(H2)-(H3) are
fulfilled by non
negative, regular potentials that decay fast enough at infinity. 
In dimension one, (H1)-(H2)-(H3') are satisfied, among others, by
the repulsive Dirac's delta potential and
 potential barriers, for which
the transmission and reflection amplitudes are explicitly known. See
Section~\ref{sec:Icalc} for more details. 

\section{Analytical results} \label{theorem}

In this section we give the analytical results that provide an approximate solution
to the problem \eqref{SCH2D}, \eqref{due} in the regime $\eps \ll 1$. 
In Theorem \ref{thm:geneS} the case of a
pure state (i.e. a wave function) is considered, and we give an
approximate solution in which
the evolution of the heavy particle is decoupled from the evolution of
the light one,
provided that the initial state has been suitably modified. 
In Theorem  \ref{thm:geneRho} we generalize the result to the case of
a mixed state (i.e. a density operator), in which the problem~\eqref{SCH2D},
\eqref{due} is replaced by the operator differential equation 
\eqref{SCHRhoScat}. Theorem \ref{thm:geneRho} provides an approximate
density operator for the heavy particle whose dynamics is governed by
a free evolution problem with modified initial data. The modification of
the initial data is given by the action of the collision operator
${\mt I}_\chi$.

For the convenience of the reader,  proofs are postponed 
to Section~\ref{appa}.

Theorems~\ref{thm:geneS} and \ref{thm:geneRho} supply the theoretical
basis of the numerical method that will be introduced in Section \ref{SEC4}.

\begin{defi}
\label{def:Sepp}
Given $\ep > 0$, the operator
$\S_\ep$, acting 
on $L^2(\RR^{2d})$  
is the unique unitary extension of
\be \label{def:Sep}
 \S_\ep (\varphi \otimes \chi)  := \varphi \otimes \, \bigl
   [U_0(-\ep^{-\gamma}) \, S^X \chi \bigr] , \quad \forall
   \varphi,\chi \in L^2(\RR^d); 
\ee
Furthermore, the operator $\widehat{\S}$, 
acting on $L^2(\RR^{2d})$,   is the unique unitary extension of
\[ 
\widehat{\S} (\varphi \otimes \chi) \ = \ \varphi \otimes S^X \chi.
\]
\end{defi} 
Notice that, with our notation,
$\S_\ep =  \bigl[ {\mathbb I} \otimes  U_0 (-\ep^{-\gamma})  \bigr] \widehat{\S}$.

\begin{thm} \label{thm:geneS}
Assume that the potential $V$ is such that hypotheses
(H1)--(H3) are satisfied and denote by $s$ a real number for which $(H3)$ 
holds.
Choose $\varphi \in H^1 (\RR^d)$ such that $| X | \varphi \in H^1
(\RR^d)$, and $\chi \in  H^{s+1} (\RR^d)$ such that  $| x | 
\chi \in H^1
(\RR^d)$.

Let $\psi_\ep (t)$  denote the solution to
\eqref{SCH2D} with $M=1$ and the initial condition \eqref{due}; moreover let
$\psi^{a}_\eps(t)$ denote the solution to the free two-body Schr\"odinger
equation 
\be \label{SCHlimscat}
\begin{cases}
i \, \partial_t \psi_\eps^a=-  \frac12 \Delta_X \psi_\eps^a - \frac 1 {2\eps} 
\Delta_x \psi_\eps^a    = H_\ep^{f} \psi_\ep^a \\
\psi_\ep^a(0) =   
\varphi \otimes U_0(-\ep^{-\gamma}) \, S^X \chi  = \S_\ep ( \varphi \otimes \chi).
\end{cases}
\ee

Then, the following estimate holds
\be \label{estim:4C}
\| \psi_\eps (t) - \psi_\eps^a(t) \|_2 \leq  C_1\left( \frac{1+\ep} \ep \,t -
\ep^{-\gamma}, \ep^{-\gamma}  \right) +   C_2 \, \ep+ C_3 \, \ep^{1-\gamma},
\ee
where the constants are given by
\bea 
C_1(\tau,\tau')  &:= &\left\| \varphi  [ S(\tau,\tau') - S] \chi(\cdot -X) \right\|_2
\label{def:C1}  \\
\label{def:C2}
 C_2 &:=& 
2\sqrt 2  \bigl( \| \nabla \varphi\|_2  \| |x|    \chi \|_2 + 
 \| |X| \varphi \|_2 \| \nabla \chi \|_2 + 
 \| X \cdot \nabla \varphi \|_2 +  \| x \cdot   \nabla \chi \|_2 \bigr) \\
&& \nonumber \hspace{1cm} + \, \sqrt 2 C_s \bigl(
\|\nabla \varphi \|_2 \| \chi \|_{H^s} +  \|
\chi \|_{H^{s+1}} \bigr) \\
\label{def:C3}
 C_3 &:=&  2 \sqrt2 \bigl(
\| \nabla \varphi\|_2  \|\nabla \chi\|_2 
+ 2 \| \Delta \chi \|_2 \bigr),
\eea
with $s $ and $C_s$ defined by the hypothesis (H3).
\end{thm}
For the proof see Section \ref{appa}.

\begin{remark} {\em
\mbox{}

 {\it i)} 
The first term in \eqref{estim:4C} is quite implicit, 
nevertheless hypotheses
 (H1)-(H2) guarantee
$$
\lim_{\tau,\tau' \rightarrow + \infty} C_1(\tau,\tau') = 0.
$$
Indeed, the existence of the strong limit that defines the scattering
operator (see \eqref{defs}) implies that,
for fixed $X$ and $\chi$, $\bigl\| S^X(\tau,\tau') \chi - S^X \chi
\bigr\|_2 \rightarrow 0$ as $\tau,\tau' \rightarrow +\infty$,  and 
therefore, observing that
$$C_1(\tau,\tau')^2 \leq |\varphi(X) |^2\, \bigl\| S^X(\tau,\tau')
\chi - S^X \chi \bigr\|_2^2 \le 4 
|\varphi(X) |^2,$$
 by dominated convergence one has
$C_1(\tau,\tau') \to 0$ as $\tau, \tau' \to \infty$. 

Notice that in order to explicitly estimate $C_1(\tau,\tau')$, 
one needs to study the one-body scattering of the
light particle by the potential $V$. 
See Proposition~\ref{prop:ScatDirac} for an example.

\medskip
{\it ii)} The constant $C_2$ in Theorem~\ref{thm:geneS} depends
on the regularity properties of the scattering operator through 
the assumption (H3). If this hypothesis is not satisfied, 
then one can prove that the constant $C_2$ may be replaced by
\[
C_2' :=\Bigl\|   \, |x | \left[  \nabla \varphi  \otimes  S^X \, \chi
+ \varphi \otimes  S^X  \, \nabla
\chi   \right]  \Bigr\|_2
+  5 \left\|   \, |x-X | | \nabla \psi^0|  \right\|_2.
\]
}
\end{remark}

\begin{remark} \label{cor:Dirac1D}
{\em Matching Theorem \ref{thm:geneS} with Proposition
~\ref{prop:ScatDirac}, one has that
for the one-dimensional system  with $V = \alpha \delta_0$,
$\alpha >0$, 
the solution $\psi_\ep$ to~\eqref{SCH2D} is well-approximated
by the solution $\psi^a_\ep$ to~\eqref{SCHlimscat}. More
precisely, for 
any initial condition of the type treated in Theorem \ref{thm:geneS}, there exists a
constant $C$ depending on $\varphi$ and $\psi$ such that 
\[
\forall \, t \ge 2\,  \ep^{1-\gamma}, \quad \| \psi_\eps (t) - \psi_\eps^a(t) \|_2
\leq C \,\Bigl[  \Bigl(\frac \ep t\Bigr) ^{\frac14} +    \ep^{1-\gamma} \Bigr].
\]
}
\end{remark}

\begin{remark} 
{\em There are some differences with respect to the previously 
known results \cite{adami_3, adami_2}.
First, we modified the initial state for the light particle 
by inserting the operator $U(-\eps^{-\gamma})$. Physically, this
means that in our idealized experiment the light particle enters
the system at time $t = - \infty$ and immediately becomes
entangled with the heavy one. On the other hand, in the
physical situation depicted in \cite{adami_2, adami_3, carlone,
dft} each light particle is injected in the system at time
zero. The mathematical consequence of our choice is that
the initial state of the light particle is 
(approximately) transformed via the action of the scattering operator instead of
the M{\o}ller wave operator. This is consistent with the original 
Joos-Zeh's formula (\cite{Joos-Zeh}).

\n
The main advantage of our choice is that, in general, the operator $S$ is rather simple to write in
Fourier variables as it involves the Fourier transform only, while the 
M{\o}ller operator involves a different (and usually implicit) eigenfunction expansion.
As a consequence, the scattering operator is better suited for a
direct analytical study and for numerical
simulations too.
}
\end{remark}

\begin{remark} {\em
Theorem~\ref{thm:geneS} can be formally restated as follows:
\be \label{thmSrough}
U_\ep(t) \, ( {\mathbb I}  \otimes U_0 (-\ep^{-\gamma})) \  \approx \ U^{f}_\ep(t) \,
\S_\ep
= \ U^{f}_\ep(t) \, (  {\mathbb I} \otimes  U_0 (-\ep^{-\gamma}) )\, \widehat{\S}
\ee
for times of order one. 
Pictorially, (\ref{thmSrough}) states that for small $\ep$
the light particle is instantaneously scattered away  by the
heavy one, which may be considered as fixed during the
interaction. 
}
\end{remark}

Let us generalize Theorem~\ref{thm:geneS} to  
the formalism of density operators. Such a step is necessary in order to
describe the dynamics of the heavy particle
 when interacting with several light particles:
indeed,
as we can see from 
\eqref{SCHlimscat},
the initial condition for the limit
model is not factorized, so after one
collision the heavy particle lies in a mixed state that has to be described by the
appropriate
density operator.

Assume that the initial state of the heavy particle is given by the
density operator $\rho^{\ssc
M}(0) \in \L^1_+$, while, as before, the light particle
at time zero lies in the state represented by the wave function $
U_0(-\ep^{-\gamma})  \chi$. 
Then, the density operator $\rho_\ep (t)$ that represents the state of the
two-body system at time $t$ solves the operator differential equation
\bea \label{SCHRhoScat}
\begin{cases}
i \partial_t \rho_\ep (t) & =  [ H_\ep , \rho_\ep (t) ] 
\\ 
\rho_\ep(0) &:= \rho^{M}(0) \otimes \, | U_0(-\ep^{-\gamma})  \chi \ra \la 
 U_0(-\ep^{-\gamma}) \chi | ,
\end{cases}
\eea
where the symbol $[A_1, A_2]$ denotes the commutator of the operators
$A_1$ and $A_2$.

For the sake of studying the dynamics of the heavy particle,
the interesting quantity is the density operator of the heavy
particle, which is denoted by $\rho^M_\ep (t)$ and defined as
\be
\rho^M_\ep (t) \ : = \ \Tr_m \rho_\ep (t)  = \ \sum_j \langle
\chi_j | \rho_\ep (t) | \chi_j \rangle,
\ee
where $\{ \chi_j \}_{j \in {\mathbb N}}$ is a complete orthonormal set for
the space $L^2 (\RR^d)$, and $\Tr_m$ denotes the so-called {\em partial trace 
w.r.t.\  the light particle}.

Let us be more precise on how to compute such a partial trace.
As $\rho_\ep (t)$ is compact, it
can be represented as an integral operator whose kernel can be
denoted, with a slight abuse of notation, by 
$ \rho_\ep (t, X, X', x, x')$. The integral kernel of the reduced density matrix for the
heavy particle then reads
\be 
\rho^M_\ep (t,X,X') \ : = \ \int_{\RR^d}  \rho_\ep (t, X, X', x, x) \, dx.
\ee
There does not exist a closed equation for the time evolution of
  $\rho^M_\ep $, but, as we shall see, as $\ep$ goes to
zero and for any $t \neq 0$, the operator $\rho^M_\ep (t)$ converges to
an operator $\rho^{M,a} (t)$ that satisfies a closed equation.  In
order to state this result  properly, we need to introduce a further operator on
$\L_1$ which we call the \emph{collision operator}. 

\begin{defi}[Collision operator] 
Suppose that the hypotheses (H1)-(H2) are satisfied. Then,
we define the {\em collision operator} 
\be \label{def:OperatorI}
\I_\chi :  \L^1(\RR^d)  \to \L^1(\RR^d)\,, \quad 
  \rho^M  \mapsto  \Tr_m \bigl( \rho^M \otimes |S^X \chi \rangle 
\langle S^{X'} \chi | \bigr).
\ee
\end{defi}

\begin{remark} {\em
It can be verified that the operator $\I_\chi$ is  well-defined and completely positive (in particular it preserves positivity). 
Moreover, it satisfies the estimate
\be \label{stimatraccia}
 \Tr \, | \I_\chi \rho^M  | \leq \Tr | \rho^M| \quad \text{with equality if }
\rho^M \in
\L^1_+.
\ee}
\end{remark}

\begin{remark} {\em
In terms of integral kernels, the action of the collision operator reads
\be \label{eq:KerI}
[\I_\chi \rho](X,X') = \rho(X,X')\, I_\chi(X,X'),
\ee 
where the  {\em{collision function}} $I_\chi$ is defined by
\be  \label{def:DecoFunction}
 I_\chi (X,X') \ : = \ \langle S^{X'} \chi | S^X \chi \rangle, \qquad X, X'
 \in \RR^d .
\ee}
\end{remark}

\noindent
Notice that the function $I_\chi$ reaches its maximum modulus at $X=
X'$, where it equals one.  
 
\begin{thm} \label{thm:geneRho}
Assume that the potential $V$ is s.t.\  the hypotheses (H1)--(H3) are
satisfied, 
choose $\rho^M (0) \in {\mt L}^1_+$ s.t.\  $\nabla  \rho^M
(0)\nabla $ and $| \cdot | \nabla \rho^M
(0) \nabla | \cdot | \in {\mt L}^1_+ $; choose $\chi \in H^s (\RR^d)$
for some $s \geq 1$.
Denote by $\rho_\eps (t)$ the solution to equation \ref{SCHRhoScat} and by
$\rho^{M,a} (t)$
the unique solution to the problem
\be \label{SCHRhoScatLim-rid}
\begin{cases}
i \partial_t \rho^{M,a} (t) & =  [ H_0, \rho^{M,a} (t) ]  \\ 
\rho^{M,a}(0) &:= {\mt I}_\chi \rho^{ \ssc M}(0).
\end{cases} 
\ee
Then, the following estimate holds
\be \label{estim:4CRho-rid}
 \| \rho^M_\ep(t) - \rho^{M,a} (t) \|_{\L^1} \leq  \widetilde 
C_1\left( \frac{1+\ep} \ep \,t -
\ep^{-\gamma}, \ep^{-\gamma}  \right) +   \widetilde  C_2 \, \ep+
\widetilde  C_3 \, \ep^{1-\gamma}  ,
\ee
where the constants are given by
\bean
\widetilde C_1(\tau,\tau')  &:= & 2 \bigl\| \rho^M(0)  |
           [S(\tau,\tau') - S] \chi(\cdot -X') \rangle  
\langle  [S(\tau,\tau') - S] \chi(\cdot -X) | \bigr\|_{\L^1}^{\frac12}
\\
\widetilde  C_2 &:=& 
4\sqrt 2  \bigl( \| \nabla \rho^M(0) \nabla \|_{\L^1}^{\frac12}  \| |x|    \chi \|_2 + 
 \| |X| \rho^M(0) |X| \|_{\L^1}^{\frac12} \| \nabla \chi \|_2  \\  && \hspace{1cm}
+  \bigl\| |X|   \nabla \rho^M(0)  \nabla |X| \bigr\|_{\L^1}^{\frac12} +  \| |x|  \nabla \chi \|_2 \bigr) \\
&& \hspace{1cm} + \,2 \sqrt 2 C_s \bigl(
\|  \nabla \rho^M(0) \nabla \|_{\L^1}^{\frac12} \| \chi \|_{H^s} +  \|
\chi \|_{H^{s+1}} \bigr) \\
\widetilde  C_3 &:=&  4 \sqrt2 \bigl(
\| \nabla \rho^M(0) \nabla \|_{\L^1}^{\frac12}  \|\nabla \chi\|_2 
+ 2 \| \Delta \chi \|_2 \bigr).
\eean
\end{thm}
The proof is given in Section \ref{appa}.

The last step in our theoretical framework consists in the possibility
of extending the previous procedure to the case of many light
particles to be injected in the system one after another.
To this purpose, one should use an approximation result analogous to
Theorem~\ref{thm:geneRho}, but adapted to a multi-particle system with light
particles arriving at different times. Instead of following this approach, which is
fully rigorous but
cumbersome and very difficult to handle (See for instance \cite{adami_2} for a
result with many simultaneous ``collisions''), we will 
  repeatedly
use  the approximation given 
by Theorem~\ref{thm:geneRho}. This means that we treat the heavy particle as if it were interacting with only one light particle at a time.

Under that approximation, the multiple use of the collision operator ${\mt 
I}_\chi$
is justified, provided 
that the constants $\widetilde C_1$,
$\widetilde  C_2$ and $\widetilde  C_3$ appearing in Theorem~\ref{thm:geneRho}
do not explode  when computed for $\rho^{M,a}(t)$ instead of
$\rho^M(0)$. The behavior of such constants can be shown to 
depend on the regularity properties of the collision function
$I_\chi$ only. In particular, the calculation of the kinetic energy of
$\rho^{M,a}(0)$ in terms of $\rho^M(0)$ and $I_\chi$, done in
Proposition~\ref{prop:Exchange}, may guarantee the correct
  behavior of the constants $\widetilde C_1$,
$\widetilde  C_2$ and $\widetilde  C_3$, but we will not go into such details.

\section{One-dimensional systems. Computation of  $\I_\chi$} \label{sec:Icalc}

In this section we restrict to one-dimensional problems and
provide a general expression for the collision function
$I_\chi$ (see \eqref{eq:IDecomp}, \eqref{def:Theta},
\eqref{def:Gamma}) whose form shows that
$I_\chi$ depends on the 
reflection and transmission amplitudes associated to the potential $V$
and on the wave 
function of the light particle. Using this expression 
we compute the energy and momentum transfer occurring in a 
two-body collision.

Furthermore, assuming that the state of the light particle is represented
by a Gaussian wave packet with a
narrow spectrum in momentum, we prove an  
approximation of $\Theta_\chi$
(see~\eqref{eq:ThetaApprox}) to be used in Section \ref{thexpla}.

\subsection{Scattering operator, reflection and transmission amplitudes. } 
\label{sec:scatmat}

Consider a particle moving on a line under the action of  the potential $V$,
and assume hypotheses (H1)-(H3). 
We define the {\em transmission amplitude $t_k$} and the
{\em reflection amplitude $r_k$} as the two complex coefficients s.t.\ 
the action of the scattering operator $S$, defined in \eqref{defs}, reads
\be \label{cohen}
(S \chi)(x)=\frac1 {\sqrt{2 \pi}} \int_{\RR} \left[ {t_k} \widehat
  \chi(k) + r_{-k} \widehat \chi(-k)  \right] e^{ i k x} \, dk\, , 
  \quad \forall \,x \in \RR,
\ee
for any $\chi \in L^2 (\RR)$. We stress that definition \eqref{cohen} 
corresponds to the following formal action on plane waves
\[
S(e^{ikx}) = r_k
e^{-ikx} + t_k e^{ikx},
\]
which, in turn, agrees with the definition of reflection and
transmission amplitudes usually found in physics textbooks, namely,
$t_k$ and $r_k$ are the two complex coefficients s.t.\  the generalized
eigenfunction 
$\psi_k$ of the
operator $H_V$ corresponding to the generalized eigenvalue  $E = \f
{k^2} 2 \neq 0, k > 0$ fulfills the asymptotics
\be \begin{split}
\psi_k (x) \ & \sim \ \f 1 {\sqrt{2 \pi}} \left( e^{ikx} + r_k
e^{-ikx} \right), \qquad x \to - \infty, \\
\psi_k (x) \ & \sim \ \f 1 {\sqrt{2 \pi}} t_k  e^{ikx},  \qquad x \to + \infty.
\end{split} \ee
\noindent
It proves useful to represent the action of $S$ through the $2 \times 2$
matrices
\begin{equation} \label{sfourier}
S(k) := 
  \begin{pmatrix}  t_k & r_{-k} \\  
  r_k  & t_{-k}  \end{pmatrix}, \qquad k > 0,
\end{equation}
that act on the vectors $( \widehat \chi(k), \widehat
\chi(-k)\,)_{k >0}$ as follows
\be \label{SinFourier}
\forall k>0, \quad   \begin{pmatrix}  \widehat{S \chi}(k) \\  \widehat{S
\chi}(-k) 
\end{pmatrix} =
  S(k)
  \begin{pmatrix}  \widehat \chi(k) \\  \widehat \chi
    (-k)  \end{pmatrix}.
\ee 
Moreover, the unitarity of $S$ implies, for $k \neq 0$,
\begin{equation} \label{eq:relRetT}
|t_k|^2 + |r_k|^2=1, \quad 
 r_k \overline{t_{-k}} + t_k \overline{r_{-k}} =0
 , \quad |r_k | = |r_{-k} |.
\end{equation}
The fact that $S$ commutes with the Laplacian, together with its unitarity,
gives
\[
\| S \chi \|_{H^s} = \| \chi \|_{H^s}, \qquad \forall s \in \RR,
\ \forall \chi \in H^s(\RR).
\]
We are ready to prove that, as stated in Section 2,
the  condition $(H3')$ in dimension one implies condition $(H3)$.
\begin{lem} \label{lem:ScatMom1D}
Suppose that for some $s \in \RR$ and $C_s >0$
 the transmission and reflection coefficients satisfy 
\be \label{eq:reghypforS}
| \partial_k t_k| + | \partial_k r_k| \leq C_s (1 + |k|^2)^{\frac s2} =: C_s
\la k
\ra^s.
\ee
Then, for all $\chi \in H^s(\RR)$
\[
\|x S \chi  \|_2 = \| \partial_k [\widehat{S \chi}] \|_2 \leq  \| x \chi 
\|_2 + 2 C_s \| \chi \|_{H^s}.
\]
\end{lem}

\begin{proof}[Proof of Lemma~\ref{lem:ScatMom1D}]
Since
\[
\partial_k \begin{pmatrix}
              \widehat{S \chi}(k) \\  \widehat{ S \chi} (-k)
             \end{pmatrix}    =
[\partial_k S(k) ] 
\begin{pmatrix} \widehat \chi(k) \\  \widehat \chi (-k) \end{pmatrix}
+
S(k) \begin{pmatrix} \partial_k \widehat \chi(k) \\  \partial_k \widehat \chi
(-k) \end{pmatrix},
\]
one gets
\[
\bigl\| \partial_k \widehat{S \chi} \bigr\|_2 
\le
 \left( \int_0^{+\infty}
\biggl|[\partial_k S(k) ] 
\begin{pmatrix} \widehat \chi(k) \\  \widehat \chi (-k) \end{pmatrix}
\biggr|^2 \,dk \right)^{\frac12}
+
 \left( \int_0^{+\infty}
\biggl| S(k) \begin{pmatrix} \partial_k \widehat \chi(k) \\  \partial_k \widehat
\chi (-k)
\end{pmatrix} \biggr|^2 \,dk \right)^{\frac12}.
\]
By unitarity of $S(k)$, the second term in the r.h.s.
equals $\| \partial_k \widehat \chi \|_2  = \|x \chi
\|_2$. Furthermore, by \eqref{eq:reghypforS},
\bean
\int_0^{+\infty}
\biggl|[\partial_k S(k) ] 
\begin{pmatrix} \widehat \chi(k) \\  \widehat \chi (-k) \end{pmatrix}
\biggr|^2 \,dk
&\le&
 4 C_s^2 \int_0^{+\infty} \la k \ra^{2s} \bigl(|\widehat \chi(k)|^2 + 
|\widehat \chi (-k)|^2 \bigr) \,dk
\\
& \le & 4 C_s^2 \| \chi \|_{H^s}^2.
\eean 
This implies the claimed result.
\end{proof}

\paragraph{\bf The effect of translation} 
If the potential $V$ is translated
by a quantity
$X$, then
the reflected wave is delayed by a phase equal to $2kX$ and the
transmitted one remains unchanged. As a consequence, one has the following

\begin{lemme} \label{coefftransl}
Let $V$ be s.t.\  the Hamiltonian operator $H_V = - \f 1 2 \partial_x^2 +
V$ satisfies assumptions (H1)-(H3). Then, the translated Hamiltonian
operator $H_V^X  = - \f 1 2 \partial_x^2 + 
V (\cdot - X)$ satisfies (H1)-(H3) and, denoting the
corresponding reflection and transmission amplitudes by $r^X_k$ and $t^X_k$, one has 

\be \label{eq:delayedRT}
r^{\ssc X}_k = e^{2i k X} r_k, \qquad t_k^{\ssc X} = t_k, \quad
\forall k \in \RR \backslash \{ 0 \} .
\ee
\end{lemme}

\begin{proof}
According to the notation of Section 2, we denote by
$\theta_X$ the translation operator s.t.\  $\theta_X \chi = \chi
(\cdot - X)$. Then, one easily gets
\be \label{translevo}
U_V^X (t) \ = \ \theta_X U_V (t) \theta_{-X},
\ee
so that $U_V^X (t)$ and $U_V (t)$
 are unitarily equivalent and assumptions
(H1)-(H3) are preserved by translation. 
Furthermore, \eqref{translevo} implies
\be 
S_V^X  \ = \ \theta_X S_V  \theta_{-X}.
\ee
By direct computation
$ \widehat{\theta_{-X} \chi} (k) \ = \ e^{ikX} \widehat{\chi} (k)$, 
so one finally gets
\[
\widehat {S_V^X \chi} (k) \ = \ e^{-2ikX} r_{-k} \widehat \chi (-k) +
 t_k \widehat \chi (k) 
\]
and the proof is complete.
\end{proof}

\begin{corollary} \label{cor:ScatMatTrans}
The matrix $S^X_V$ reads
\be \label{eq:ScatMatTrans}
S^{\ssc X}_V (k) =
  \begin{pmatrix}   
  t_k & e^{-2i kX}\,r_{-k} \\  
  e^{2ikX} \, r_k  & t_{-k}
\end{pmatrix}.
\ee
\end{corollary}

Lemma~\ref{coefftransl} (and Corollary~\ref{cor:ScatMatTrans}) allow us to get a 
rather simple expression for the collision function $I_\chi$.
\begin{proposition}
For a one-dimensional  two-particle  system, endowed with an interaction
potential $V$ such that the hypotheses (H1)-(H3) are verified, the
collision function $I_\chi$ defined in \eqref{def:DecoFunction} can
be expressed as
\begin{equation}  \label{eq:IDecomp}
I_\chi (X,X')  =   1 - \Theta_\chi(X-X') + i\, \Gamma_\chi(X) - i\,
\Gamma_\chi(X'),
\end{equation}
with the definitions
\begin{align} \label{def:Theta}
\Theta_\chi(Y) &:= 
  \int_\RR   \!  \left( 1  -  e^{2ikY}\right) |r_k|^2  | \widehat
\chi(k)  |^2 \,dk  ,\\
 \label{def:Gamma}
\Gamma_\chi(X) &:=  i \int_\RR  e^{2i kX} \, \overline{r_{-k}} \,
t_{k}  \, \overline{\widehat \chi(-k)} \widehat \chi(k)  \,dk.
\end{align} 
\end{proposition}

\begin{proof}
The proof
is an elementary computation to be carried out  
using definition~\eqref{def:DecoFunction}, the equation~\eqref{cohen}, 
the relations~\eqref{eq:relRetT}, and Lemma~\ref{coefftransl}.  
\end{proof}

\begin{remark}{\em
By the change of variable $k \to -k$ in the integral defining
$I_\chi$ and the relations~\eqref{eq:relRetT},
 one immediately finds that $\Gamma_\chi (X)$ is real for any $X$.
}
\end{remark}

\medskip
\paragraph{\bf Effect of the collision operator
on kinetic energy and momentum of the
heavy particle.}

In order to interpret the functions $\Theta_\chi$ and $\Gamma_\chi$ we
study the transfer of energy and momentum between the heavy and the
light particle.

 We recall that
for a particle in the mixed state $\rho$ lying in a
$d$-dimensional space,  the average momentum and kinetic energy 
are given by 
\be \label{def:Moment}
P(\rho) = \Tr \Bigl( \frac12 \bigl[ (-i \nabla) \rho + \rho (-i \nabla) \bigr]
\Bigr)
\qquad \text{or } \;
P(\rho) = \frac i2 \int_{\erre^d} (\nabla_2 - \nabla_1) \rho(X,X) \,dX,
\ee
\be \label{def:Ekin}
E_{kin}(\rho) = \frac12\, \Tr ( -i\nabla \cdot \rho [-i\nabla]  ) 
\qquad \qquad \text{or } \;
  E_{kin}(\rho) = \frac12\, \int_{\erre^d}   \nabla_2 \cdot \nabla_1 \rho(X,X) \,dX.
\ee
The probability current $\vect j$ is defined, in terms of the density
operator, by
\be
\vect j :=  \frac12  \bigl[ \rho(-i \nabla) + (-i\nabla)\rho \bigr]   \qquad
\text{or }\; \vect j(X,X') := \frac i 2 ( \nabla_2  - \nabla_1)
\rho(X,X').
\ee
Remark that $P(\rho) = \Tr \vect j$. For the sake of interpreting
  the forthcoming proposition, one can consider that, if $\rho$ is the
  density operator representing the state of the heavy particle {\em
    before the collision}, then, in our approximation, ${\mt I}_\chi \rho$
  is the density operator representing the state of the heavy particle {\em
    after the collision}.

\begin{proposition} \label{prop:Exchange}
The momentum and the kinetic energy of a particle moving on a line, as it lies
in the mixed state represented by the density operator ${\mt I}_\chi \rho$, are given by
\begin{align}
 \label{eq:DeltaMoment}
P\bigl({\mt I}_\chi \rho \bigr)  & =   P(\rho) + i \Theta_\chi'(0) +
\f 12 \Tr
\Bigl( \Gamma_\chi'  \rho + \rho  \Gamma_\chi' 
\Bigr),
\\
\label{eq:DeltaEnerKin}
 E_{kin}( \mt I_\chi\rho) &=  E_{kin}(\rho) +  i  \Theta_\chi'(0) 
P(\rho)
    +  \frac12\, \Theta_\chi''(0)  +  \frac12\,  \Tr \bigl(  \Gamma_\chi' j +
j \Gamma_\chi' \bigr),
\end{align}
where $\Gamma_\chi'$ denotes the operator whose action is the
multiplication by the derivative of $\Gamma_\chi$ and $j$ is the only
component of the current $\vec j$ that is present in the
one-dimensional case.
\end{proposition}

\begin{proof}[Proof of Proposition~\ref{prop:Exchange}.]
From
decomposition~\eqref{eq:IDecomp}, one immediately gets
\bean
 \partial_1 I_\chi(X,X') &=&   -  \Theta_\chi'(X-X') + i 
\Gamma_\chi'(X) \\
 \partial_2 I_\chi(X,X') &=&   \Theta_\chi'(X-X') - i  
\Gamma_\chi'(X') ,
\eean
where $\partial_j$ denotes the derivative
w.r.t.\  the $j^{th}$ argument. By exploiting the second identity in 
\eqref{def:Moment}, a straightforward computation yields
\bean
P\bigl( \mt I_\chi \rho \bigr) = 
 P(\rho) + i \Theta_\chi'(0) +  \int_{\erre}  \Gamma_\chi'(X) 
\rho(X,X) \,dX, 
\eean
which may be rewritten as~\eqref{eq:DeltaMoment}.

\noindent
Concerning kinetic energy, by the second identity in \eqref{def:Ekin}
one gets
\bean
2\, E_{kin}( \mt I_\chi\rho) & = &  2\, E_{kin}(\rho) + \int_{\erre}  \bigl[ 
(\partial_1 I_\chi)(X,X)
(\partial_2 \rho)(X,X) + 
 (\partial_2 I_\chi)(X,X) (\partial_1\rho)(X,X) \bigr]  \,dX \\
 && \hspace{1cm} + \; \int_{\erre} \bigl( \partial_2 \partial_1 I_\chi \bigr) (X,X) 
\rho(X,X) \,dX.
\eean
Using decomposition~\eqref{eq:IDecomp}, one finally has
\bean
E_{kin}({\mt I}_\chi\rho) 
 & = &  E_{kin}(\rho) - \frac12\, \Theta_\chi'(0)  \int_{\erre}  \bigl[  
(\partial_2\rho)(X,X) -  (\partial_1\rho)(X,X) \bigr] \,dX \\
  &&  \hspace{.5cm} + \;\frac i2 \,  
 \int_{\erre}    \Gamma_\chi'(X)  \bigl[   (\partial_2\rho)(X,X) -
(\partial_1\rho)(X,X) \bigr]  \,dX 
  +  \frac12\,  \Theta_\chi''(0).
\eean
This finally leads to~\eqref{eq:DeltaEnerKin}.
\end{proof}

\begin{remark}{\em
First, by \eqref{def:Theta}, one has 
$\Theta_\chi(0) =0$, ${Re (\Theta_\chi'(0)) = 0}$, and  ${Im (\Theta_\chi''(0)) 
=
  0}$, so that all quantities in Proposition \ref{prop:Exchange} are real.
In particular, notice that
\be \label{eq:itheta'}
 i \Theta_\chi'(0) = 2 \int_{\erre} k |r_k|^2 |\widehat \chi(k) |^2 \, dk,
\ee
which is in general different from zero, so that a transfer of
  momentum and energy is possible even though one could intuitively
  suspect that the light particle is in fact too light in order to
  exchange momentum or energy with the heavy one.
In order to understand this fact, recall that
the light particle has a  momentum independent of $\eps$ and a kinetic energy of order 
 $\ep^{-1}$.  Thus the collision occurs
between two particles with momentum of the same order, for which exchanges
of momentum and energy can take place.
 
 Besides, the above formula~\eqref{eq:itheta'} has a relatively
 simple interpretation. The plane wave $e^{ikx}$
 has a probability $|r_k|^2$ of being reflected, i.e. to gain a momentum
 $-2k$. 
Since the state of the incoming
 particle can be understood as a superposition of plane waves with
 weight $\widehat \chi (k)$,
the average gain in momentum amounts to $-2 \int_\RR k \, |r_k|^2 |
\widehat \chi (k) |^2 dk$ for the light particle.
By conservation of momentum, the average gain in momentum for the heavy 
 particle equals the r.h.s. of \eqref{eq:itheta'}. 

On the other hand, the last term in \eqref{eq:DeltaMoment} does not have, at least to our
concern, a clear interpretation. This is due to the fact that it takes
into account the interference between the reflected and the
transmitted waves, so that there is no classical counterpart to provide
some understanding.

For the kinetic energy the situation is analogous: the sum of the
second and the third term in the r.h.s. of \eqref{eq:DeltaEnerKin}
\begin{align*}
 i \Theta_\chi '(0) P(\rho)  + \frac12\, \Theta_\chi ''(0) 
 &=  2 \int_{\erre}
(k+P(\rho))k |r_k|^2 |\widehat \chi(k) |^2 \,dk \\
&= \frac12 \int_{\erre} \bigl[ (2\, k+P(\rho))^2 - P(\rho)^2 \bigr] |r_k|^2 |\widehat \chi(k) |^2 \,dk
\end{align*}
can be understood similarly to the first term in the r.h.s. of \eqref{eq:DeltaMoment},
while the last term
is due to a superposition effect between transmitted and
reflected waves and its meaning is therefore less transparent. 
}
\end{remark}

\medskip
\paragraph{ \bf The case of an initial Gaussian state for the light
  particle.}

 Let us specialize to the case in which
the initial state of the incoming light particle is represented by a Gaussian
wave function, {\it i.e.} 
\be \label{gauschi}
\chi(x) = \frac1{(2\pi\sigma^2)^{1/4}} e^{-\frac{(x-x_l)^2}{4
\sigma^2}+ i p x},
\ee
where $x_l \in \RR$ is the  centre of the wave packet, $\sigma$
its spread, and $p$ its mean
momentum. Then,
\[
 \widehat \chi(k)  =  
  \left(\frac {2 \sigma^2} \pi \right)^{1/4} e^{- \sigma^2(k-p)^2  -
i (k-p) x_l}.
\]
We shall make this choice of state for the light particle in
Section \ref{SEC4}, when dealing with numerical simulations.
For this reason, we give simplified expressions for $\Theta_\chi$ and
$\Gamma_\chi$ and we provide some related approximation formulas that
prove easy to 
handle.  
In fact, for the Gaussian case
definitions (\ref{def:Theta}) and (\ref{def:Gamma}) yield
\begin{equation} 
\begin{array}{lll}
\Theta_{\sigma,p}(Y)
 &=& \sigma  \sqrt{\frac 2 \pi }  \int_\RR   \!  \left( 1  -  e^{2ikY}\right)
|r_k|^2  e^{-2\sigma^2(k-p)^2 }  \,dk \,,\\[3mm]
\Gamma_{\sigma,p}(X) &=&  i \sigma  \sqrt{\frac 2 \pi }
e^{-2\sigma^2 p^2} \int_\RR  \, t_k
\overline{r_{-k}}  \, e^{-2\sigma^2 k^2 + 2 i k(X- x_l) }  \,dk.
\label{eq:GammaG1}
\end{array}
\end{equation}
If the wave
packet has a large spread in position, so that its support in momentum is small
compared to the scale at which $|r_k|^2$ varies, then we can
approximate $\Theta_{\sigma,p}$ by using $|r_{p}|^2$ instead of $|r_k|^2$ in
the integral, and get  the following
approximation
\begin{align}
 \Theta_{\sigma,p}^{app}(Y)
& :=  |r_{p}|^2 \left( 1  -  \sigma  \sqrt{\frac 2 \pi } \int_\RR   \! 
  e^{2ikY -2\sigma^2(k-p)^2 }  \,dk  \right)
  \nonumber  \\
\ & =  |r_{p}|^2 \left( 1  -   e^{2i p Y - \frac{Y^2}{2 \sigma^2 }} 
\right). 
\label{eq:ThetaApprox}
\end{align}

Approximating $\Gamma_\chi$ turns out to be more difficult.  However, as
a first step, assuming that the light particle has a large momentum, we can
approximate $\Gamma_\chi$ by $0$ since the factor $e^{-2\sigma^2 p^2}$ is negligible
for $\sigma\, p$ large enough. 

The approximations introduced here can be expressed in terms of
density matrices. Indeed, one has the following proposition:
\begin{prop} \label{prop:Theta_error}
For any positive, self-adjoint operator $\rho$ with $\Tr \rho = 1$,
the following estimate holds 
\begin{align*}
 \bigl\| \Theta_{\sigma,p}(X-X') \rho(X,X')  - 
\Theta_{\sigma,p}^{app}(X-X') \rho(X,X') \bigr\|_{\mt L^1} 
 & \le \sqrt{\frac 2 {\pi\sigma^2}  } \, \left\| \frac {d |r_k|^2} {dk}
\right\|_\infty, \\
\bigl\|  i [ \Gamma_{\sigma,p}(X) - \Gamma_{\sigma,p}(X')] \rho(X,X') 
\bigr\|_{\mt L^1}
& \le 2 \, e^{- 2 \sigma^2 p^2},
\end{align*}
where we denoted an operator by its integral kernel.
\end{prop}

\begin{proof}
We will use the following simple estimates: for a wave packet $\chi$ with 
center $x_0$, spread $\sigma$, and momentum $p$, we have
\begin{align}
\int_\erre \left| | r_k |^2 -  | r_p |^2 \right|  | \widehat \chi (k) |^2 \,dk  
 & \leq \left\| \frac {d |r_k|^2} {dk}
\right\|_\infty \int_\erre   | \widehat \chi
  (k) |^2  | k - p | \,dk \nonumber \\
 & =  \f 2 {\sigma \sqrt { 2 \pi}} \left\| \frac {d |r_k|^2} {dk}
\right\|_\infty \int_0^{+ \infty} k e^{-k^2} \, dk \nonumber \\
& = \f 1 { \sigma \sqrt{2
    \pi} }  \left\| \frac {d |r_k|^2} {dk} \right\|_\infty,
    \label{estim:theta_inf}
\end{align}
and
\begin{align}
| \Gamma_{\sigma,p}(X) |  \le \sigma \sqrt{\f 2 \pi} e^{-2 \sigma^2 p^2} \int_\erre 
e^{- 2\sigma^2 k^2}\,dk = e^{- 2 \sigma^2 p^2}.
\end{align}

We shall only perform the proof in the case where $\rho$ is a rank one 
projector~: $\rho(X,X') = \varphi(X) \overline{\varphi(X')}$, where $\| 
\varphi\|_2 = 1$.
The general case follows by diagonalisation of a general $\rho$ and summation 
of the error given in the rank one case. Using~\eqref{estim:theta_inf}, we get

\begin{align*}
 \bigl\|   ( & \Theta^{app}_{\sigma, p}  (X - X')  - \Theta_{\sigma,p} (X - X') 
)
 \rho (X,X')  \bigr\|_{{\mt L}_1} \\
 & \le \left\| \rho \right\|_{{\mt L}_1}
 \int_\erre  \bigl| | r_k |^2 -  | r_p |^2 \bigr|\,  | \widehat \chi  (k) |^2  \,dk  +
 \left\| \int_\erre (| r_k |^2 -  | r_p |^2) e^{2ik(X-X')} | \widehat \chi
  (k) |^2 \rho (X,X')\,  dk \right\|_{{\mt L}_1} \\
& \le   \f 1 { \sigma \sqrt{2
    \pi} }  \left\| \frac {d |r_k|^2} {dk} \right\|_\infty
    +  \int_\erre 
    \bigl\| e^{2ikX} \varphi (X)  e^{-2ikX'} \overline{\varphi (X')}
\bigr\|_{\mt L_1}\left|| r_k |^2 -  | r_p |^2 \right|  | \widehat \chi
  (k) |^2  \,dk \\
& \le \sqrt{\frac 2 {\pi\sigma^2}  } \, \left\| \frac {d |r_k|^2} {dk}
\right\|_\infty.
\end{align*}

Before going to the estimate on $\Gamma_{\sigma,p}$, we recall that for any rank one 
operator $\rho'$, i.e. operator with kernel of the form 
$\rho'(X,X')= \varphi_1(X) \overline{\varphi_2(X')}$, we have the equality $ \| \rho' \|_{\mt L^1} = \| \varphi_1 \|_2 \| \varphi_2\|_2$. If we 
 apply this to the rank one operators with kernel 
$\Gamma_{\sigma,p}(X) \varphi(X) \overline{\varphi(X')}$ and
$ \varphi(X) \overline{ \Gamma_{\sigma,p}(X') \varphi(X')}$, we get
\begin{align*}
\bigl\|  i [ \Gamma_{\sigma,p}(X) - \Gamma_{\sigma,p}(X')] \rho(X,X') 
\bigr\|_{\mt L^1} 
& \le 
\bigl\|  \Gamma_{\sigma,p}(X)  \varphi(X) \overline{\varphi(X')} \bigr\|_{\mt 
L^1} +
\bigl\|   \Gamma_{\sigma,p}(X') \varphi(X) \overline{\varphi(X')} \bigr\|_{\mt 
L^1} \\
& \le  2\, \| \Gamma_{\sigma,p} \varphi \|_2 \| \| \varphi \|_2 
\le 2 \, \| \Gamma_{\sigma,p} \|_\infty \le 2 \, e^{- 2 \sigma^2 p^2}.
\end{align*}
This concludes the proof.
\end{proof}

%
\subsection{Particular potentials of interest.} \label{sec:diffpot}

Here, we briefly introduce three
 particular potentials that we shall use in the numerical simulations.

\medskip
\paragraph{ \bf Dirac's delta potential}

In the case $V=\alpha \delta_0$, with $\alpha>0$, 
the reflection and transmission amplitudes are given by (see
Proposition \ref{prop:ScatDirac})
\be \label{def:RetTDirac}
r_k = - \frac{\alpha}{\alpha - i |k|}, \qquad t_k = -
\frac{i|k|}{\alpha - i
|k|}\,, \quad \forall k \in \RR.
\ee
 
In the next section, we will use \eqref{def:RetTDirac}
to compute the function $I_\chi$ numerically via 
(\ref{eq:IDecomp}), \eqref{def:Theta}, \eqref{def:Gamma}.
  To avoid the numerical integration, one can use
formula~\eqref{eq:ThetaApprox}, which gives 
\begin{equation} \label{eq:ThetaGauss}
 \Theta_{\sigma,p}^{\delta,app}(Y) =
\frac{\alpha^2}{\alpha^2 + p^2} \Bigl( 1  -   e^{2i p Y - \frac{Y^2}{2 \sigma^2 }}  \Bigr).
\end{equation}

\medskip
\paragraph{ \bf Potential barrier} 
A further potential for which 
the scattering matrix can be explicitly computed
is the potential barrier, i.e. 
$$
V(x) := V_0 \indic_{[-a,a]}, \qquad V_0 = \frac{\alpha}{2 a} \, \quad a >0,
$$
where ${\indic}$ denotes the characteristic function and $\alpha>0$
measures the strength of the interaction. Letting 
$E= \frac{k^2}2$ denote the energy of the incoming wave
and defining $k_0:= \sqrt{2(E-V_0)} \in \CC$,
the transmission
and reflection amplitudes have the forms
\bea \label{T}
t_k & = & \frac{4 k k_0 e^{-2i k a}}{(k + k_0)^2 e^{-2i k_0 a}  -
  (k-k_0)^2 e^{2ik_0a}}\,, \quad \forall k \in \RR \backslash \{ 0 \},
  \eea
\bea \label{R}
r_k & = & \frac{(k^2 -k_0^2) e^{-2i k a} (e^{-2ik_0a} -
  e^{2ik_0a})}{(k + k_0)^2 e^{-2i k_0 a}  - (k-k_0)^2 e^{2ik_0a}},
\quad \forall k \in \RR \backslash \{ 0 \} .
\eea
 
\medskip
\paragraph{ \bf Numerical approximation for more general
  potentials}
In the case of more general potentials, there is no analytic
expression for the amplitudes $r_k$ and $t_k$, however, we can compute
them numerically.

We assume that the potential $V$ rapidly
decreases  at infinity, 
and choose a sufficiently large $a$ such that we can approximate $V$
by $0$ on $\RR 
\backslash [-a,a]$. 
Let us shortly summarize
the classical procedure to calculate the reflection and
transmission amplitudes.

We look for generalized eigenfunctions $\psi_k$ of the Hamiltonian
$-\frac12 \Delta + V$  
associated to the eigenvalue $E=\frac{k^2}2$. Thanks to our approximation, 
these eigenfunctions are combinations of the free waves $e^{ikx}$ and $e^{-ikx}$ 
outside the interval $[-a,a]$. 
For $k > 0$ we look for solutions satisfying
\be \label{eq:scat_num}
\psi_k(x) :=
\begin{cases}
e^{ik(x+a)} + {r_k} e^{-ik(x+a)} & \text{for }\; x <-a, \\
{t_k} e^{ik(x-a)}& \text{for }\; x >a.
\end{cases}
\ee
In order to find the values of
$t_k$ and $r_k$, one must solve  the  
stationary Schr\"odinger equation associated with transparent 
 boundary conditions in the interval $[-a,a]$
\be  \label{SCH_stat}
\begin{cases} 
- \frac 1 2 \psi''_k(x) + V \psi_k = E \psi\,, \quad x \in [-a,a],\\
\psi'_k(-a) + ik \psi_k(-a) = 2ik, \\
\psi'_k(a) - ik \psi_k(a) = 0.
\end{cases}
\ee 
Transparent boundary conditions express the fact that the wave
function as well as its derivative are continuous at $\pm a$. 
Using the continuity of the wave function and of its derivative at
$x=\pm a$, it  
can be checked that the boundary conditions in~\eqref{SCH_stat} are indeed 
satisfied if and only if conditions~\eqref{eq:scat_num} are satisfied for some 
$r_k$ and $t_k$.  The reflection and
transmission amplitudes are then given by 
\be \label{TTRR}
t_k:= \psi_k(a) \,, \quad r_k:= \psi_k(-a)-1\,, \quad \forall k >0.
\ee
For a wave coming from the right, i.e. $k<0$, the procedure is
analogous.

\section{Numerical asymptotic resolution of the two-body
  Schr\"odinger system} \label{SEC4} 

In this section we use the approximations introduced in Sections \ref{theorem} and
\ref{sec:Icalc} in order to efficiently resolve  the two-body
Schr\"odinger equation (\ref{SCH2D})-(\ref{due}) in the regime $\eps \ll 1$.
The final aim is to quantify and study numerically the
decoherence effect induced on the heavy particle by the interaction
with the light one.
\subsection{Model and initial
data} \label{subsec41}

According to Theorem \ref{thm:geneRho},
for small  values of $\eps$ we can replace the
resolution of the two-body Schr\"odinger equation
(\ref{SCH2D})-(\ref{due}) or, equivalently, of equation \eqref{SCHRhoScat} for
density operators, by the resolution of 
system ~\eqref{SCHRhoScatLim-rid} for the reduced density operator
of the heavy particle. Rephrasing the latter as an equation for the
integral kernel $\rho^{M,a}(t,X,X')$ of the operator  $\rho^{M,a}(t)$,
one gets
\be \label{matrixlim}
\begin{cases}
\ds i \, \partial_t \rho^{M,a} (t,X,X')=-  \frac 1{2M} (\Delta_X
- \Delta_{X'})
 \rho^{M,a}  (t,X,X') \,, \quad
\forall (X,X') \in \RR^2\,, \,\,\, \forall t \in \RR^+ \\[3mm] 
\ds \rho^{M,a}(0,X,X') = \rho^{M}_0 (X,X')\,{I}_{\chi} (X,X') ,
\end{cases}
\ee
where the collision function $I_\chi$ is given by formulas
(\ref{eq:IDecomp}), (\ref{def:Theta}), (\ref{def:Gamma}), and
$\rho_0^M (X,X')$ is the integral kernel of the operator $\rho_0^M$,
which represents the state of the heavy particle before the collision. 
We set
\be \label{rom0}
\rho^M_0(X,X'):= \varphi(X)\, \overline{\varphi(X')},
\ee
where
\be \label{initwf}
\varphi(X): = N \left( \varphi_-(X) + \varphi_+ (X) \right)
\ee
with
\bea
 \label{def:phipm}
\varphi_\pm(X) & := & 
\frac 1 { (2 \pi)^{1/4}\, \sqrt{\sigma_H}}  
 e^{-\frac{(X \mp X_0)^2}{4 \sigma_H^2}} e^{ \mp i p_H X } \\
N & : = & \sqrt 2 \left( 1 + e^{- \f {X_0^2} {2 \sigma_H^2}} e^{- 2
    \sigma_H^2 p_H^2}            \right)^{\f 1 2}. \label{enne}
\eea
{The parameters $X_0$, $p_H$ and $\sigma_H$
  are positive.} 
\begin{figure}[htbp]
\begin{center}
\includegraphics[width=7.7cm]{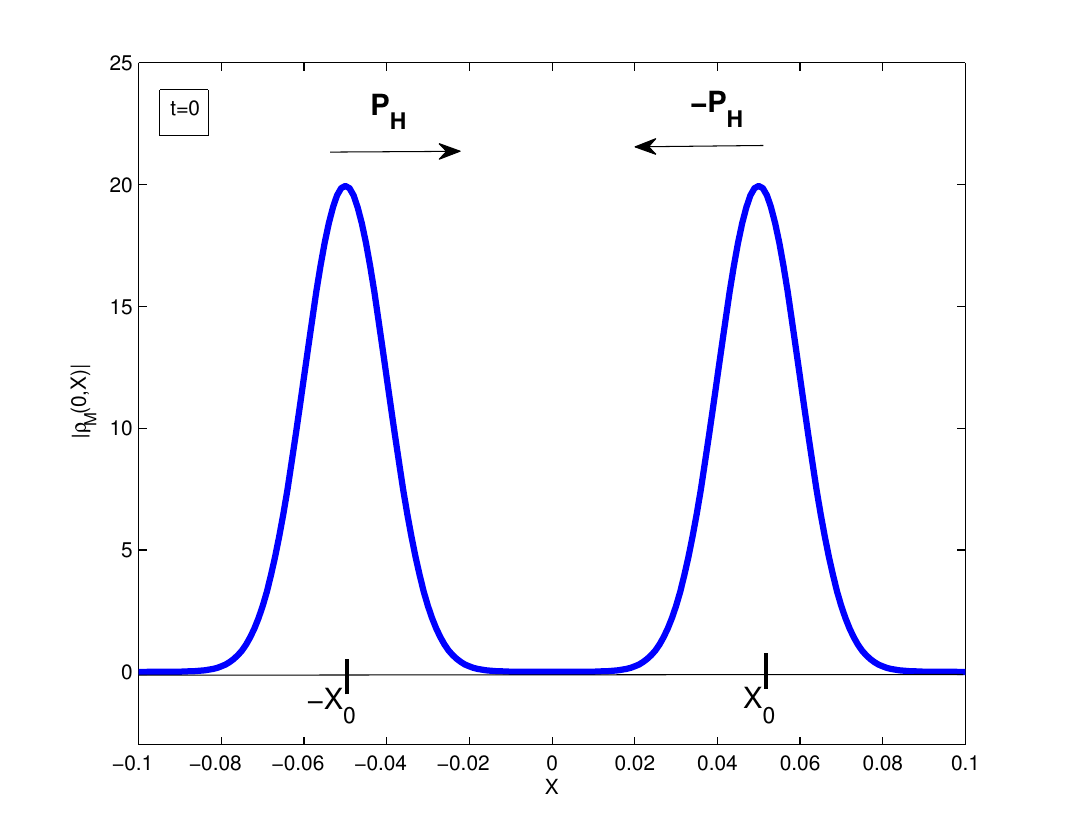}\hfill
\includegraphics[width=7.7cm]{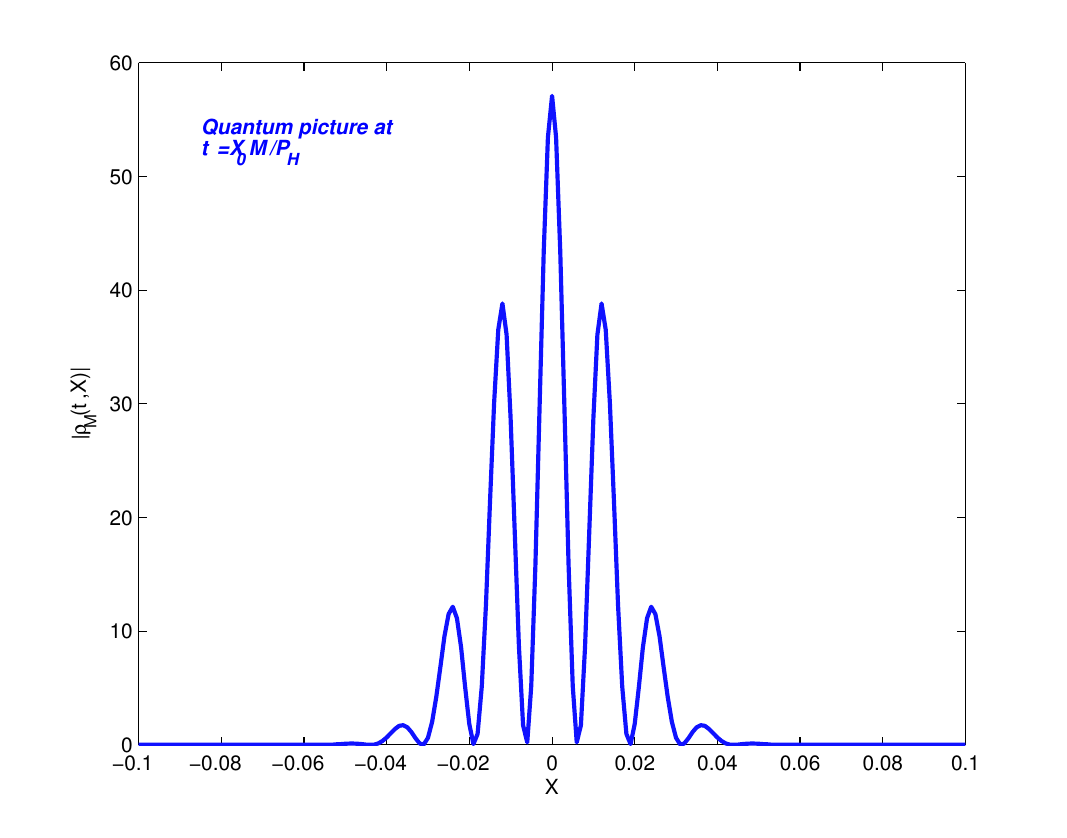}
\end{center}
\vspace{-0.3cm}
\caption{\label{Init} 
{\footnotesize Left: Probability density associated to the initial
  state of the heavy particle.  
Right: Probability density associated to the state of the heavy
particle in the case of no interaction, at the time of maximal overlap
of the two bumps}}
\end{figure}

\noindent
Then, the integral kernel \eqref{rom0}
can be rewritten as
\begin{align} \nonumber 
\rho^M_0(0, & X,X')  =   N^2  \bigl[  \varphi_-(X) + \varphi_+(X) \bigr] 
\bigl[ \overline{ \varphi_-(X')} + \overline{\varphi_+(X')} \bigr] 
 \\ 
 & =  N^2 \left[  \varphi_-(X)  \overline{ \varphi_-(X')} +
   \varphi_-(X)  \overline{ \varphi_+(X')} +
  \varphi_+(X)  \overline{ \varphi_-(X')} +   \varphi_+(X)  \overline{
    \varphi_+(X')} \right]. 
    \label{rhoinitapp}
\end{align}
The two terms $\varphi_\pm(X) 
\overline{\varphi_\pm(X')}$ are called {\em diagonal},
while the two terms
$\varphi_\pm(X) \overline{\varphi_\mp(X')}$ are called {\em
  antidiagonal}. In fact, in view of   
definition~\eqref{def:phipm} the products 
$\varphi_\pm(X) \overline{\varphi_\pm(X')}$ rapidly decay outside of a diagonal 
region $\{ |X-X'| \simeq \sigma_H \}$, while the products $\varphi_\pm(X) 
\overline{\varphi_\mp(X')}$ are essentially supported in the region
$\{ |X+X'| \simeq \sigma_H  
\}$. 
 
Physically, the density matrix ``before the collision'' $\rho_0^M$ or,
equivalently, 
the initial wave function \eqref{initwf},
 describes a state consisting of a
quantum superposition of two localized bumps centred respectively at
$\pm X_0$ and 
moving against 
each other with relative speed $2 p_H /M$, as illustrated in the left
plot of Figure \ref{Init}.  If no light particle or, more generally, no
  interaction is present, then one should use $\rho_0^M (X,X')$ as
  initial data in \eqref{matrixlim}. Thus, at time $M X_0 / p_H$ the
  non-diagonal terms in 
  \eqref{rhoinitapp} give rise 
 to  
an interference pattern, shown in the right plot of
Figure \ref{Init}. The emergence of     interference 
is due to the
non-diagonal terms in \eqref{rhoinitapp}.
On the other hand, due to the collision, the initial data in \eqref{matrixlim} is is replaced by $\rho_0^{M,a} (X,X') = I_\chi
(X,X')\rho_0^M (X,X')$. We will show in Section \ref{53} that the
presence of the factor $I_\chi$  dampens the interference.

\subsection{Numerical domain and discretization}

Here we give some brief explanation about the numerical resolution of
equation \eqref{matrixlim}.

First,
we truncate the spatial domain
$\RR^2$ to a bounded simulation domain $\Omega_X^2:=(-H,H) \times
(-H,H)$ and impose 
boundary conditions on $\partial \Omega_X$. To simplify 
computations, we choose homogeneous Neumann boundary
conditions, which prescribe that the particle is reflected at the
boundaries. However, if the domain is sufficiently
    large, the presence of the boundaries has negligible influence on the
    dynamics of the heavy particle.

Second, we discretize equation \eqref{matrixlim}.
For the discretization in time  we employ the 
Peaceman-Rachford scheme which is unconditionally stable and second-order accurate.
Let us explain in more detail the steps in the scheme.  We start by
discretizing the time interval
$[0,T]$ and the simulation domain of the 
heavy particle $\Omega_X=(-H,H)$. 
Let us introduce the time and space steps
\[
\Delta t = \frac T L >0     , \qquad h_X : = \f {2H} {J-1}>0 ,
\qquad {\textrm{with}} \ L,J \in {\mathbb{N}}
\]
and define the
homogeneous discretization
$t_l : = l \Delta t$, $X_j = -H + (j-1) h_X,$ so that
\[
0 = t_0 \le \cdots \le t_l \le \cdots \le t_L=T, \quad -H  = X_1 \le
\cdots \le X_j \le \cdots \le X_J=H. 
\]
Then, defining the operators $A,B:{\cal H} \subset L^2(\Omega_X) \rightarrow
L^2(\Omega_X)$  
\[
A:=- \frac 1 {2 M} \Delta_X, \quad  B:= \frac 1 {2 M}
\Delta_{X'},\quad {\cal H}:= \{ \phi \in H^2(\Omega_X)\,\, / \,\,
\partial_n \phi=0, \,\, \textrm{on}\,\, \partial \Omega_X\, \}, 
\]
where $\partial_n$ denotes the outward normal to the boundary
$\partial  \Omega_X$, the Peaceman-Rachford scheme for the system
(\ref{matrixlim}) writes
\be \label{sehrcrank}
\rho^{l+1}=(iId - \frac{\Delta t}2 A)^{-1}(iId + \frac{\Delta t} 2
B)(iId - \frac{\Delta t}2 B)^{-1}(iId + \frac{\Delta t} 2
A)\rho^{l}, \qquad  l=0,\cdots,L-1,
\ee
where $\rho^l$ (resp. $\rho^l_{ij}$) denotes the approximation of 
$\rho^{M,a}(t_l)$ 
(resp. $\rho^{M,a}(t_l,X_i,X_j)$).
Notice that \eqref{sehrcrank}
is a sequence of Euler-explicit,
Crank-Nicolson and Euler-implicit steps.
Equivalently, one performs a sequential resolution of two 1D
systems 
$$
(iId - \frac{\Delta t} 2 B)\rho^{l+1/2}=(iId + \frac{\Delta t} 2
A)\rho^{l} \,, \quad (iId - \frac{\Delta t} 2 A)\rho^{l+1}=(iId +
\frac{\Delta t} 2 B)\rho^{l+1/2}. 
$$
Finally, we discretize the
operators $A$ and $B$ in space via a standard
second-order centered method.

The 
parameters employed in the simulations are summarized in Table \ref{tab}.
 
\begin{table*}[htbp]
\begin{center} 
\begin{tabular}{|c|c||c|c|}
\hline
$2*H$&$2*10^{-1}$&$J$&$201$\\
\hline
$T$&$1.92*10^{-2}$&$L$&$120*20+1$\\
\hline
$\hbar$&$1$&$p_H$&$3.4*M$\\
\hline
$M$&$100$&$p$&$1.25;\,2.5;\, 3.5*10^{2}$\\
\hline
$\sigma_H, \sigma$&$10^{-2}, 2*10^{-2}$ &$X_0, x_l$&$5*10^{-2}, 2*10^{-1} $\\
\hline
$\alpha$&$0,\cdots,40 * 10^{2}$&&\\
\hline
\end{tabular}
\end{center}
\caption{{\footnotesize Parameters used in the numerical simulations.}}
        \label{tab}
\end{table*}
\vspace{0.3cm}

Let us briefly explain the reasons why the present numerical method is 
faster than the one previously employed in \cite{Cla_Riccardo}.

\noindent
First, thanks to Theorem \ref{thm:geneRho} all information on the interaction is
embodied in the collision operator ${\mt I}_\chi$ and is present in
problem\eqref{matrixlim} through the initial condition
only. Therefore,
one can get rid of any variable related to the light particle and thus
of the fast time scale. The initial multi-scale problem then reduces
to a one-scale problem, allowing a considerable gain in efficiency
as compared to the 
method employed in \cite{Cla_Riccardo}. 

\noindent
Second, the scheme is an alternating-direction implicit
(ADI) one, i.e. the actions of the two operators $A$ and $B$,
acting respectively on the variable $X$ and $X'$, are separated, so
  that, compared to a direct
resolution of (\ref{matrixlim}) via Crank-Nicolson method, the
computational costs are drastically reduced.

\subsection{Numerical results and interpretation} \label{53} 
Here we present some numerical results obtained via the 
resolution method of the evolution equation (\ref{matrixlim})
introduced in the previous section.
We give a detailed analysis for the case of a Dirac's delta
interaction potential, and then stress the main analogies with the
cases of a potential barrier and of a Gaussian potential. Finally, we
sketch the
case with
multiple light particles. For any choice of the interaction potential $V$, 
the reflection and transmission amplitudes are
computed as detailed in Section \ref{sec:Icalc} and the
corresponding collision function $I_\chi$ is calculated numerically
by formulas (\ref{eq:IDecomp}), (\ref{def:Theta}),
(\ref{def:Gamma}). 

\subsubsection{Dirac's delta potential} \label{SEC41}
Here we consider the case
$V(x)=\alpha  \delta_{0}(x)$, with $\alpha \in \RR^+$. 

 The left plot in Figure \ref{ima_rho_Dirac} shows the quantity
$ | \rho^{M}_0 (X,X') |$ (i.e. the state of the heavy particle before
 the collision with the light one).
Notice that the non-trivial values of 
  $\rho_0^M (X,X')$ are concentrated in four bumps. In accordance with
the terminology introduced in 
  Section \ref{subsec41}, the two bumps
located around the diagonal $X=X'$ are 
  called {\em diagonal} while the two others, located around the
  set $X = -X'$, are called {\em antidiagonal}. The diagonal bumps
  give the probability density associated to the state of the
  heavy particle, while the antidiagonal bumps are responsible for the
  interference. Diagonal and antidiagonal bumps share the
same shape and the same size.

\noindent 
The right plot in Figure \ref{ima_rho_Dirac} displays
$ |\rho^{M,a}(0,X,X')| =| {I}_\chi (X,X') \rho^{M}_0 (X,X')|$
(i.e. the state of the heavy particle immediately after the 
collision) in the test case $\alpha=10^3$. 
It is easily seen
that, as an effect of the
collision with the light particle, the antidiagonal bumps are damped,
thus providing the expected attenuation of the 
interference. 

\begin{figure}[htbp]
\begin{center}
\includegraphics[width=7.7cm]{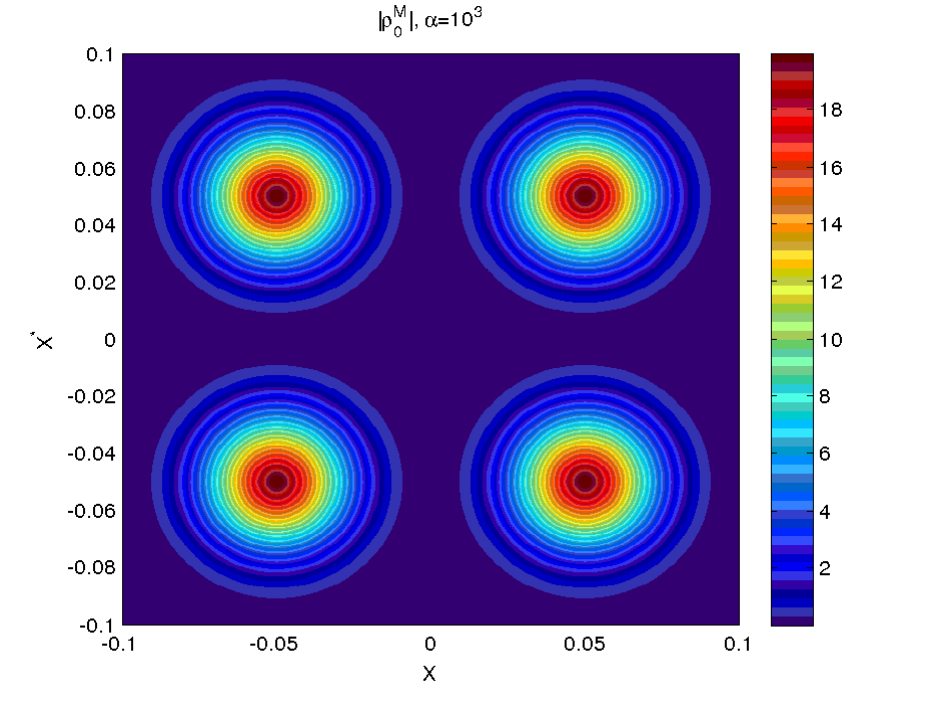}\hfill
\includegraphics[width=7.7cm]{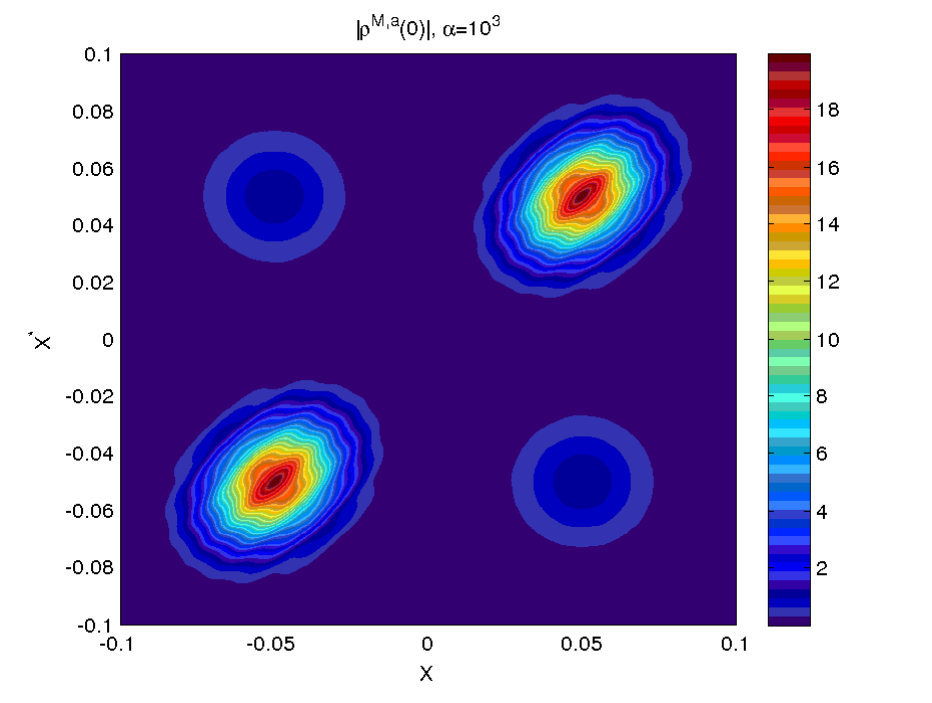}
\end{center}
\vspace{-0.3cm}
\caption{\label{ima_rho_Dirac} 
{\footnotesize  Test case: Dirac potential with $\alpha=10^3$. Left:
  Plot of $|\rho^{M}_0 (X,X')|$ before the
  collision; 
Right: Plot of $|\rho^{M,a}(0,X,X')|$ immediately after the collision}.} 
\end{figure}

Figure \ref{decoh_effect_Dirac} is devoted to  the collision function
$I_\chi$. In the left plot we
show  $|I_\chi(X,X')|$ corresponding to the right plot of 
 Figure \ref{ima_rho_Dirac}, while
in the right plot of Figure \ref{decoh_effect_Dirac}  
we give $|I_\chi(X,-X)|$ for different values of $\alpha$. 
One can observe that, as the strength of the potential varies,
the band width 
of $|I_\chi(X,-X)|$ remains unchanged;  on the other hand,
  notice that the more the
  potential is intense,
the more the quantity $|I_\chi(X,-X)|$ is reduced for large
values of $X$. 
It is precisely this reduction which causes the damping of the antidiagonal
bumps in Figure \ref{ima_rho_Dirac}.

\begin{figure}[htbp]
\begin{center}
\includegraphics[width=7.7cm]{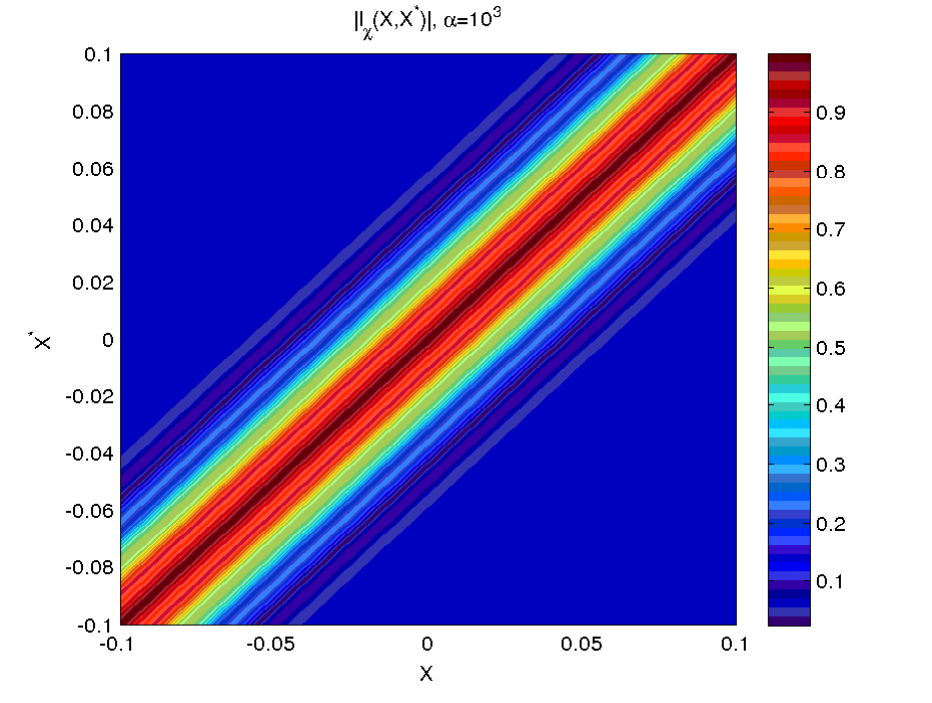}\hfill
\includegraphics[width=7.7cm]{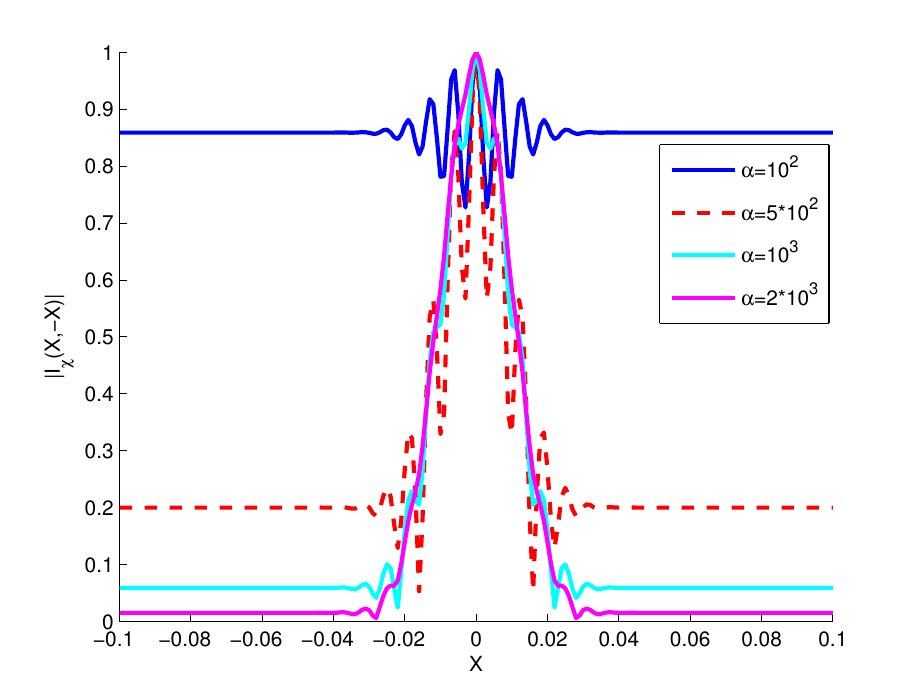}
\end{center}
\vspace{-0.3cm}
\caption{\label{decoh_effect_Dirac} 
{\footnotesize Left: Plot of $|I_\chi(X,X')|$ for $\alpha=10^3$. Right: Plot
  of $|I_\chi(X,-X)|$ for
  several values of $\alpha$.}}
\end{figure}

{In order to examine how the decoherence effect varies with the
  momentum of the light particle, in Figure
\ref{decoh_effect_vit} we plot 
$|I_{\chi}(0.05,-0.05)|$ for several values of $\alpha$ and three
different momenta $p$ of the light particle. We observe that the
larger  the momentum is, the smaller the decoherence
effect on the heavy particle is. This can be explained by the fact that
most of the light particle is transmitted when its momentum is large.}

\begin{figure}[htbp]
\begin{center}
\includegraphics[width=7.7cm]{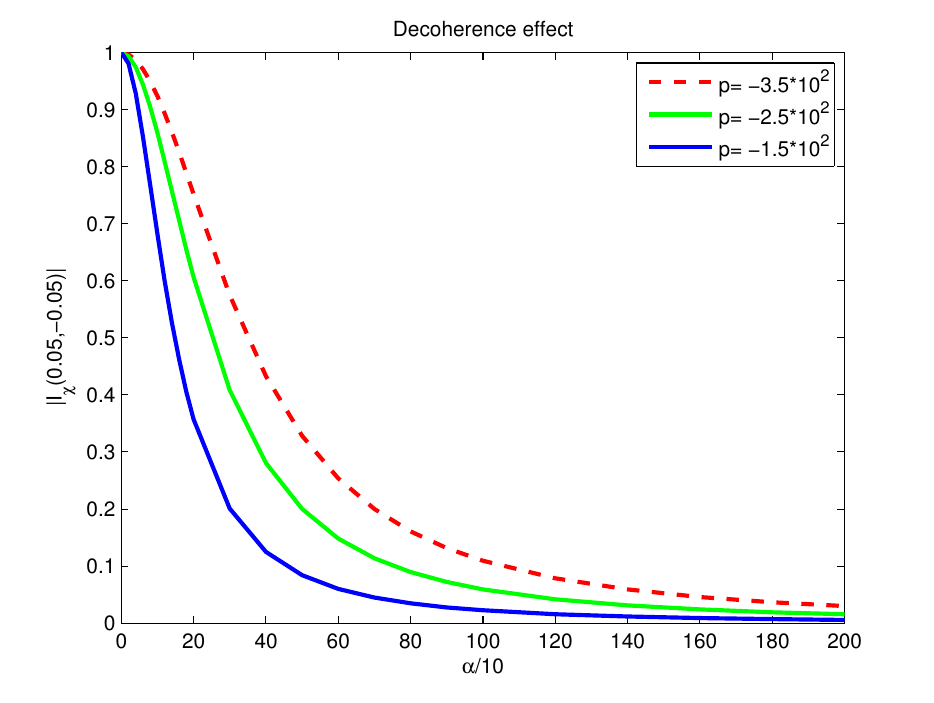}
\end{center}
\vspace{-0.3cm}
\caption{\label{decoh_effect_vit} 
{\footnotesize The quantity $|I_\chi (0.05, -0.05)|$ as a function of
  $\alpha$ for three different values of the momentum
 of the light
  particle.}} 
\end{figure}

Finally, in Figure \ref{fringes1} we display the 
probability density $\rho^{M,a}(t_*,X,X)$ associated to the state of
the heavy particle 
at the time $t_*=X_0 M 
/p_H$ of maximal overlap of the two diagonal bumps. 
The left plot in  Figure
\ref{fringes1} corresponds to a collision with a light particle
arriving from the right with mean momentum $p=-2.5*10^2$,
 for  several
potential strengths $\alpha$. One sees that the probability density
associated to the state of 
  the heavy particle splits into  a component that exhibits complete
  interference and a bump  that 
 travels with mean momentum $p_H+p>p_H$ towards the
right without experiencing interference. 
We refer to the component that displays interference as the
{\em coherent part}, while the component in which interference is
absent is referred to as the {\em decoherent part}.

\noindent
In the right plot, the light
particle has momentum $p=0$ and is located at the centre $x_l=0$.
The interference pattern exhibits a clear decoherence
  effect. In particular, notice that inside
  the pattern there are no points with zero probability. The
  corresponding plot is similar to the ones exhibited in
  \cite{Cla_Riccardo} through a direct use of the Joos-Zeh formula.
In fact, this plot too can be understood as the simultaneous presence
of a coherent and of a decoherent part, except that here,  since
the momentum of the decoherent 
part is zero, the two
components share 
the same support. 

A theoretical explanation of the appearance of the decoherent bumps is
given is Section \ref{thexpla}.

\begin{figure}[htbp]
\begin{center}
\includegraphics[width=7.7cm]{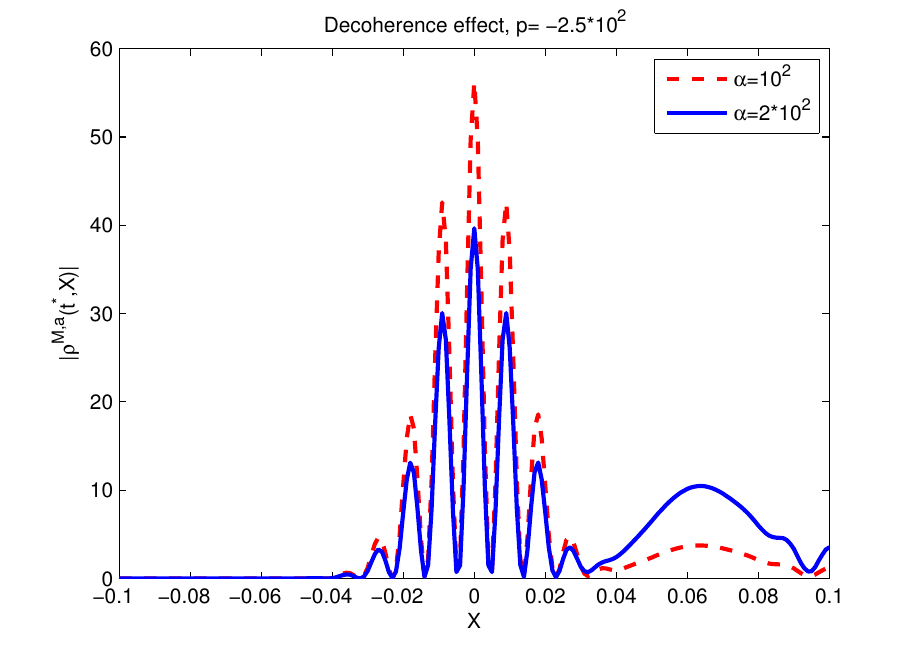}\hfill
\includegraphics[width=7.7cm]{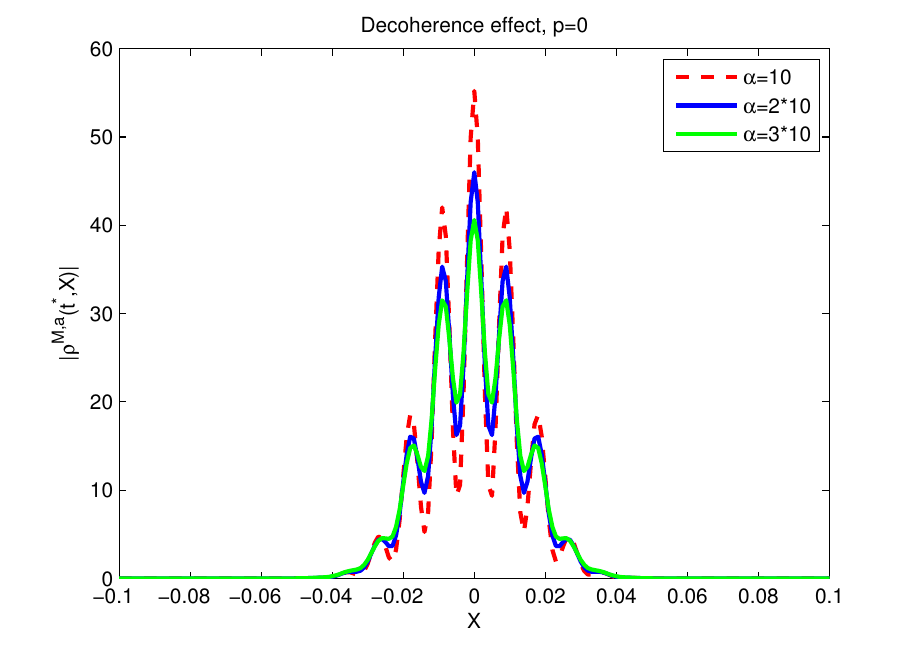}
\end{center}
\vspace{-0.3cm}
\caption{\label{fringes1} 
{\footnotesize Attenuation of the interference pattern of the heavy
  particle, in the case that the light particle comes from the left with
  $p=2.5*10^2$ (left Fig.) resp. $p=0$ (right Fig.)}} 
\end{figure}


\subsubsection{Potential barrier and Gaussian potential} \label{SEC42}
For the potential barrier
\be \label{potbar}
V(x) :=  V_0 {\indic}_{[-a,a]}\,, \quad V_0= \frac \alpha 
{2a}, \quad \alpha \in \RR^+\,, \quad a \in [10^{-4},10^{-2}],
\ee
as well as for the Gaussian potential
\be \label{potgau}
V(r):=V_0 e^{-\frac{r^2} {2 \sigma^2}}, \quad 
V_0= \frac \alpha {\sqrt{2 \pi} \sigma }, \quad \alpha \in \RR^+\,,
\quad \sigma \in [10^{-4},10^{-2}], 
\ee
we carried out computations and simulations following the line of
Section \ref{SEC41}.

For the former case, 
reflection and transmission amplitudes are given by formulas
(\ref{T})-(\ref{R}). For the latter case, we followed
the computation of the reflection and transmission amplitudes as
defined by the procedure detailed in (\ref{SCH_stat})-(\ref{TTRR}).

In both cases, the normalization constants $V_0$ have been chosen in order to guarantee 
that
\[
\int_{\RR} V(x)\, dx= \alpha, 
\]
so that the effects put in evidence
in this section can be compared with the effects carried out by the 
Dirac's delta potential $\alpha\, \delta_0$.

As far as scattering is concerned, the only consequence of the
interaction potential 
is the values of the reflection and transmission amplitudes. 
Thus, we just compare $r_k$, $t_k$ and
$I_\chi$ for the Dirac's delta, 
the potential barrier \eqref{potbar} and the Gaussian potential
\eqref{potgau}. The results are 
illustrated in Figures \ref{Coef_DBG} and \ref{Theta_DBG} for fixed
potential strength $\alpha=5*10^2$, momentum $p=-2.5*10^2$ and several
choices of
$a$ and $\sigma$. As expected, we found that the results obtained for
the potential barrier as $a \rightarrow 0$, as well as those obtained 
for the Gaussian potential as $\sigma \to 0$, approach those obtained 
using the Dirac's delta potential.

\begin{figure}[htbp]
\begin{center}
\includegraphics[width=7.7cm]{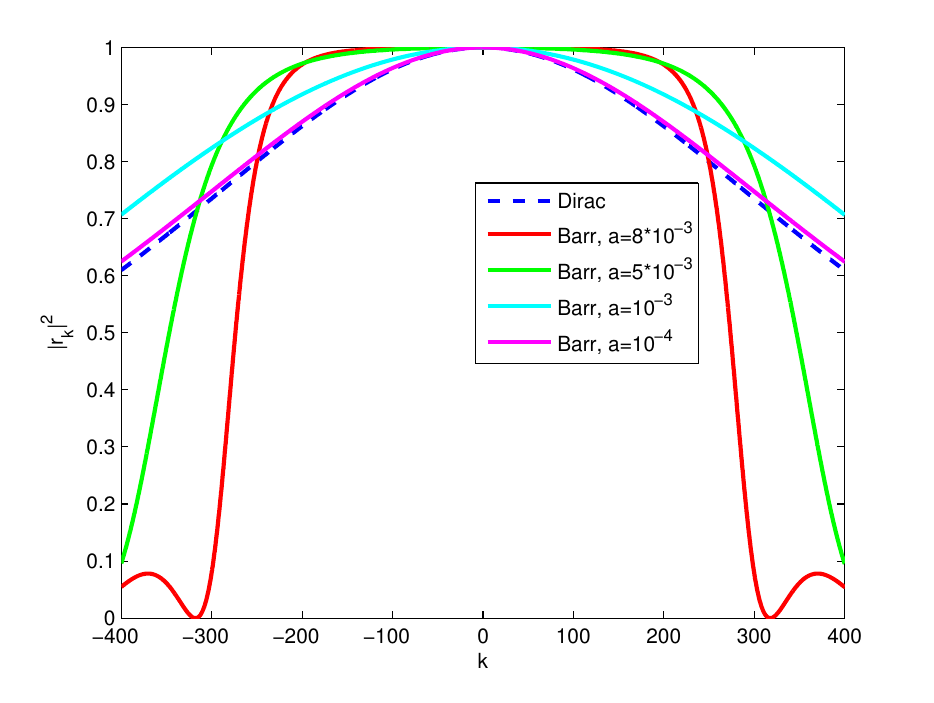}\hfill
\includegraphics[width=7.7cm]{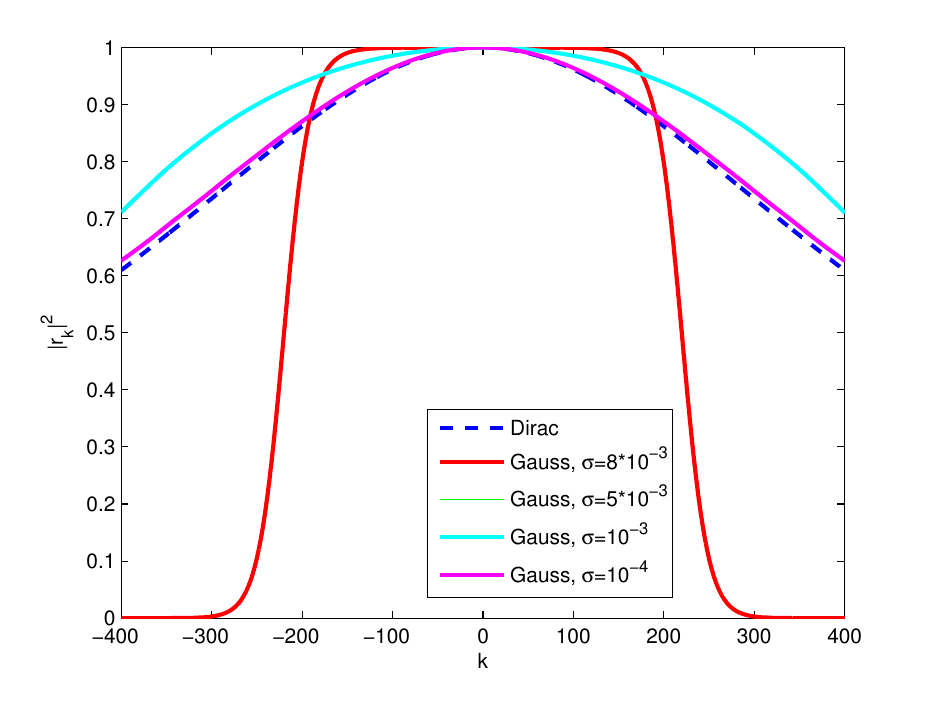}
\end{center}
\vspace{-0.3cm}
\caption{\label{Coef_DBG} 
{\footnotesize Comparison of the reflection amplitudes $|r_k|^2$
  corresponding to three different interaction potentials, with fixed
  potential strength $\alpha=5*10^2$ and various $a$ and $\sigma$
  values. Left: Dirac's delta and potential barrier. Right: Dirac's
  delta and Gaussian
  potential.}}  
\end{figure}

\begin{figure}[htbp]
\begin{center}
\includegraphics[width=7.7cm]{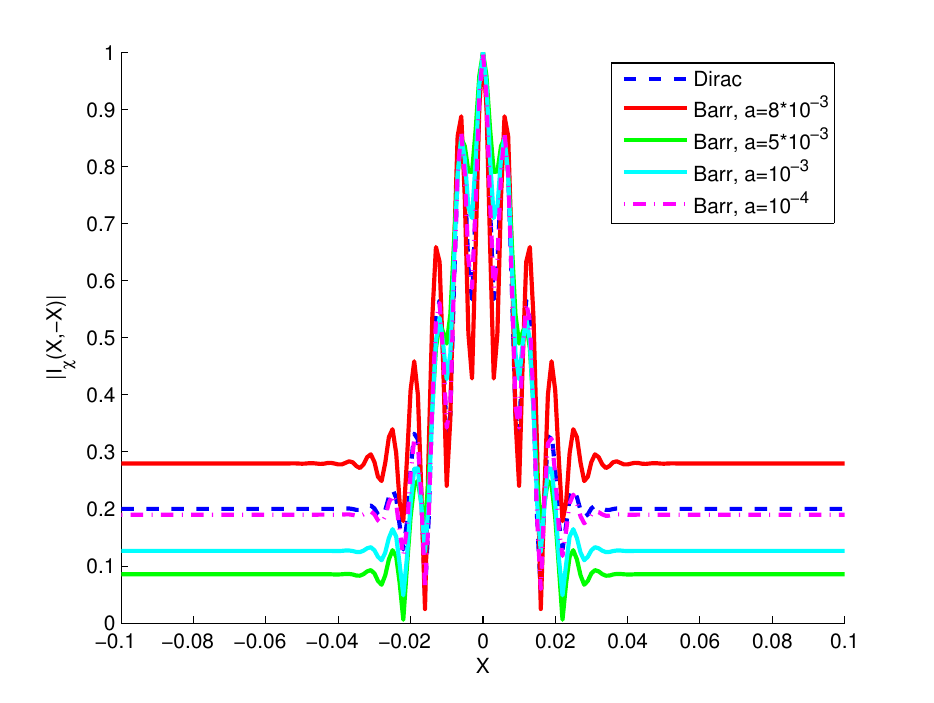}\hfill
\includegraphics[width=7.7cm]{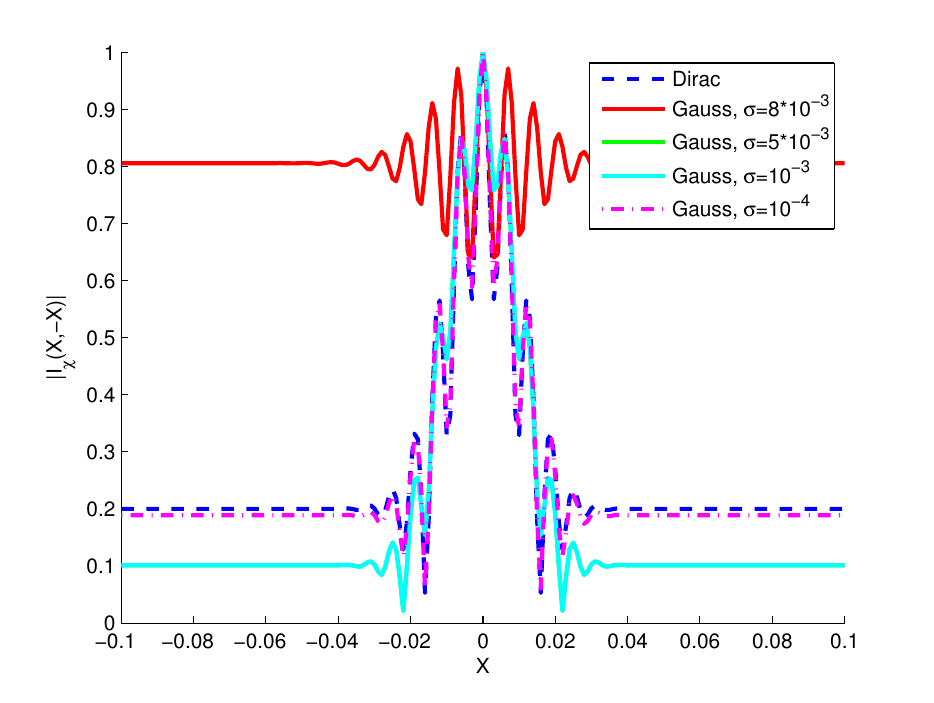}
\end{center}
\vspace{-0.3cm}
\caption{\label{Theta_DBG} 
{\footnotesize Comparison of the collision function $I_\chi(X,-X)$
  corresponding to three different interaction potentials, with fixed
  potential strength $\alpha=5*10^2$ and various $a$ and $\sigma$
  values. Left: Left: Dirac's delta and potential barrier. Right: Dirac's
  delta and Gaussian
  potential.}}  
\end{figure}
\subsubsection{Several light particles} \label{SEC43}

We suppose that many light particles are injected one-by-one into the
computation domain, in such a way that the heavy particle undergoes 
a finite sequence of collisions at times
$t_k:=4  k \, \Delta t$. At any collision, the state of the
  light particle is supposed to be the same, i.e., the $k.$th
  colliding light particle lies in the state represented by the wave
  function
$U_0 (t_k - \varepsilon^{- \gamma} ) \chi$.
Through any time interval $(t_k,t_{k+1})$ between two
collisions, the 
heavy particle evolves freely. 
The state of the heavy particle after each collision
$\rho(t_k^+)$ is then related to the state before collision $\rho(t_k^-)$
by 
$$
\rho(t_k^+) = {\mathcal I}_\chi [\rho(t_k^-)].
$$
On the left plot of Figure \ref{decoh_effect_coll} we show 
the probability density $\rho^{M,a}(t^*,X,X)$ associated to the state
of the heavy particle at the time of maximal overlap. The plot refers to 
the case of a Dirac's delta potential
with strength $\alpha=10$, momentum of the light particle $p=0$ and 
$N=1,2,3$ collisions. As expected, multiple collisions enforce the 
destruction of the interference pattern.

If one is interested in the limit of infinite incoming light
  particles, then a significant re-scaling of the potential 
should be $\alpha/\sqrt{N}$ with fixed $\alpha$: with this scaling, the decoherence 
effect
should remain of order one (see the right plot in Figure \ref{decoh_effect_coll}). A
detailed mathematical study of this effect in the case $N \rightarrow
\infty$ will be treated in a subsequent paper. 
\begin{figure}[htbp]
\begin{center}
\includegraphics[width=7.7cm]{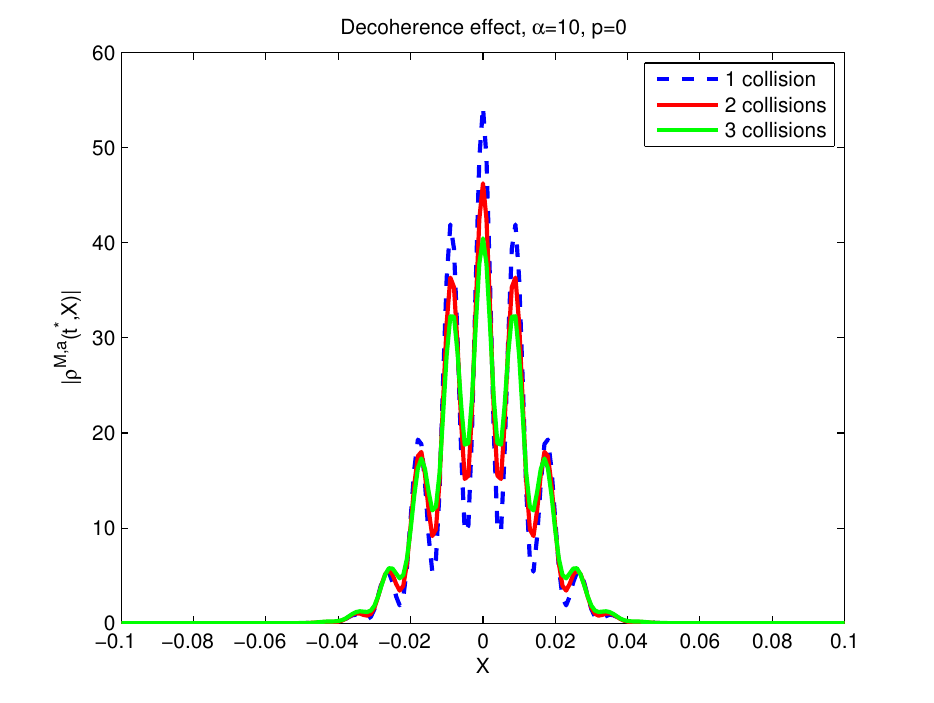}\hfill
\includegraphics[width=7.7cm]{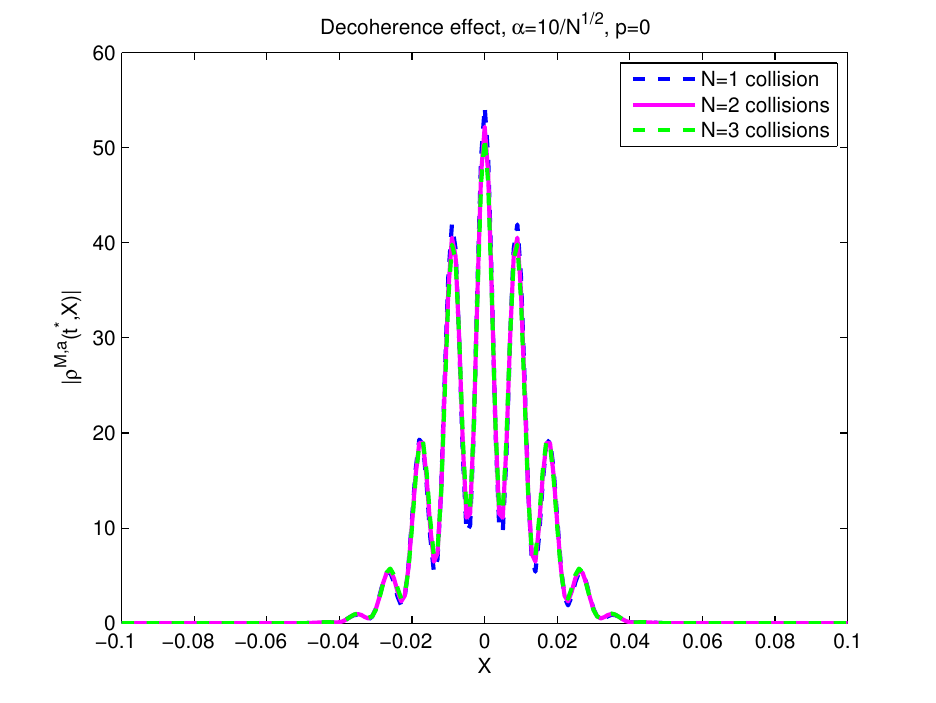}
\end{center}
\vspace{-0.3cm}
\caption{\label{decoh_effect_coll} 
{\footnotesize Attenuation of the interference pattern of the heavy
  particle in the case of several collisions. Test case: Dirac's
  delta potential, $p=0$. Left: Fixed $\alpha=10$ and several
  collisions $N=1,2,3$. Right: $N=1,2,3$ with $\alpha=10/\sqrt{N}$.}} 
\end{figure}

\subsection{Theoretical explanation} \label{thexpla}

Here we propose a theoretical explanation for the plots in Figure
\ref{fringes1}, 
described in Subsection \ref{SEC41} as the decomposition of the
probability density 
associated to the heavy particle into a coherent and a decoherent part.

To this purpose, we first assume
$\sigma p \gg 1$, which means that the light particle must travel fast
enough; as proven in Proposition \ref{prop:Theta_error}, this
assumption 
makes the function $\Gamma_\chi$, defined in \eqref{def:Gamma},
negligible. Besides, owing to this hypothesis, the normalization
constant $N$ defined in \eqref{enne} can be approximated by one.

Second, we 
suppose
$\left\| \f {d |r_k |^2} {dk} \right\|_\infty \ll \sigma$, so that the
variation of the reflection amplitude is slow, and $r_k$ can 
be always considered as equal to $r_p$.

By these assumptions one easily gets
\[
I_\chi(X,X') \approx 1 - \Theta_{\sigma,p}^{app}(X-X') =  1 - |r_p|^2 + |r_p|^2 
e^{2ip(X-X')- \frac{(X-X')^2}{2\sigma^2}},
\]
where $\Theta_{\sigma,p}^{app}$ was defined in \eqref{eq:ThetaApprox}.

Now, the main assumption states that
the following ordering holds between the spatial scales involved in the
collision: 
\begin{equation} \label{ordering}
\sigma_H \ll \sigma \ll X_0.
 \end{equation}
The physical meaning of \eqref{ordering} is
transparent: if $\sigma \ll X_0$, then the incoming light particle
can distinguish the two bumps of the heavy particle from each other; 
furthermore, if
$\sigma_H \ll \sigma$, then  any bump of the heavy particle
approximately acts on
the light particle as a pointwise scattering centre.

From Table~\ref{tab} one has that this condition
is satisfied if one replaces the symbol ``$\ll$'' by ``$<$'';
this fact suggests that the following explanation can hold also under
more relaxed hypotheses.

Thanks to \eqref{ordering}, referring to
the initial density matrix as expressed in \eqref{rhoinitapp}, 
the two diagonal terms
$\varphi_\pm(X) 
\overline{\varphi_\pm(X')}$ 
are essentially supported  in the region  $|X-X'| \ll \sigma$,
while the non-diagonal terms $\varphi_\pm(X) 
\overline{\varphi_\mp(X')}$ 
are essentially supported  in the region  $|X+X'| \ll \sigma$.
Therefore, we can further approximate $I_\chi$ in 
the two regions $\{ |X-X'| \ll \sigma \}$ and $\{ |X-X'| \gg \sigma \}$ by
\[
I_\chi(X,X') \approx \begin{cases}
                      1 - |r_p|^2 & \text{if } \; |X-X'| \gg \sigma,\\
                      1 - |r_p|^2 + |r_p|^2e^{2ip(X-X')} & 
                      \text{if } \; |X-X'| \ll \sigma.
                     \end{cases}
\]
so that, using also $N \approx 1$, one obtains $\rho^{M,a}(0,X,X') 
\approx \rho^{M,b}(0,X,X')$, where 
\begin{align} \label{eq:ApproxDensMat}
\rho^{M,b}(0,X,X') := |t_p|^2  \rho^M_0(X,X') + 
\frac{|r_p|^2}2 &e^{2i pX}\varphi_-(X) \overline{e^{2ipX'}\varphi_-(X')} \\
& + \frac{|r_p|^2}2  e^{2i pX} \varphi_+(X) \overline{e^{2i pX'}\varphi_+(X')}.
\nonumber
\end{align}
Thus, the approximated  initial state 
$ \rho^{M,b}(0)$ can be understood as the \emph{statistical mixing} of
three pure states:  
\begin{itemize}
  \item  the initial pure state 
$\rho^M(0)$ with weight $|t_p|^2$;
  \item the pure state represented by the wave 
function $e^{2ip \cdot}\varphi_-$ with weight $\frac12 |r_p|^2$;
  \item the pure state represented by the wave 
function $e^{2ip \cdot}\varphi_+$
with weight $\frac12 |r_p|^2$. 
\end{itemize}
We remark that the wave functions $ e^{2ip \cdot}\varphi_\pm $ show
the same spatial 
localization as $\varphi_\pm$, respectively, but their momentum has
increased by  
$2 p$. Therefore,
the wave 
function $e^{2ip \cdot}\varphi_-$ ($e^{2ip \cdot}\varphi_+$)
describes the heavy particle localized on 
the left (right) and accelerated by the reflection of the light one.

We are now ready to interpret Figure ~\ref{fringes1}. Let us evolve
$\rho^{M,b}$ according to the free dynamics \eqref{matrixlim}. 
At the time $t^\ast$ of maximal overlap of the two initial bumps, the
first pure state on the r.h.s. of \eqref{eq:ApproxDensMat}
shows the expected interference fringes as in 
Figure~\ref{Init}, but such fringes are damped by the factor
 $ |t_p|^2 $. This
explains the fringes in the picture on the left of 
Figure~\ref{fringes1}. We remark in particular that all the oscillations
reach the zero value, as it occurs when the heavy particle lies in a 
 pure state. At the same time, the pure states corresponding to the
 second and third terms in \eqref{eq:ApproxDensMat}
are also overlapping, since they are both accelerated by the same
quantity. But, since the r.h.s. of \eqref{eq:ApproxDensMat} is a
statistical mixture,  the two bumps will superpose classically,
 without giving rise to
interference fringes. This explains the bump on the right side of 
the left diagram in Figure~\ref{fringes1}. 

In the picture on the right of Figure~\ref{fringes1}, this
interpretation can still hold, but since in that case $p=0$, the
interference fringes of the first state are also superposed with the two
bumps created by the second and third states of the mixing, so the
decomposition ~\eqref{eq:ApproxDensMat} is not so clearly readable.

The content of the present subsection can be made rigorous by proving
that the approximation ~\eqref{eq:ApproxDensMat} holds in the
trace-class norm. This is indeed the case since we have the following
result:
\begin{thm} \label{prop:last_approx}
Let the initial state $\varphi$ of the heavy particle 
 have the form stipulated in~\eqref{initwf}-\eqref{enne},
 and let the incoming state $\chi$ of the light particle be chosen as in 
\eqref{gauschi}.
Define
the density matrix $\rho^{M,a} (0)$ 
as in \eqref{SCHRhoScatLim-rid}, 
 and let  $\rho^{M,b} (0)$ denote
the density matrix with integral kernel having the form\eqref{eq:ApproxDensMat}.

Then, the following estimate holds:
\be 
\| \rho^{M,a} (0) - \rho^{M,b} (0) \|_{\mt L^1} \leq C \left( e^{- \sigma^2 p^2} +
\f 1 \sigma \left\| \frac{d|r_k|^2} {dk}   \right\|_\infty
+ \frac {\sigma_H} \sigma + e^{- \f {2 X_0^2} {\sigma^2} }
+ e^{- \f {X_0^2} {2 \sigma_H^2} }
\right).
 \ee
\end{thm}

\begin{remark}{\em
As an immediate consequence, if $$
\sigma_H \ll \sigma \ll X_0, \quad  \left\| \frac{d|r_k|^2} {dk}
\right\|_\infty \ll \sigma, \quad \text{and} \quad \f 1 \sigma \ll p, 
$$
then the difference between $\rho^{M,a}$ and $\rho^{M,b}$ is small.}
\end{remark}

The proof of Theorem~\ref{prop:last_approx} is given in
Section \ref{proof3}.

\begin{remark} [Entanglement between the two particles] \label{Rk:entang}
{\em 
Under the same hypotheses of Theorem~\ref{prop:last_approx},
one can work out a simpler expression for the initial two-particle
wave function than the one given in Theorem \ref{thm:geneS}. First,
the assumption $\| \f {dr_k} {dk} \|_\infty \ll \sigma$ gives
\[
S \chi \approx S^{app} \chi_\ep := t_p \, \chi_\ep +  r_p \, R_0 
\chi_\ep,
\]
where $\chi_\ep = U (-\ep^{-\gamma}) \chi$, and
the reflection operator $R_0$ is defined by 
$R_0 \chi_\ep (x) := \chi_\ep (-x)$ 
(its action is invariant under Fourier transformation).
Second, by translational invariance
\begin{equation} \label{Scatapp}
S^{X,app} \chi_\ep := t_p \,\chi_\ep + r_p  \,\theta_{2X} \, R_0 
\chi_\ep,
\end{equation}
where $\theta_{2X} u = u (\cdot- 2X)$.
Thus, after some computations (that can be performed more easily in the Fourier space), 
one can replace the initial condition of the limit 
equation~\eqref{SCHlimscat} by
\begin{multline} \label{entang}
\varphi \otimes S^X \chi_\ep \approx 
\varphi \otimes S^{X,app} \chi_\ep \approx
t_p\, \varphi \otimes \chi_\ep
+ \frac{r_p}{\sqrt 2}  e^{2ipX} \varphi_- \otimes  \bigl( e^{-2i p X_0}\theta_{-2X_0} \, R_0 \chi_\ep \bigr) \\
+ \frac{r_p}{\sqrt 2}  e^{2ipX} \varphi_+ \otimes  \bigl( e^{2i p X_0} \theta_{2X_0} \, R_0 \chi_\ep \bigr).
\end{multline}
The new phase factors come from the approximation of $\theta_{2X}$ by $\theta_{2X_0}$  or 
$\theta_{-2X_0}$. Remark that the phase factor on the light particle is constant and thus not important; on the contrary, the phase factor on the heavy particle means that it is accelerated by $2p$.
The two-body wave function in~\eqref{entang} represents an entangled
state: the light particle 
is transmitted when the heavy particle remains in its initial state, it 
is reflected from $- X_0$ when the heavy particle is located at $- X_0$ (and accelerated) and 
so on. Since the three states of the light particle that appear in the previous approximation are almost orthogonal under the assumptions of Theorem~\ref{prop:last_approx}, the associated density matrix for the heavy particle
turns out to be well approximated by~\eqref{eq:ApproxDensMat}.
}

\end{remark}

\section{Proofs} \label{appa}
The present section contains the proofs of the approximation
theorems presented in Section \ref{theorem}. In particular, Section
\ref{P1} deals with Theorem~\ref{thm:geneS} and section \ref{P2} with
Theorem  \ref{thm:geneRho}. 

\subsection{Proof of Theorem~\ref{thm:geneS}} \label{P1}

We preliminarily warn the reader that part of this
  section is devoted to the proof of results that are analogous
to those contained in
  Theorem 1 in \cite{adami_3}. We include this section
  anyway, both for the sake of completeness and because the results we
  need are slightly different from the one in
  \cite{adami_3}. All proofs presented here are new.

\medskip
\paragraph{\bf The reduced variables and a useful lemma.}

Let us first introduce the centre of mass $R$ and the relative
position $r$ of the two-body problem. We define
\[
 R : = \f {X + \ep x} {1 + \ep}, \qquad r : = x - X, 
 \]
or, equivalently,
\[ 
X : = R - \f {\ep x} {1 + \ep}, \qquad x = R + \f r {1 + \ep}.
\]
The new variables naturally induce a unitary transformation 
on $L^2 (\RR^d)$, given by
\[
(\T_\ep \psi) (R, r) \ : = \ \psi \left(  R - \f {\ep x} {1 + \ep},  R +
\f r {1 + \ep} \right), \qquad 
(\T_\ep^{-1} \phi) (X,x)  \ : = \ \phi \left(  \f {X + \ep x} {1 +
  \ep},  x - X \right).
\]
The previous definition can be extended to the case $\ep = 0$. 
The following lemma compares  $\T_\ep \psi$ and $\T_0 \psi$.

\begin{lem} \label{lem:TepT0}
For any $\psi \in L^2(\RR^{2d})$ s.t.\  $(x-X) \cdot (\nabla_X +
\nabla_x) \psi \in L^2 (\RR^{2d})$, we have the following estimate
\begin{align} \label{bea}
\| \T_\ep \psi - \T_0 \psi \|_2 & \le   \ep \,  \| (x-X) \cdot (\nabla_X + \nabla_x) \psi \|_2 \
 \le \ \ep \, \| r \cdot \nabla_R \T_0 \psi \|_2\,. \nonumber
\end{align}
\end{lem}

\begin{proof}
 
Denoting $\phi = \T_0 \psi$ and
\[
 \hat \phi (k,r) \ = \ \f 1 {(2 \pi)^{\f d 2}} \int_{\RR^d} e^{-i k
   \cdot R} \phi (R,r) \, dR
\]
one has
\be \begin{split}
\| \T_\ep \psi - \T_0 \psi \|_2^2 & \le  
\| \T_\ep  \T_0^{-1} \phi - \phi \|^2_2  \\
& =  \int_{\RR^{2d}} \left| \phi \left(R - \f {\ep r} {1 + \ep}, r \right) - \phi (R,r)
\right|^2  dR \, dr\\
& =   \int_{\RR^{2d}} \left| e^{-i k \cdot \f {\eps r} {1 + \eps}}
  - 1  \right|^2 | \hat \phi (k,r) |^2 \, dk dr 
\\
& \leq \ep^2 \| r \cdot \nabla_R \phi \|_2^2,
\end{split} \nonumber \ee
and we get the claimed inequality using that $\nabla_R\, \phi(R,r) =
(\nabla_X + \nabla_x) \psi(R,R+r)$. 
\end{proof}

Moreover, $\T_0$ has the following property
\be \label{commut_T0S0}
\T_0 \, \widehat{\S} = [ \mathbb I \otimes S ]\,  \T_0,
\ee
which is a consequence of the definition of $\widehat{\S}$ (see
Definition~\ref{def:Sepp}). In that definition
the action of $\widehat{\S}$ includes 
the scattering of the light particle by a potential centred at the location of the
heavy particle, while in the reduced variables, the scattering takes place
in the relative position variable.
We will also use the following elementary identity, which may be proved directly.

\begin{lem} \label{lem:CommxU}
For all $j=1,\ldots,d$, 
any $\tau \in \RR$ and any $\chi \in L^2(\RR^d)$ such that $|x| \chi \in L^2(\RR^d)$
\be \label{easy} x_j U_0(\tau) \chi \ = \   U_0(\tau) \bigl[ i \tau
\partial_j \chi +  x_j \chi \bigr].
\ee
This implies in particular that 
\be \label{estim:CommxU}
\| \,|x| U_0(\tau) \chi \|_2 \ \le  \ \sqrt 2 \bigl[    \| \,
  |x| \chi \|_2 + \tau \| \nabla \chi \|_2 \bigr].
\ee
\end{lem}

\medskip
\paragraph{\bf Step 1. Rewriting the problem in reduced variables.}

Let $\psi_\eps$ be the solution to \eqref{SCH2D}, \eqref{due} with
$M=1$.
Denoting
\be \label{def:tpsi}
\tpsi_\ep:= \T_\ep \psi_\ep,\qquad \tpsi_\ep^a:= \T_\ep \psi_\ep^a,
\ee
one has that $\tpsi_\ep$ and $\tpsi_\ep^a$ are respectively solutions to
\be \label{SCH2Dbis} \begin{cases}
i \, \partial_t \tpsi_\ep=- \frac1{2(1+\ep)} \Delta_R \tpsi_\ep +
\frac{1+\ep}\ep  \Bigl( - \frac12 \Delta_r \tpsi_\ep +  V(r)\, \tpsi_\ep \Bigr),
\\
\tpsi^\ep(0) =  \T_\ep \psi^0_\ep =  \T_\ep \bigl[ \mathbb I \otimes
U_0(-\ep^{-\gamma}) \bigr] \psi,
\end{cases}
\ee
and
\be \label{SCH2Dlimbis} \begin{cases}
i \, \partial_t \tpsi_\ep^a=- \frac1{2(1+\ep)} \Delta_R \tpsi_\ep -
\frac{1+\ep}{2
\eps} \Delta_r \tpsi_\ep, \\
\tpsi_\ep^a(0)
=  \T_\ep \S_\ep \psi = \T_\ep [\mathbb I \otimes U(-\ep^\gamma)] \widehat{\S} \psi ,
\end{cases}
\ee
where $\psi := \varphi \otimes \chi$. Notice that in problem
\eqref{SCH2Dbis} the variables $R$ and $r$ are decoupled,
therefore we can express the solution in terms 
of semigroups acting separately on $R$ and $r$, i.e.
\begin{align} \label{def2:tpsi}
\tpsi_\ep(t) & =
\left[ U_0\left({\textstyle \frac t{1+\ep}} \right) \otimes U_V
\left({\textstyle\frac {(1+\ep)t}\ep }\right) \right] \T_\ep
\bigl[ \mathbb I \otimes
U_0(-\ep^{-\gamma}) \bigr] \psi, \\
\tpsi_\ep^a(t) & =
\left[ U_0\left({\textstyle \frac t{1+\ep}} \right) \otimes U_0
\left({\textstyle\frac {(1+\ep)t}\ep }\right) \right] \T_\ep
\bigl[ \mathbb I \otimes
U_0(-\ep^{-\gamma}) \bigr] \widehat{\S} \psi  \label{def2:tpsia} .
\end{align}

\noindent
In order to estimate the distance between $\psi_\ep (t)$ and
$\psi^a_\ep (t)$ we introduce two intermediate terms
$\psi_\ep^b(t)$ and $\psi_\ep^c(t)$, defined as follows 
\begin{align} \label{def:tpsib}
\tpsi_\ep^b(t) & =
\left[ U_0\left({\textstyle \frac t{1+\ep}} \right) \otimes U_V
\left({\textstyle\frac {(1+\ep)t}\ep }\right) \right] \T_0
\bigl[ \mathbb I \otimes
U_0(-\ep^{-\gamma}) \bigr] \psi \\
\tpsi_\ep^c(t) & =
\left[ U_0\left({\textstyle \frac t{1+\ep}} \right) \otimes U_0
\left({\textstyle\frac {(1+\ep)t}\ep }\right) \right] \T_0
\bigl[ \mathbb I \otimes
U_0(-\ep^{-\gamma}) \bigr] \widehat{\S} \psi. \label{def:tpsic}
\end{align}
Then,
\be
\| \psi_\ep(t) - \psi_\ep^a(t) \|_2 
 \le \|\tpsi_\ep(t) - \tpsi_\ep^b(t) \|_2 +
\| \tpsi_\ep^b(t) - \tpsi_\ep^c(t) \|_2 +
\| \tpsi_\ep^c(t) - \tpsi_\ep^a(t) \bigr \|_2. \label{splitt}
\ee
A control of $\tpsi_\ep(t) - \tpsi_\ep^b(t)$ may be obtained thanks to 
Lemmas~\ref{lem:TepT0} and~\ref{lem:CommxU}, and the same lemmas together with 
the hypothesis (H3) allows us to control $\tpsi_\ep^c(t) - \tpsi_\ep^a(t)$. This will be explained in Step $2$.
To control the term $\tpsi_\ep^b(t) - \tpsi_\ep^c(t)$, we will use the commutation properties of $\widehat{\S}$. 
We explain in Step 3 how it leads to the term involving $C_1$ in the
estimate~\eqref{estim:4C}. 

\medskip
\paragraph{\bf Step 2. The approximation of infinitely massive particle.}
In fact, the replacement of $\T_\ep$ by $\T_0$ is equivalent to the approximation 
that the massive particle has an infinite mass, so that it does not move during 
the evolution of the light one.
Using Definition~\ref{def:tpsib}, the unitarity 
of $U_0$ and $U_V$, and Lemma~\ref{lem:TepT0}  we get 
\begin{align*}
 \| \tpsi_\ep(t) - \tpsi_\ep^b(t) \|_2 & = \| (\T_\ep - \T_0) [\mathbb
   I \otimes U_0(-\ep^\gamma)] \psi \|_2 \\
 & \le  \ep \, \| (x-X) \cdot(\nabla_X+ \nabla_x) [\mathbb I \otimes U_0(-\ep^\gamma)] \psi \|_2 \\
 & \le  \ep  
 \| (x-X)  \cdot (\nabla \varphi \otimes  U_0(-\ep^{-\gamma})  \chi 
+ \varphi \otimes  U_0(-\ep^{-\gamma})  \nabla \chi)
 \|_2, \\
\ep^{-1} \| \tpsi_\ep(t) - \tpsi_\ep^b(t) \|_2 & \le   \| \nabla
\varphi \|_2 \|  \, | \cdot |  U_0(-\ep^{-\gamma})  
\chi \|_2
+ \| X \cdot \nabla \varphi \|_2  \\ 
& \hspace{3cm }+   \|   x \cdot  U_0(-\ep^{-\gamma})
\nabla \chi \|_2  +  \|  X \varphi \|_2   \| \nabla  \chi \|_2,
\end{align*}
where we used the fact that $U_0$ commutes with derivatives and
that $\|\varphi\|_2=\|\chi\|_2=1$. 
Using Lemma~\ref{lem:CommxU}, one can get rid of the propagators $
U_0(-\ep^{-\gamma})$ in the previous  
estimate, namely
\bea
\ep^{-1} \| \tpsi_\ep(t) - \tpsi_\ep^b(t) \|_2 &\leq & \sqrt 2 
\| \nabla \varphi\|_2 \bigl( \| | \cdot|    \chi \|_2 +  \ep^{-\gamma} \|\nabla \chi 
\|_2 \bigr) + \| X \cdot \nabla \varphi \|_2  \nonumber \\
&& \hspace{1cm} + \,  \sqrt 2 \bigl( \| x \cdot    \nabla \chi \|_2 + \ep^{-\gamma} \|
\Delta \chi \|_2 \bigr) + \| | \cdot | \varphi \|_2 \| \nabla \chi
\|_2, \nonumber 
\eea
so that eventually
\begin{align} \label{estim:tpsiepb}
\| \tpsi_\ep(t) - & \tpsi_\ep^b(t)   \|_2 \leq   K_1 \ep +  K_2 \ep^{1-\gamma},   \\
\text{with} \qquad & K_1 :=  \sqrt2 \bigl( \| \nabla \varphi\|_2 \| |x|    \chi \|_2 
+ \| x \cdot    \nabla \chi \|_2  \bigr) +
\| | X | \varphi \|_2 \| \nabla \chi \|_2 +  \| X \cdot \nabla \varphi \|_2,
\nonumber \\
& K_2 := \sqrt 2  \bigl( \| \nabla \varphi\|_2  \|\nabla \chi\|_2.
+ \| \Delta \chi \|_2  \bigr)\,.  \nonumber
\end{align}

\medskip
Similarly, from Definitions~\eqref{def2:tpsia}
and~\eqref{def:tpsic} one gets
\begin{align*}
 \| \tpsi_\ep^c(t) - \tpsi_\ep^a(t) \|_2 & = \| (\T_\ep - \T_0)
    [\mathbb I \otimes U_0(-\ep^\gamma)] \widehat{\S} \psi \|_2 \\ 
 & \leq \ep \, \| r \cdot [\mathbb I \otimes U_0(-\ep^\gamma)] [\mathbb I  \otimes S]  \nabla_R \T_0  \psi \|_2,
 \end{align*}
where we have used that $U_0$ commutes with translation, the relation~\eqref{commut_T0S0}, 
and the fact that $\nabla_R$ commutes with $\mathbb I \otimes S$ and
$\mathbb I \otimes U_0(\tau)$.  
Applying 
Lemma~\ref{lem:CommxU} (integrated on $R$) to the function 
$\bar \phi := [\mathbb I \otimes U_0(-\ep^\gamma)] [\mathbb I  \otimes S]   
\nabla_R  \T_0  \psi$, we get 
\begin{align*}
2^{-1/2}\| \tpsi_\ep^c(t) - \tpsi_\ep^a(t) \|_2 & \le \ep \, 
 \| r \cdot [\mathbb I  \otimes S]   
 \T_0 (\nabla_X + \nabla_x) \psi \|_2 + \ep^{1 -\gamma}
 \|  \nabla_r [\mathbb I  \otimes S]   
 \T_0 (\nabla_X + \nabla_x) \psi \|_2,
\end{align*}
where in the last line we used that $\psi = \varphi \otimes \chi$. In
order to bound the second term in the r.h.s we can use the
conservation of the kinetic 
energy under the action of $S$ and get
\begin{align*}
\|  \nabla_r [\mathbb I  \otimes S]   
 \T_0 (\nabla_X + \nabla_x) \psi \|_2 & = 
\|  \T_0 \nabla_x (\nabla_X + \nabla_x) \psi \|_2 \\
 & \le \| \nabla \varphi \|_2 \| \nabla \chi \|_2 + \| \Delta \chi \|_2.
\end{align*}
Using the regularity assumption (H3) on the
scattering operator to bound the first term of the r.h.s. we get
\begin{align*}
\| r \cdot [\mathbb I  \otimes S]  \T_0 (\nabla_X + \nabla_x) \psi \|_2 & \le 
\| |r|  \T_0 (\nabla_X + \nabla_x) \psi \|_2 +
C_s \| \T_0 (\nabla_X + \nabla_x) \psi \|_{L^2_R(H^s_r)} \\
& \le 
\| |x-X| (\nabla_X + \nabla_x) \psi \|_2 +
C_s \bigl( \| \psi \|_{L^2_X(H^{s+1}_x)} 
+ \| \nabla_X \psi \|_{L^2_X(H^s_x)}
\bigr)\\
& \le  \| \nabla \varphi\|_2 \| |x|    \chi \|_2 +\| \,|X| \nabla \varphi \|_2  
+ \| \,|x| \nabla \chi \|_2   +  \| | X | \varphi \|_2 \| \nabla \chi \|_2 \\
& \hspace{3cm} + C_s \bigl(  \| \nabla \varphi\|_2 \| \chi\|_{H^s}+  \| \chi \|_{H^{s+1}}   \bigr),
\nonumber 
\end{align*}
where we used the fact that $\psi_0 = \varphi \otimes \chi$ is factorized.
Putting all together, we get
\begin{align} \label{estim:tpsica}
\| \psi_\ep^c(t) - & \psi_\ep^a(t) \|_2 \le K_3 \ep  + K_4 \ep^{1 -\gamma},\\ 
\text{with }&
2^{-1/2} K_3 := \| \nabla \varphi\|_2 \| |x|    \chi \|_2 +\| \,|X| \nabla \varphi \|_2  
+ \| \,|x| \nabla \chi \|_2   +  \| | X | \varphi \|_2 \| \nabla \chi \|_2  \nonumber \\
& \hspace{3cm}
+ C_s \bigl( \|\nabla \varphi \|_2 \| \chi \|_{H^s} +  \|
\chi \|_{H^{s+1}} \bigr), \nonumber \\
& 2^{-1/2}  K_4 := \| \nabla \varphi \| \| \nabla \chi \| + \| \Delta \chi \|. \nonumber
\end{align}

\paragraph{\bf Step 3. The approximation of a fast scattering.}
Now we estimate $\| \widetilde \psi_\ep^b (t) -  \widetilde \psi_\ep^c
(t) \|$. Starting from Definitions~\eqref{def:tpsib}  
and ~\eqref{def:tpsic}, using
the fact that $\T_0$ commutes with  
$ \mathbb I \otimes U_0(t)$ and the relation~\eqref{commut_T0S0}, we obtain
\begin{align*}
\tpsi_\ep^b(t) & =
\left[ U_0\left({\textstyle \frac t{1+\ep}} \right) \otimes U_V
\left({\textstyle\frac {(1+\ep)t}\ep }\right) \right] 
\bigl[ \mathbb I \otimes
U_0(-\ep^{-\gamma}) \bigr] \T_0 \psi \\
& = \left[ U_0\left({\textstyle \frac t{1+\ep}} \right) \otimes U_0
\left({\textstyle\frac {(1+\ep)t}\ep } -\ep^{-\gamma} \right) \right]  
\bigl[ \mathbb I \otimes S(\tau,\tau') \bigr] \T_0 \psi \\
\text{and} \quad \tpsi_\ep^c(t) & =
\left[ U_0\left({\textstyle \frac t{1+\ep}} \right) \otimes U_0
\left({\textstyle\frac {(1+\ep)t}\ep } -\ep^{-\gamma} \right) \right] 
 [\mathbb I \otimes S ]  \T_0 \psi,
\end{align*}
where we introduced $\tau = \ep^{-\gamma}$ and $\tau'= \frac {(1+\ep)t}\ep -\ep^{-\gamma}$.
Using the unitarity of $U_0$, we get 
\begin{align}
 \| \tpsi^b_\ep(t) -  \tpsi^c_\ep(t) \|_2 & = \| \mathbb I \otimes [ S(\tau,\tau') - S] 
 \T_0 \psi \|_2 \nonumber \\
 & = \| \varphi  [ S(\tau,\tau') - S] \chi(\cdot -X) \|_2. \label{estim:tpsibc}
\end{align}

\paragraph{\bf Conclusion.}
Putting together \eqref{splitt}, \eqref{estim:tpsiepb}, \eqref{estim:tpsica}, and
 \eqref{estim:tpsibc} the proof is complete.

\bigskip

\subsection{Proof of Theorem \ref{thm:geneRho}} \label{P2}

We preliminarily recall that the initial density operator $\rho^M (0)$ of the heavy
particle (see \eqref{SCHRhoScat})
is a compact, positive, self-adjoint
operator whose trace equals one. Thus
by the spectral theorem there exists a sequence $0 \leq \lambda_j
\leq 1$, $\sum_j \lambda_j = 1$, and a complete orthonormal set
$ | \varphi_j \rangle \in L^2 (\RR^d)$ such that 
\be \label{eq:decomprho}
\rho^M (0) \ = \ \sum_j \lambda_j   | \varphi_j \rangle \langle  \varphi_j | .
\ee

\paragraph{\bf Estimate of the difference of the two-body density
  operators  $\rho_\ep(t)$ and  $\rho_\ep^a (t)$}

We recall from Section \ref{theorem} that the two-body density
operator $\rho_\ep$ is the solution to the operator equation
\[
 i \partial_t \rho_\ep(t) = [H_\ep, \rho_\ep(t)]
 \]
with initial data
\begin{equation*} \begin{split}
\rho_\ep(0) \ & = \ \rho^M(0) \otimes |  U(-\ep^{-\gamma}) \chi \rangle
\langle U(-\ep^{-\gamma}) \chi | \ = \ \sum_{j} \lambda_j | \psi_{j,
  \eps} (0) \rangle \langle \psi_{j, \eps} (0) |, \\ &   | \psi_{j,
  \eps} (0) \rangle : = | \varphi_j \rangle | U_0 (\eps^{-\gamma})
\chi \rangle,
\end{split} \end{equation*}
where we applied the decomposition in\eqref{eq:decomprho}. Therefore,
\[
 \rho_\eps (t) = \sum_j \lambda_j | \psi_{j, \eps} (t) \rangle
\langle  \psi_{j, \eps} (t) |, 
\]
where $ \psi_{j, \eps} (t)$ is the solution to eq. \eqref{SCH2D} with
initial data $ \psi_{j, \eps} (0)$.

\noindent
Analogously, the two-body density operator $\rho^a_\ep$ is the
solution to
\[
i \partial_t \rho_\ep^a(t) = [H_\ep^f, \rho_\ep^a(t)],
\]
with initial data 
\begin{equation*} \begin{split}
\rho_\ep^a(0) & =  
\S_\ep \bigl[\rho^M(0) \otimes | \chi \rangle 
\langle \chi | \bigr]  \S_\ep^\ast = \sum_j \lambda_j | \psi_{j,
  \eps}^a (0) \rangle \langle \psi_{j, \eps}^a (0) |, \\ &    | \psi_{j,
  \eps}^a (0) \rangle : = | \varphi_j \rangle | U_0 (\eps^{-\gamma}) S^X
\chi \rangle,
\end{split} \end{equation*}
where we applied decomposition \eqref{eq:decomprho}. Then,
$$ \rho_\eps^a (t) = \sum_j \lambda_j | \psi_{j, \eps}^a (t) \rangle
\langle  \psi_{j, \eps}^a (t) |, $$
where $ \psi_{j, \eps}^a (t)$ is the solution to eq. \eqref{SCH2D} with
initial data $ \psi_{j, \eps}^a (0)$.

\noindent
Let us estimate the distance between $\rho_\ep(t)$ and
$\rho_\ep^a$. We get
\begin{align*}
\|  \rho_\ep (t) -  \rho_\ep^a (t)  \|_{\L^1} &\le \sum_j \lambda_j   \bigl\| |\psi_{j,\ep}(t) \rangle \langle  \psi_{j,\ep}(t)| - | \psi_{j,\ep}^a(t) \rangle \langle  \psi_{j,\ep}^a(t) | \bigr\|_{\L^1} \\
& \le 2 \sum_j \lambda_j   \bigl\| \psi_{j,\ep}(t) -  \psi_{j,\ep}^a(t) \bigr\|_2\,,
\end{align*}
where we have used the fact that  for any $\zeta_1, \zeta_2$ in $L^2
(\RR^{2d})$ with $\| \zeta_1 \|_2 = \| \zeta_2\|_2=1$ 
$$
\Tr \bigl|   | \zeta_1 \rangle \langle  \zeta_1
|  - | \zeta_2 \rangle \langle  \zeta_2
| \bigr| \leq \ 2 \| \zeta_1 - \zeta_2 \|_2.
$$
It remains to sum up the error bounds given by Theorem
\ref{thm:geneS}. One gets
\begin{align} 
 \nonumber \frac12 \| \rho_\ep (t) -   \rho_\ep^a (t) \|_{\L^1} 
& \leq     \sum_j \lambda_j \| e^{-itH_\ep}   | \varphi_j \rangle
|U_0 (-\ep^{-\gamma}) \chi \rangle -  U_\ep^f (t)  | \varphi_j \rangle
|U_0 (-\ep^{-\gamma}) S^{X} \chi \rangle  \| \\
& \leq    2 \sqrt 2 \| \Delta \chi \|_2
\ep^{1-\gamma} + \sqrt 2 C_s  \ep \|
\chi \|_{H^{s+1}}     \nonumber \\ 
& \hspace{2cm} +\sum_j \lambda_j \biggl[ C_{1,j} \Bigl( {\textstyle \frac{1+\ep} \ep} \,t -
\ep^{-\gamma}, \ep^{-\gamma}  \Bigr) +
 C_{2,j}  \ep+ C_{3,j}  \ep^{1-\gamma}
\bigg] , \label{pretrace}
\end{align}
where 
\bea 
C_{1,j}(\tau,\tau')  &:= &\left\| \varphi_j  \left[ S(\tau,\tau') -  S
\right]\chi(\cdot -X) \right\|_2 ,
\label{def:C1j}  \\
\label{def:C2j}
 C_{2,j} &:=& \nonumber 
2\sqrt2 \bigl( \| \nabla \varphi_j\|_2 \| \, |x|  \chi \|_2 + 
 \| X \varphi_j \|_2 \| \nabla \chi \|_2 +  \| \,|X|  \nabla \varphi_j
 \|_2 +  \|  \,|x|    \nabla \chi \|_2 \bigr)\\ 
&& \hspace{77mm} + \, C_s 
\|\nabla \varphi_j \|_2 \| \chi \|_{H^s}, \\
\label{def:C3j}
 C_{3,j} &:=&  2 \sqrt 2
\| \nabla \varphi_j\|_2  \|\nabla \chi\|_2 .
\eea
Summing up in all constants with respect to $j$ and using
Cauchy-Schwarz inequality leads to the error bounds given by Theorem
\ref{thm:geneS}. For instance, 
\begin{eqnarray*}
\sum_j \lambda_j C_{1,j} (\tau, \tau') & = & \sum_j \lambda_j
\left\| \varphi_j  \left[ S^X(\tau,\tau') -  S^X
\right] 
\chi  \right\| \\
& \leq & \left( \sum_j \lambda_j   \right)^{\f 12}  \left(
\sum_j \lambda_j   \left\| \varphi_j   \left[ S^X(\tau,\tau') -  S^X
\right]  \chi  \right\|^2  \right)^{\f 12} \\
 & = &  \left(
 \, \Tr  \left[ \,\sum_j \lambda_j  \bigl| \varphi_j   \left[ S^X(\tau,\tau') -  S^X
\right] \chi \bigr\rangle \bigl\langle  \varphi_j   \left[
     S^X(\tau,\tau') -  S^X 
\right]  \chi \bigr| \right] \right)^{\f 12}  
 \\
 & = &
2 \bigl\| \rho^M(0)  | [S(\tau,\tau') - S] \chi(\cdot -X') \rangle 
\langle  [S(\tau,\tau') - S] \chi(\cdot -X) \bigr\|_{\L^1}^{\frac12},
\end{eqnarray*}
and analogously 
\begin{align*}
 \sum_j \lambda_j \| X \cdot \nabla \varphi_j \|_2 & \le
 \left( \sum_j \lambda_j\right)^{\frac12}
 \left( \sum_j \lambda_j \| |X| \nabla \varphi_j \|_2^2 \right)^{\frac12} 
\\
  & \le \Bigl[ \Tr  \bigl( |X|  \, i\nabla \rho^M(0) \, i\nabla |X| \bigr) \Bigr]^{\frac12}.
\end{align*}
The others terms may be handled analogously.
Due to \eqref{stimatraccia} the same estimate as \eqref{pretrace} holds 
for $\| \rho_\ep^M - \rho_\ep^{M,a}\|_{\L^1}$.
It only remains to recall that $\rho_\ep^{M,a} = U_0 (t) \rho^M (0)
U_0 (-t)$ is indeed independent of
$\ep$.

\medskip
\paragraph{\bf The dynamics of the density
  operator $\rho_\ep^{M,a}$.} 

From the fact that $\rho^a_\ep$ is the solution to $i \partial_t
\rho^a_\ep := [H^f_\ep, \rho^a_\ep]$, using the notation $\bar
t(\ep) := \frac{1+\ep}\ep t$, we get
$$
\rho^a_\ep(t) = [U_0(t) \otimes U_0( \bar t(\ep))] \rho_\ep^a(0)
    [U_0(-t) \otimes U_0( -\bar t(\ep))]. 
$$
Choosing a basis $(\chi_i)_{i \in \mathbb N}$ of $L^2(\RR^d)$ one gets  by definition of the partial trace
\begin{align*}
\rho^a_\ep(t) & = \sum_i  \langle U_0( \bar t(\ep)) \chi_i |
\rho^a_\ep(t) | U_0( \bar t(\ep)) \chi_i \rangle \\ 
& = \sum_i  \langle U_0( \bar t(\ep)) \chi_i | [U_0(t) \otimes U_0(
  \bar t(\ep))] \rho_\ep^a(0) [U_0(-t) \otimes U_0( -\bar t(\ep))] |
U_0( \bar t(\ep)) \chi_i \rangle \\ 
& = \sum_i  \langle \chi_i | [U_0(t) \otimes \mathbb I] \rho_\ep^a(0) [U_0(-t) \otimes \mathbb I] | \chi_i \rangle \\
& =  U_0(t)  \Bigl[\sum_i  \langle \chi_i |\rho_\ep^a(0)  | \chi_i \rangle \Bigr] U_0(-t)\\
& = U_0(t) \rho_\ep^{M,a}(0) U_0(-t).
\end{align*}
This implies that $\rho_\ep^{M,a}$ is a solution to the free transport
equation. Then, it remains to identify the initial condition. One finds 
\begin{align*}
\rho_\ep^a(0) &:= 
 [\mathbb I \otimes U_0(-\ep^{-\gamma})] \widehat{\S} \rho^M(0) \otimes | \chi \rangle \langle \chi \widehat{\S}^\ast [\mathbb I \otimes U_0(+\ep^{-\gamma})],
\\
\rho_\ep^{M,a}(0)&= \Tr \bigl[ \widehat{\S} \bigl(\rho^M(0) \otimes | \chi \rangle \langle \chi | \bigr) \widehat{\S}^\ast  \bigr] \\
& = \Tr \bigl[\rho^M(0) \otimes | S^X \chi \rangle \langle S^{X'}\chi | \bigr].
\end{align*}
In terms of kernels, the last identity can be expressed as 
$$\rho_\ep^{M,a}(0,X,X') =
\rho^M(0,X,X') \langle S^X \chi | S^{X'} \chi \rangle
= \rho^M(0,X,X') I_\chi(X,X').
$$

\subsection{Proof of Theorem 
\ref{prop:last_approx} } \label{proof3}

First, we can cut the error into three parts
\begin{align*}
 \| \rho^{M,a} (0,X,X') - \rho^{M,b}  (0,X & ,X') \|_{{\mt L}_1} 
\leq 
\| ( i \Gamma_\chi (X) -  i \Gamma_\chi (X') ) \rho^M
 (0,X,X')\|_{{\mt L}_1} \\
& +  \| (\Theta^{app}_{\sigma, p} (X - X') - \Theta_{\chi} (X - X') )
 \rho^M (0,X,X')  \|_{{\mt L}_1} \\
& +  \| (1 - \Theta^{app}_{\sigma, p} (X - X') ) \rho^M (0,X,X')  
 - \rho^{M,b} (0,X,X') \|_{{\mt L}_1} \\ 
 &    \le  (I) +  (II) +  (III).
\end{align*}
The terms $(I)$ and $(II)$ are easily  estimated using
Proposition~\ref{prop:Theta_error}. We get
\begin{equation} \label{boundI+II}
(I)  \le 2 e^{-2 \sigma^2 p^2} \quad \text{and} \quad
(II)  \le \sqrt{\frac 2 {\pi\sigma^2}  } \, \left\| \frac {d |r_k|^2} {dk}
\right\|_\infty\,.
\end{equation}

\def\tvarphi{\widetilde \varphi}

It remains to estimate $(III)$. Denoting
$\widetilde \varphi_\pm (X) : =  e^{2ipX}  \varphi_\pm (X)$, we may separate 
$(III)$ into
\begin{align*}
(III) & =  \biggl\|  | r_p |^2  e^{2 i p (X -X') - \f 
{(X - X')^2}{2 \sigma^2}} \rho^M (0,X,X') - \f {|r_p|^2} 2 \bigl[
  \tvarphi_- (X)  {\overline{\tvarphi_- (X')}}  -   \tvarphi_+ (X)  
\overline{\tvarphi_+ (X')} \bigr] \biggr\|_{{\mt L}_1}
\\ &
\leq \f {|N^2-1|}2 \,  | r_p |^2  \left\| e^{2 i p (X -X') - \f 
{(X - X')^2}{2 \sigma^2}} ( \varphi_+ (X) + \varphi_- (X') )
( \overline{ \varphi_+ (X')} + \overline{\varphi_- (X')})
\right\|_{{\mt L}_1} \\ 
&  \hspace{10mm } + \f  {| r_p |^2} 2  \left\|   \tvarphi_+ (X)
\overline{\tvarphi_+ (X')} \left( 1 - e^{- \frac {( X - X' )^2}{2
    \sigma^2}} \right)\right\|_{{\mt L}_1}  \\
&  \hspace{10mm } + \f  {| r_p |^2} 2  \left\|  \tvarphi_- (X)
\overline{\tvarphi_- (X')} \left( 1 - e^{- \frac {( X - X' )^2}{2
    \sigma^2}} \right)
\right\|_{{\mt L}_1} \\
&  \hspace{10mm } + \f  {| r_p |^2} 2  \left\|  \tvarphi_+ (X)
\overline{\tvarphi_- (X')}  e^{- \frac {( X - X' )^2}{2
    \sigma^2}} \right\|_{{\mt L}_1}
+ \f  {| r_p |^2} 2  \left\| \tvarphi_- (X) 
\overline{\tvarphi_+ (X')}  e^{- \frac {( X - X' )^2}{2
    \sigma^2}} \right\|_{{\mt L}_1} \\
= & (III.a) + (III.b) + (III.c) + (III.d) + (III.e).
\end{align*}
To estimate $(III.a)$, we use $ |1 - N^{-2}|  \leq e^{- \f 
{X_0^2}{2 \sigma_H^2}}$, and notice that  
\[
 \left\| e^{2 i p (X -X') - \f 
{(X - X')^2}{2 \sigma^2}} ( \varphi_+ (X) + \varphi_- (X') )
( \overline{ \varphi_+ (X')} + \overline{\varphi_- (X')})
\right\|_{{\mt L}_1} = \f 2 {N^2}.
\]
Indeed, by the identity
\be \label{idagussian}
e^{- \f {(X-X')^2} {2 \sigma^2}} \ = \ \f {\sqrt 2} {\sqrt \pi \sigma}
\int_\RR e^{- \f{(X - \lambda)^2}{\sigma^2}} e^{- \f {(X' - \lambda)^2}{\sigma^2}}
\, d \lambda \ee
one immediately has 
\be \nonumber \begin{split}
 & e^{2 i p (X -X') - \f 
{(X - X')^2}{2 \sigma^2}} ( \varphi_+ (X) + \varphi_- (X) )
( \overline{ \varphi_+ (X')} + \overline{\varphi_- (X')}) \\
 = &  \f {\sqrt 2} {\sqrt \pi \sigma} \int_\RR \left(  e^{2 i p X} e^{-
   \f{(x - \lambda)^2}{\sigma^2}}  (
 \varphi_+ (X) + \varphi_- (X) )\right)   \left(  e^{-2 i p X'} e^{-
   \f{(X' - \lambda)^2}{\sigma^2}}  (
 \overline{\varphi_+ (X')} + \overline{\varphi_- (X')} )\right) \,
 d\lambda .
\end{split} \ee
Therefore the operator to be estimated is positive and its trace norm
can be computed by integrating the integral kernel on the diagonal $X
= X'$, which obtains $\f 2 {N^2}$. Summarizing, one obtains
\be \label{Ia}
(III.a) \ \leq \ e^{- \f {X_0^2}{2  \sigma_H^2}}.
\ee
Let us estimate $(III.b)$. Denoting
$\widetilde \gamma(X,X') : =  e^{- \frac {( X - X' )^2}{2
    \sigma^2}}  \tvarphi_+ (X) \overline{\tvarphi_+ (X')}$, one has
\[
(III.b) =  \f  {| r_p |^2} 2  \left\|  \tvarphi_+(X) 
\overline{\tvarphi_+(X')}  - \widetilde \gamma(X,X') \right\|_{{\mt L}_1}.
\]
Proceeding as was done for $(III.a)$, we see that $\widetilde \gamma$ is
a positive  
 operator with trace one.  To go on, we follow Remark 1.4 
in \cite{schlein}. Setting 
$A(X,X') = \tvarphi_+ (X) \overline{\tvarphi_+ (X')} -\widetilde \gamma(X,X')$, 
we see that $A$ (seen now as an operator) can have only one positive 
eigenvalue, denoted $\lambda_+$ (otherwise there would exist a space of dimension 
two where $A$ is positive, but this is impossible because $\tvarphi_+ (X) 
\overline{\tvarphi_+ (X')}$ is the kernel of a rank one projection). Since $A$ 
has zero trace, it must be $\Tr |A| = 2 \lambda_+$ and 
$\|A \|_{\mt L^1} = 2 \| A\| \le 2 \|A \|_{\mt L ^2}$, where ${\mt L}_2$ 
denotes the Hilbert-Schmidt norm and the norm without index is the usual 
operator norm. This fact allows one to bound $(III.b)$ by
\[
(III.b) \le   {| r_p |^2}   \left\|| \tvarphi_+(X) \overline{
 \tvarphi_+(X')}  - \widetilde \gamma(X,X') \right\|_{{\mt L}_2}
\]
and the  Hilbert-Schmidt norm can easily be computed as the $L^2$-norm
of the corresponding integral kernel, namely 
\begin{align*}
 \left\| | \tvarphi_+(X) \overline{
 \tvarphi_+(X')}  - \widetilde \gamma(X,X') \right\|_{{\mt L}_2}^2
 & =   \int_{\RR^2} \left| \tvarphi_+ (X)  \overline{\tvarphi_+ (X')} 
 \left( 1 - e^{- \frac {( X - X' )^2}{2
     \sigma^2}} \right) \right|^2 \, dX \, dX' \\
 & \leq   \f 1 {2 \pi \sigma_H^2} \int_{\RR^2} e^{- \f{X^2+ (X')^2}{2
     \sigma_H^2}}  \left( 1 - e^{- \f{(X-  X')^2}{2
     \sigma^2}} \right)  \, dX \, dX' \\
 & =  1 - \f 1 \pi  \int_{\RR^2} e^{- (X^2+ (X')^2)} e^{- \f {\sigma_H^2}
 {\sigma^2} (X - X')^2}  \, dX \, dX'.
\end{align*}
In the last line we rescaled variables as $X \to \sqrt 2
\sigma_H X$ and $X' \to \sqrt 2
\sigma_H X'$. Now, observe that 
\[ 
e^{- \f {\sigma_H^2}
{\sigma^2} (X - X')^2} \geq  e^{- \f {2 \sigma_H^2}
{\sigma^2} (X^2 + (X')^2)},
\]
so the integral decouples and
one obtains
\[
 \left\| | \tvarphi_+(X) \overline{
 \tvarphi_+(X')}  - \widetilde \gamma(X,X') \right\|_{{\mt L}_2}^2
 \le 1 - \f 1 \pi \left( \int e^{- \f {\sigma^2 + 2 \sigma_H^2}
  {\sigma^2} X^2} \, dX 
\right)^2  \le   2 \f {\sigma_H^2} {\sigma^2},
\]
and finally
\be \label{Ib}
(III.b) \ \leq \  \sqrt 2 \,|r_p |^2  \f {\sigma_H} \sigma.
\ee
The term $(III.c)$ may be bounded by the same quantity.

Let us focus on $(III.d)$. In order to estimate it, we make use of the identity
\eqref{idagussian} and obtain
\begin{align*}
(I.d) & =  \f {| r_p |^2}{\sigma \sqrt{2 \pi} } \left\| \int_\RR 
  \tvarphi_+ (X) e^{- \f {(X - \lambda)^2} {\sigma^2}}  e^{- \f {(X' -
      \lambda)^2} {\sigma^2}}
  \overline{\tvarphi_- (X')} \, d \lambda \right\|_{{\mt L}_1} \\
& \le  \f {| r_p |^2}{\sigma \sqrt{2 \pi} }   \int_\RR  \left\|
  \tvarphi_+ (X) e^{- \f {(X - \lambda)^2} {\sigma^2}}  e^{- \f {(X' -
      \lambda)^2} {\sigma^2}}
  \overline{\tvarphi_- (X')} \right\|_{{\mt L}_1}
\, d \lambda \\
& =   \f {| r_p |^2}{\sigma \sqrt{2 \pi} }   \int_\RR  \left\|
  \tvarphi_+ (X) e^{- \f {(X - \lambda)^2} {\sigma^2}} \right\|_2 \left\| e^{- 
\f {(X' -   \lambda)^2} {\sigma^2}}
  \overline{\tvarphi_- (X')} \right\|_{2}
\, d \lambda \,.
\end{align*}
By a direct computation,
\[
   \left\|
 e^{2ipX}
  \varphi_\pm (X) e^{- \f {(X - \lambda)^2} {\sigma^2}} \right\|_2^2  
 = \f {\sigma} {\sqrt{4 \sigma_H^2 + \sigma^2}} e^{- \f {(\lambda \mp 
X_0)^2} {2 \sigma_H^2 + \f {\sigma^2}2}}
\]
so that
\[
(III.d)  \leq  \f { |r_p|^2} {\sqrt {2 \pi} \sqrt{4 \sigma_H^2 +
      \sigma^2}}
\int_\RR   e^{- \f {(\lambda - X_0)^2}
   {4 \sigma_H^2 + \sigma^2 } }  e^{- \f {(\lambda + X_0)^2}
   {4 \sigma_H^2 +  \sigma^2 }} \, d \lambda
=   \f { |r_p|^2} {2 \sqrt 2 }  e^{- \f {X_0^2}
   {2 \sigma_H^2 + \f {\sigma^2} 2}}.
\]
Term $(III.e)$ can be estimated in the same way.
Finally, putting everything together, we get the requested bound.

\section{Appendix}

\subsection{The Dirac's delta potential in dimension one.}

Assume that the potential is a Dirac's delta
 with strength $\alpha$, i.e. $V = \alpha
\delta_0$.

The operator $- \frac12 \Delta_x + \alpha \delta_0:   D(H_\alpha) \subset L^2(\RR)
\rightarrow L^2(\RR)$ is defined on the domain
\be \label{Hdelta}
D(H_\alpha) := \{ \psi \in H^2(\RR^-) \cap H^2(\RR^+) \text{ s.t.\  }  \psi(0^+)
=\psi(0^-)  \text{ and }  \psi'(0^+) - \psi'(0^-) =  2 \alpha
\psi(0^+) \}
\ee
by the action
\[
 \left( - \frac12 \Delta_x + \alpha \delta_0 \right) \psi (x) =
-  \frac12 \psi'' (x), \qquad x \neq 0,
\]
and is a self-adjoint operator in $L^2(\RR)$. 

The propagator 
\[
U_\alpha(t) : = \exp \left[ -i t
\left( - \frac12 \Delta_x + \alpha \delta_0 \right) \right]
\]
is 
explicitly known (\cite{schulman,ABD95}). In order to express it,  we
shall use the following operators:
\bean
\text{the symmetry operator } &&  \mt R \chi := \frac12 [\chi +
\chi(-\cdot)] \\
\text{the projection on positive  positions} &&  \mt P_x^+ \chi :=
\indic_{\RR^+} \, \chi \\
\text{the projection on positive momenta} &&  \mt P_k^+ \chi  := \mt F^{-1}
\left(\indic_{\RR^+} \, \widehat \chi \right)\\
\text{the translation by } u  &&  \theta_u \chi  := \chi (\cdot - u)
\eean
where $\mt F^{-1}$ denotes the inverse Fourier transform. The projection on
negative positions $\mt P_x^-$ and on negative momenta $\mt P_k^-$ are
defined in a similar way. Remark that all these operators on
$L^2(\RR)$ have norm 
equal to $1$, and that $\mt R$, $\mt P_k^+$ and $\mt P_k^-$ and
$\theta_u$ commute 
with the free evolution group $U_0(\tau)$.

\begin{prop} \label{prop:PropOp}
Given $\alpha > 0$, for any $t \in \RR \backslash \{0 \}$ the propagator
$U_\alpha (t)$ 
can be expressed as 
\be \label{eq:PropFreeOp}
U_\alpha(t) = U_0(t) - 4 \alpha \, \mt R \, \mt P_x^+ \left( \int_0^{+\infty} 
e^{-\alpha u} \theta_{-u} \, du \right)\, U_0(t) \, \mt P_x^- \, \mt
R.
\ee
\end{prop}
\begin{proof}[Proof of Proposition~\ref{prop:PropOp}]
From \cite{schulman, ABD95} we know that $U_\alpha (t)$ is the
integral operator defined by the kernel
\be \label{def:propDirac}
U_\alpha (t,x,x') 
:= U_0(t,x-x') - \alpha \int
_0^\infty e^{-\alpha u} U_0(t,u + |x| + |x'|) \,  du,
\ee
where $U_0$ denotes the propagation kernel of the free \sch equation
\be \label{def:propfree}
U_0(t,y) := \frac1{\sqrt{2i \pi t}} e^{\frac i {2t} y^2}.
\ee
In order to obtain \eqref{eq:PropFreeOp} it suffices to notice
\bean
\int_\RR U_0(t,u + |x| + |x'|) \chi(x') \,dx 
&=& \int_\RR  U_0(t,u + |x| + |x'|) \, \mt R \,
\chi(x') \,dx'  \\
&=&2  \int_\RR U_0(t,u + |x| - x') \,
\indic_{\RR^-}(x') \, \mt R \, \chi(x') \,dx'  \\
&=&4   \mt R \, \indic_{\RR^+}(x) \int_\RR U_0(t,u +
x - x') \, \mt P_x^- \, \mt R \, \chi(x') \,dx' \\
&=&4   \mt R \, \mt P_x^+   \theta_{-u}
\, U_0(t) \, \mt P_x^- \, \mt R \, \chi (x).
\eean
and the proof is complete.
\end{proof}

The following proposition gives the convergence rate to the scattering
operator needed in order to obtain Corollary~\ref{cor:Dirac1D} from
Theorem~\ref{thm:geneS} and~\ref{thm:geneRho}.

\begin{prop}  \label{prop:ScatDirac}
If $V = \alpha \delta_0$, with $\alpha > 0$, then 
the scattering operator $S_\alpha$ for the
Hamiltonian
$-\f 1 2 \Delta + V$ is well defined and
given by
\be \label{eq:OpScatDelta}
S_\alpha = Id - 4 \,\alpha \,   \mt R  \, \mt P_k^+ \left( \int_0^{+\infty} 
e^{-\alpha u} \theta_{-u} \, du \right)  \mt R.
\ee
The associated scattering matrix is defined as usual
by the following transmission
and reflection amplitudes
\be 
r_k = - \frac{\alpha}{\alpha - i |k|}, \qquad t_k =
- \frac{i|k|}{\alpha - i
|k|}.
\ee
Moreover, there exists a constant $C_2$
such that, for any $\chi \in L^2 (\RR)$  satisfying $ \langle \cdot
\rangle^2 \chi \in L^2 (\RR)$, one has
\be \label{eq:RateScatDirac}
\left\| [S_\alpha(\tau,\tau') - S_\alpha ] \chi\right\|_2 \leq  C_2 \left( 3 \, \|
\la x \ra^2
\chi\|_2 + \frac2{\alpha^2} \right)
\min(\tau,\tau')^{-\frac14}.
\ee
\end{prop}

The key ingredient of the proof is the stationary phase
estimate~\eqref{estim:dispfree2}, see Lemma~\ref{lem:dispfree}, which
is proven  in the next section. 

\begin{proof}[Proof of Proposition~\ref{prop:ScatDirac}] \

{\sl Step 1. Identifying the limit. }
We need to study the limit of $U_0(-\tau) \, U'_\alpha(\tau + \tau') \,
U_0(-\tau')$ as $\tau, \tau' \to + \infty$. Using 
formula~\eqref{eq:PropFreeOp} and the fact that $\mt R$ and $U_0$
commute, one gets
\begin{align*}
U_0(-\tau) \, & U'_\alpha(\tau + \tau')   \, U_0(-\tau') =  \\ 
& Id   - 4 \alpha \, \mt R \, 
U_0(-\tau) \, \mt P_x^+  \left( \int_0^{+\infty}  e^{-\alpha u} \theta_{-u} \,
du \right)\, U_0(\tau + \tau') \, \mt P_x^- \, U_0(-\tau') \, \mt
R.
\end{align*}
Roughly speaking, for large negative times the component of the wave
function
lying in the negative half-line approximately coincides
with the component of the wave function  that travels with positive
speed.
It seems then natural to replace, for large $\tau'$, $\mt P_x^-
U_0(-\tau')$ by $ U_0(-\tau') \mt P_k^+$.
Using the fact that $U_0(t)$, $\theta_u$ and $\mt P_k^+$ commute 
with one another at any $t$, one obtains
\begin{align}
 U_0(-\tau)  \, &U'_\alpha(\tau + \tau') \, U_0(-\tau')  =  \nonumber  \\ 
 & Id - 4 \alpha \, \mt R \,
U_0(-\tau) \, \mt P_x^+ \, U_0(\tau) \, \mt P_k^+ \left( \int_0^{+\infty}
 e^{-\alpha u} \theta_{-u} \, du \right)  \mt R + \eta_1 (\tau, \tau'),  \nonumber \\
 & \text{with} \quad \| \eta_1(\tau, \tau')\chi \|_2 \leq
4  \bigl\|  [\mt P_x^- \, U_0(-\tau') - U_0(-\tau') \mt P_k^+ ] \,
\mt R \chi \bigr\|_2   \label{eq:scatapprox1}.
\end{align}
In the last estimate we used the fact that all concerned operators 
have norm one, and a factor $\alpha^{-1}$ comes by the
integral in $u$. The next step consists in erasing the operator $\mt
P_x^+$ in the 
r.h.s of \eqref{eq:scatapprox1}. Indeed, it acts after
the operator $\mt P_k^+$, and therefore everything should move to
the left for positive times anyway. We obtain
\begin{align}
 U_0(-\tau)  \, U_\alpha(\tau + \tau') \, U_0(-\tau')  &=    Id -   
 S'_\alpha + \mt \eta_2 (\tau, \tau') + \mt \eta_1  (\tau, \tau') \nonumber \\
 \text{with } \quad S'_\alpha &:= 4 \, \alpha \, \mt R  \, \mt P_{k}^+ \left(
 \int_0^{+\infty} e^{-\alpha u} \theta_{-u} \, du \right)  \mt R
 \label{def:S1} \\
 \text{and, defined} \;  \chi_\alpha :=  \alpha \int_0^{+\infty}  e^{-\alpha u} \theta_{-u} 
\mt R \chi \, du  , & \quad    \| \mt \eta_2  (\tau, \tau') \chi \|_2 \leq    4
 \, \bigl\|
\mt P_x^- \, U_0(\tau) \, \mt P_k^+ \, \chi_\alpha
\bigr\|_2,
\label{eq:scatapprox2}
\end{align}
where in the last line, we used $\mt P_x^+ + \mt P_x^- = Id$. The
operator $S_\alpha = Id - \alpha S_\alpha'$ can be explicitly written
in Fourier variables. Indeed, for all
$\chi \in L^2$, using that $\mt R$ commutes with the  Fourier
transform
$\mt F$, one has
\bean
\widehat{S'_\alpha \chi} (k) &= & \mt F \left[ 4 \,  \mt R  \, \mt P_k^+
\left( \alpha  \int_0^{+\infty}  e^{-\alpha u} \theta_{-u} \, du \right)  \mt R  \chi 
\right] (k) \\
& = & 2\,  \mt F \left[  \alpha  \left( \int_0^{+\infty}  e^{-\alpha
    u} \theta_{-u} \, du 
\right)  \mt R \chi \right] (|k|) 
= 2    \alpha \int_0^{+\infty}  e^{-\alpha u} \mt F \left[
  \theta_{-u}    \mt R \chi 
\right] (|k|)  \, du \\
& = &  2  \alpha   \int_0^{+\infty}  e^{-(\alpha - i |k|) u} \mt F \left[  \mt R \chi
\right] (|k|)  \, du 
= \frac{2 \alpha}{\alpha - i |k|} \mt R \widehat \chi(|k|) \\
&=& \alpha \frac{\widehat \chi(k) + \widehat \chi(-k)}{\alpha - i |k|}.
\eean 
Owing to~\eqref{def:propDirac}, the scattering operator $S_\alpha$ 
is given by
\bean
\widehat{S_\alpha \chi} (k)  &:=& \widehat \chi (k) -  \frac\alpha{\alpha - i
|k|} (\widehat \chi(k) + \widehat \chi(-k)) \\
&=&  \frac{-i|k|}{\alpha - i |k|} \widehat \chi(k)   - \frac\alpha{\alpha - i
|k|} \widehat \chi(-k),
\eean
which, in view of ~\eqref{sfourier}, provides the transmission and
reflection amplitudes. 

\medskip
{\sl Step 2. Control of the error terms $ \eta_1 (\tau, \tau') \chi$
  and $ \eta_2 (\tau, \tau') \chi$.}
For the control of $\eta _1  (\tau, \tau') \chi$ one first observes 
\bean
\mt P_x^- \, U_0(-\tau') - U_0(-\tau') \mt P_k^+ &=&
\mt P_x^- \, U_0(-\tau') [\mt P_k^- + \mt P_k^+] -
[\mt P_x^- + \mt P_x^+]U_0(-\tau') \mt P_k^+ \\
& =& \mt P_x^- \, U_0(-\tau') \, \mt P_k^- -
 \mt P_x^+ \, U_0(-\tau') \mt P_k^+,
\eean
so that a bound on $\| \eta_1  (\tau, \tau') \chi \|_2$ follows from two applications of
Lemma~\ref{lem:dispfree}, to be proven in the next section. 
Thus, for any $n \ge 2$
\[
\| \mt \eta_1  (\tau, \tau') \chi \|_2 \leq 
{2 C_2} \, \| \la x \ra^2 \chi \|_2 \,  (\tau')^{-\frac14}.
\]
Next, a bound on $\| \eta_2(\tau, \tau') \chi \|_2$ follows by 
Lemma~\ref{lem:dispfree}, with $\chi_\alpha$ as initial data, and by
noticing that the moments
of $\chi_\alpha$ are related to those of $\chi$. Precisely, 
\bean
\| \la x \ra^2 \chi_\alpha \|_2 &:=&  \alpha \left\|  \int_0^{+\infty} e^{-\alpha u }
\la 
x \ra^2 \theta_{-u} \mt R  \chi \,du \right\|_2 
\le  \alpha \int_0^{+\infty} e^{-\alpha u }\left\| \la x - u \ra^2  \mt R \chi 
\right\|_2 \,du
\\
& \leq & 2 \alpha \int_0^{+\infty} e^{-\alpha u } \bigl( \left\| \la x
\ra^2
\chi \right\|_2 + \la u \ra^2 \| \chi \|_2 \bigr) \,du \\
& \le & 2 \left\| \la x \ra^{n+1} \chi
\right\|_2  + \frac{2}{\alpha^2}.
\eean
Therefore,
\[
\| \eta_2\chi \|_2 \leq C_2 \, \| \la x \ra^{n+1}
\chi_\alpha \|_2 \,  \tau^{-\frac14} \leq  C_2
\left[  \left\| \la x \ra^2 \chi
\right\|_2  + \frac{2}{\alpha^2} \right] \tau^{-\frac14}.
\]
\end{proof}

\subsection{A stationary phase estimate}

Here we give a stationary phase lemma. It is crucial in order to
prove the convergence of the scattering operator for the Dirac's delta potential in
dimension one, as stated in Proposition~\ref{prop:ScatDirac}.

\begin{lem} \label{lem:dispfree}
There exists a constant $C_2$ such that the following
estimate holds
\be \label{estim:dispfree2}
 \forall \tau \in \RR^+, \qquad \| \mt P_x^- U(\tau) \mt P_k^+ \chi 
\|_2 \leq  C_2 \, \| \la x \ra^2
\chi \|_2 \,  \tau^{-\frac14}.
 \ee
The same estimates are also valid for $\mt P_x^+ U(\tau) \mt P_k^-$, 
$\mt P_x^+ U(-\tau) \mt P_k^+$ and $\mt P_x^- U(-\tau) \mt P_k^-$, always
with positive $\tau$.
\end{lem}

\begin{proof}[Proof of Lemma~\ref{lem:dispfree}]
We follow the classical argument used to obtain stationary phase estimates.
The first step consists in separating low frequencies from high ones in $\chi$. We choose
a smooth function $g : \RR \to [0,1]$ such that $g=1$ on $(-\infty,1]$, $g=0$ on
$[2, + \infty)$. We introduce a scale $\eta < 1$ to be fixed more
    precisely later, and the
associated function $g_\eta(k) := g \left( \frac {k} \eta
\right)$. We shall use
the decomposition 
\be \label{eq:lowhighdecomp}
\chi = \chi_l + \chi_h, \quad \text{with}\;  \widehat \chi_l =  \widehat \chi \,
g_\eta, \quad \widehat \chi_h = \widehat \chi \, (1 - g_\eta).
\ee
The contribution of $\chi_l$ is bounded by
\begin{align} \label{eq:contribchi_l}
\| \mt P_x^- U(\tau) \mt P_k^+  \chi_l  \|_2 &\leq \| \mt P_k^+  \chi_l \|_2  =
\| \indic_{\RR^+} \widehat \chi_l \|_2  \le \sqrt {2 \eta} \, \| \widehat \chi_l \|_\infty
\nonumber \\
& \leq     \sqrt{ 2 \eta \, \| \widehat \chi_l \|_2 \| \partial_k \widehat \chi_l
\|_2}   \le  C \, \| \la x \ra \chi \|_2 \, \sqrt \eta,
\end{align}
where we have used a Gagliardo-Nirenberg-Sobolev inequality.

\noindent
The contribution of $\chi_h$ can be controlled by using stationary phase methods. In fact,
denoting, for some fixed $\tau$, $\chi^\ast_h = \mt P_x^- U(\tau) \mt
P_k^+  \chi_h$, for any $x < 0$ we get 
\be \label{eq:chih}
\chi^\ast_h(x)  =  \int e^{ikx} \widehat{\chi^\ast_h}(k)  \,\frac
    {dk}{\sqrt {2 \pi}} 
 =   \int e^{ikx} \widehat{U(\tau)  \mt P_k^+ \chi_h}(k)  \, \frac{dk} {\sqrt {2 \pi}} 
 =   \int_\eta^{+\infty} e^{-i \tau \left(\frac{k^2}2 - \frac{kx}\tau \right)}
\widehat{  \mt P_k^+ \chi_h}(k) \,  \frac{dk} {\sqrt {2 \pi}}.
\ee
Introducing the differential operator $\square_y$ defined by
\be
[\square_y h]  (k)= \frac d {dk} \left( \frac{h(k)}{k - y }\right),
\ee
and integrating twice by parts, we obtain
\begin{eqnarray} \nonumber
\chi^\ast_h( \tau y ) & = & \int e^{-i \tau \left(\frac{k^2}2 - k y \right)} \widehat{\chi_h}(k) \,dk =
- \frac i \tau
\int e^{-i \tau \left(\frac{k^2}2 - k y \right)} \square_y  \widehat{\chi_h}
 (k) \,dk \\
& = &  - \frac 1 {\tau^2} \label{eq:intpart}
   \int_\eta^{+\infty} e^{-i \tau \left(\frac{k^2}2 - k y \right)}
\square_y^2 \widehat{ \chi_h}(k) \,dk. \label{eq:chih2}
\end{eqnarray}
The quantity $\square_y^2 \widehat{ \chi_h}$ may be rewritten as the
following sum
\[
\square_x^2 \widehat{ \chi_h}(k) =
\sum_{n_1+n_2+n_3=2} c_{n_1,n_2,n_3} \frac 1{\eta^{n_2}}
\frac{\partial_k^{n_1} \widehat \chi(k)}{( k - y )^{2+n_3}}
\partial_k^{n_2} \left[1 - g \left( \frac k \eta\right),
  \right]
\]
where all amplitudes $c_{n_1,n_2,n_3}$ are bounded (in absolute value) by $3$. 
Using this sum in equation~\eqref{eq:chih2}, we can perform some integration on $k$ and get
\[
| \chi^\ast_h( \tau y )   |\le 
\frac {\| \partial_k^2 \widehat \chi\|_2}{\tau^2 (\eta-y)^{3/2}}+ 
\frac{C_2}{\tau^2} \sum_{n_1+n_2+n_3=2, n_1 \neq 2}  \frac {\|
  \partial_k^{n_1} \widehat
  \chi\|_\infty}{\eta^{n_2}(\eta-y)^{2+n_3-1}}.   
\]
The first term in the r.h.s. comes from the term with $n_1=2$, for
which we used 
Cauchy-Schwarz inequality. The constant $C_2$ depends on $ \|
\partial^i g \|_ \infty$ for $i=1,2$. 
Remark that for $n_1=0,1$, we may always bound $\| \partial_k^{n_1} \widehat \chi\|_\infty$ by $\| \partial_k^2 \widehat \chi\|_2$ thanks to the Gagliardo-Nirenberg-Sobolev inequality $\| \zeta
\|_\infty^2 \le  \| \partial_k \zeta \|_2 \| \zeta \|_2$. In view of this and
since $y<0$, the worst term in the sum of the r.h.s. is the one
obtained for  $n_2 =2$. This leads  to the bound
\[
| \chi^\ast_h( \tau y )   |\le  
\frac { C_2 \| \partial_k^2 \widehat \chi\|_2}{\tau^2 } \biggl( \frac1{(\eta-y)^{3/2}} + \frac1{\eta^2 (\eta-y)} \biggr).
\]
 Taking the square and integrating with respect to $x= \tau y$, we obtain
 \be \label{eq:contrib_h}
 \| \chi^\ast_h \|_2 \le  \frac { C_2 \| \partial_k^2 \widehat \chi\|_2}{\tau^{3/2} } 
 \biggl( \frac1\eta + \frac1{\eta^{5/2}} \biggr) \le  \frac { C_2 \| \partial_k^2 \widehat \chi\|_2}{\tau^{3/2} \eta^{5/2}},
 \ee
 when $\eta \le 1$. 
 Adding~\eqref{eq:contribchi_l} and~\eqref{eq:contrib_h}, we finally obtain
\[
\| \mt P_x^- U(\tau) \mt P_k^+  \chi \|_2 \le 
  C_2  \| \la x \ra^2 \chi \|_2  \Bigl(  \sqrt \eta + \tau^{-3/2} \eta^{-5/2}
\Bigr).
\]
The optimal choice for $\eta$ is then $\eta = \tau^{-1/2}$ which
leads to
\[
\| \mt P_x^- U(\tau) \mt P_k^+  \chi \|_2 \le C_2 \| \la x \ra^2 \chi \|_2
\tau^{- \frac14}.
\]
\end{proof}


\noindent {\bf Acknowledgments.} 
The authors would like to thank the anonymous referees, for their very careful reading of our manuscript. 

The authors would like to acknowledge
support from the ANR LODIQUAS  (Modeling and Numerical Simulation of
Low Dimensional Quantum Systems, 2011-2014) and the French-Italian
research project GREFI-MEFI. R.A. was partially supported by the PRIN
project ``Critical point theory and perturbative methods for nonlinear
differential equations''.

\end{document}